\pdfoutput=1
\documentclass[11pt]{article} 
\usepackage{natbib}

\usepackage[letterpaper,margin=1in]{geometry}
\usepackage{enumerate}
\usepackage[T1]{fontenc}
\usepackage{microtype}
\usepackage{titlesec}
\usepackage{booktabs}
\usepackage{makecell}
\titleformat*{\paragraph}{\bfseries}

\usepackage{turnstile}

\usepackage{amsthm,amsmath,amssymb,amsfonts,amssymb,mathtools}
\usepackage{amsmath,amssymb,amsfonts,amssymb,mathtools}
\usepackage{xcolor}
\usepackage{empheq}
\usepackage{dsfont}
\usepackage{mathrsfs}
\usepackage{graphicx}
\usepackage{aliascnt} 
\usepackage{xspace}
\usepackage{bbm}
\usepackage{caption}
\usepackage{tikz}
\usetikzlibrary{patterns}
\usetikzlibrary{patterns}
\usetikzlibrary{arrows,shapes,automata,backgrounds,petri,positioning}
\usetikzlibrary{shadows}
\usetikzlibrary{calc}
\usetikzlibrary{spy}
\usetikzlibrary{angles, quotes}
\usetikzlibrary{matrix}
\usepackage{pgf,pgfplots}
\usepackage{pgfmath,pgffor}
\pgfplotsset{compat=1.17}

\usepackage[shortlabels]{enumitem}
\usepackage{framed}
\usepackage[noend]{algorithm2e}
\usepackage[noend]{algpseudocode}

\usepackage{thm-restate}

\definecolor[named]{ACMBlue}{cmyk}{1,0.1,0,0.1}
\definecolor[named]{ACMYellow}{cmyk}{0,0.16,1,0}
\definecolor[named]{ACMOrange}{cmyk}{0,0.42,1,0.01}
\definecolor[named]{ACMRed}{cmyk}{0,0.90,0.86,0}
\definecolor[named]{ACMLightBlue}{cmyk}{0.49,0.01,0,0}
\definecolor[named]{ACMGreen}{cmyk}{0.20,0,1,0.19}
\definecolor[named]{ACMPurple}{cmyk}{0.55,1,0,0.15}
\definecolor[named]{ACMDarkBlue}{cmyk}{1,0.58,0,0.21}

\usepackage[colorlinks,citecolor=blue,linkcolor=magenta,bookmarks=true]{hyperref}
\usepackage[nameinlink]{cleveref}
\crefname{ineq}{Inequality}{Inequality}
\creflabelformat{ineq}{#2{\upshape(#1)}#3}
\crefname{sub}{Subsection}{Subsection}
\creflabelformat{Subsection}{#2{\upshape(#1)}#3}
\crefname{sdp}{SDP}{SDP}
\creflabelformat{sdp}{#2{\upshape(#1)}#3}
\crefname{lp}{LP}{LP}
\creflabelformat{lp}{#2{\upshape(#1)}#3}
\usepackage{aliascnt}
\usepackage{cleveref}
\crefname{ineq}{Inequality}{Inequality}
\creflabelformat{ineq}{#2{\upshape(#1)}#3}
\crefname{sub}{Subsection}{Subsection}
\creflabelformat{Subsection}{#2{\upshape(#1)}#3}
\crefname{sdp}{SDP}{SDP}
\creflabelformat{sdp}{#2{\upshape(#1)}#3}
\crefname{lp}{LP}{LP}
\creflabelformat{lp}{#2{\upshape(#1)}#3}

\newtheorem{theorem}{Theorem}[section]

\newtheorem{lemma}[theorem]{Lemma}
\newtheorem{informal theorem}[theorem]{Theorem (informal statement)}

\newtheorem{problem}[theorem]{Problem}
\newtheorem{proposition}[theorem]{Proposition}
\newtheorem{corollary}[theorem]{Corollary}
\newtheorem{claim}[theorem]{Claim}
\newtheorem{fact}[theorem]{Fact}

\newtheorem{definition}[theorem]{Definition}

\usepackage{times}

\newcommand{\eqdef}{\coloneqq}

\newcommand{\lp}{\left}
\newcommand{\rp}{\right}
\newcommand\norm[1]{\left\| #1 \right\|}

\renewcommand\vec[1]{\mathbf{#1}}
\DeclareMathOperator*{\pr}{\mathbf{Pr}}
\DeclareMathOperator*{\E}{\mathbf{E}}
\DeclareMathOperator*{\Var}{\mathbf{Var}}

\def\d{\mathrm{d}}
\newcommand{\normal}{\mathcal{N}}

\newcommand{\bx}{\mathbf{x}}
\newcommand{\by}{\mathbf{y}}
\newcommand{\bv}{\mathbf{v}}

\newcommand{\R}{\mathbb{R}}

\newcommand{\Z}{\mathbb{Z}}
\newcommand{\N}{\mathbb{N}}
\newcommand{\eps}{\epsilon}

\newcommand{\poly}{\mathrm{poly}}

\newcommand{\sgn}{\mathrm{sign}}
\newcommand{\sign}{\mathrm{sign}}

\newcommand{\Ind}{\mathds{1}}
\newcommand{\1}{\Ind}

\newcommand{\x}{\vec x}

\newcommand{\normald}[1]{\mathcal{N}_{#1}}

\newcommand{\iid}{{i.i.d.}\ }
\newcommand{\abs}[1]{\lp| #1 \rp|}
\renewcommand\Pr{\pr}

\DeclarePairedDelimiter\floor{\lfloor}{\rfloor}

\def\nnewcolor{0}
\ifnum\nnewcolor=1
\newcommand{\inote}[1]{\footnote{{\bf [[Ilias: {#1}\bf ]] }}}
\newcommand{\dnote}[1]{\footnote{{\bf [[Daniel: {#1}\bf ]] }}}
\newcommand{\tnote}[1]{\footnote{{\bf [[Thanasis: {#1}\bf ]] }}}
\newcommand{\snote}[1]{\footnote{\textcolor{orange}{sihan: #1}}}

\else

\newcommand{\inote}[1]{}
\newcommand{\dnote}[1]{}
\newcommand{\nnote}[1]{}
\newcommand{\vnote}[1]{}
\newcommand{\snote}[1]{}
\newcommand{\tnote}[1]{}
\fi

\newcommand{\ngca}{\mathcal M_{A,\vec v}}
\newcommand{\ngcav}{\mathcal M_{A,\vec v}}

\newcommand{\multix}{
\vec x^{(1)}, \ldots, \vec x^{(n)}}
\newcommand{\multiy}{
\vec y^{(1)}, \ldots, \vec y^{(n)}}

\newcommand{\bmultix}{
\vec x^{(1:n)}}
\newcommand{\bmultiy}{
\vec y^{(1:n)}}

\newcommand{\interpolate}[1]{\vec x^{(1)}, \cdots, \bar {\vec y} + #1 \vec v, \cdots, \vec y^{(n)}
}
\newcommand{\origin}{\vec x^{(1)}, \cdots, \bar {\vec y} + \xi^* \vec v, \cdots, \vec y^{(n)}
}

\newcommand{\cg}{c_{g}}
\newcommand{\ct}{c_{\text{trunc}}}

\newcommand{\cE}{\mathcal{E}}

\newcommand{\lnormal}{\normal\lp(\vec 0, \vec I - \vec v \vec v^\top\rp)}

\newcommand{\usphere}{ \mathcal U\lp(  \mathcal{S}^{d-1} \rp) }
\newcommand{\alt}{\mathrm{alt}}

\renewcommand{\epsilon}{\varepsilon}
\renewcommand{\hat}{\widehat}

\newcommand{\baryi}{\bar {\vec y}^{(i)}}

\newcommand{\Id}{ {\vec I} }
\newcommand{\cB}{ \mathcal B }
\newcommand{\cN}{ \mathcal N }
\newcommand{\cD}{ \mathcal D }
\newcommand{\huber}{ \text{huber} }
\newcommand{\bmu}{\boldsymbol{\mu}}

\title{PTF Testing Lower Bounds for Non-Gaussian Component Analysis}

\author{
Ilias Diakonikolas\thanks{Supported by NSF Medium Award CCF-2107079, and an H.I. Romnes Faculty Fellowship.}\\
University of Wisconsin-Madison\\
{\tt ilias@cs.wisc.edu}\\
\and
Daniel M. Kane\thanks{Supported by NSF Medium Award CCF-2107547 and NSF Award CCF-1553288 (CAREER).}\\
University of California, San Diego\\
{\tt dakane@cs.ucsd.edu}\\
\and
Sihan Liu\thanks{Supported by NSF Medium Award CCF-2107547 and NSF Award CCF-1553288 (CAREER).}\\
University of California, San Diego\\
{\tt sil046@ucsd.edu}
\and
Thanasis Pittas\thanks{Supported by NSF Medium Award CCF-2107079.}\\
University of Wisconsin-Madison\\
{\tt pittas@wisc.edu}\\
}

\begin{document}

\maketitle

\begin{abstract}%
This work studies information-computation gaps for statistical problems. 
A common approach for providing evidence of 
such gaps is to show sample complexity lower bounds 
(that are stronger than the information-theoretic optimum)  
against natural models of computation.
A popular such model in the literature is 
the family of \emph{low-degree polynomial tests}. 
While these tests are defined in such a way that 
make them easy to analyze, 
the class of algorithms that they rule out is somewhat restricted.
An important goal in this context has been to obtain lower bounds 
against the stronger and more natural class 
of low-degree \emph{Polynomial Threshold Function (PTF) tests}, i.e., 
any test that can be expressed as comparing 
some low-degree polynomial of the data to a threshold. 
Proving lower bounds against PTF 
tests has turned out to be challenging. 
Indeed, we are not aware of any non-trivial PTF testing  
lower bounds in the literature. 

In this paper, we establish the first non-trivial 
PTF testing lower bounds for a range of statistical tasks. 
Specifically, we prove a near-optimal PTF testing lower bound 
for Non-Gaussian Component Analysis (NGCA). 
Our NGCA lower bound implies similar 
lower bounds for a number of other statistical problems. 
Our proof leverages a connection to 
recent work on pseudorandom generators for PTFs 
and recent techniques developed in that context. 
At the technical level, we develop several tools of independent interest,  
including novel structural results for 
analyzing the behavior of low-degree polynomials 
restricted to random directions. 
\end{abstract}

\thispagestyle{empty}

\newpage
\setcounter{page}{1}
\section{Introduction}

In classical statistical estimation, the focus has primarily 
been on determining the minimum amount of information 
required to estimate the parameters of an 
unknown distribution to a desired level of accuracy. 
Classical statistical theory provides general methodology 
to characterize this quantity for a range of 
estimation and inference tasks. When 
taking computational aspects into account, the situation 
becomes more subtle. Statistically optimal 
estimators often entail an exhaustive search. 
On the other hand, known computationally efficient estimators 
often require more data than is necessary. 
A fundamental question 
is whether these observed gaps are inherent. 
An information-computation gap describes 
a scenario where no computationally efficient method  
can achieve the information-theoretic limits.

A key question is how to formally establish the existence of 
an information-computation gap for a particular problem.
Traditional methods from complexity theory, such as NP-hardness, seem 
inadequate for this purpose; see, e.g.,~\cite{ABX08}. 
Over the past decade, a line of work in theoretical computer 
science has made progress in our understanding of this broad 
question. A prominent approach to establishing information-computation gaps involves showing unconditional lower bounds 
within natural (yet restricted) 
computational models---such as Statistical Query (SQ) 
algorithms~\cite{Kearns93,FGR+13}, 
low-degree polynomials 
(LDP)~\cite{hopkins2018statistical,kunisky2019notes}, 
and Sum-of-Squares (SoS) algorithms 
(see, e.g.,~\cite{BarakSteurerNotes}). 
These methodologies have provided rigorous evidence of information-computation tradeoffs 
for a range of fundamental and well-studied statistical tasks.

In this paper, we consider the class of algorithms based on low-degree Polynomial Threshold Functions (PTFs).
A degree-$k$ 
PTF $f: \R^N \to \{ 0, 1\}$ is a function of the form $f(\vec x) = \sgn(p(\vec x))$, where $p: \R^N \to \R$ is a polynomial of degree at most $k$ and $\sign(u)$ denotes the function which is equal to $1$ 
whenever $u \geq 0$ and $0$ otherwise. 
PTFs is a natural class of Boolean functions that has 
been extensively studied in complexity theory and machine learning 
over the past six decades; 
see, e.g.,~\cite{Rosenblatt:58, Chow:61, MTT:61, Dertouzos:65, MinskyPapert:88} for some early work 
and~\cite{DGJ+:10, DHK+:10, DKN10, meka2010pseudorandom, Kane11focs, Kane14b, DeS14, DeDFS14, DRST14, DKS18-nasty,DK19-chow,o2020fooling,kelley2022random,DKKL24}.\looseness=-1

We will use the term ``PTF tests'' for the associated class of 
algorithms. As we will explain below, PTF tests are strictly stronger 
than LDP tests. Perhaps surprisingly, prior to this work, no non-trivial information-computation gaps were known against this class.

\vspace{-0.3cm}

\paragraph{Background}
Before defining the family of PTF tests, we provide some background. We focus on hypothesis testing problems, as lower bounds for more complex tasks (like learning or parameter estimation) often stem from testing lower bounds. In particular, the null hypothesis is a single distribution $D_{\emptyset}$, and the alternative is sampled from a prior $\mu$ on a family $\mathcal D_{\alt}$ of distributions ($\mu, D_{\empty}$, and $\mathcal D_{\alt}$ are known to the testing algorithm).\looseness=-1

\begin{problem}[Hypothesis Testing]\label{def:testing_problem}
    We are given $n$ samples in $\R^d$ generated in one of two ways:
    \begin{itemize}[leftmargin=1.5em, itemsep=0pt, topsep=3pt, partopsep=0em]
        \item (Null Hypothesis) The samples are drawn \iid from a known distribution $D_{\emptyset}$.
        \item (Alternative Hypothesis) A member $D_{\alt}$ is sampled according to a known prior distribution $\mu$ on a family of alternative distributions $\cD_{\alt}$, and then the $n$ samples are drawn \iid from $D_{\alt}$.
    \end{itemize}
    Given the samples, the goal is to distinguish between the two cases with high constant probability.
\end{problem}

A natural class of tests is based on PTFs. 
As we review below, the well-studied class of Low-Degree Polynomial Tests consists of PTFs but they take a specifc form.

\vspace{-0.4cm}

\paragraph{Low-Degree Polynomial (LDP) Tests}
We now define the family of LDP tests and discuss existing testing lower bounds in this model.
Informally, the family contains tests of the following form: For a polynomial $p$ that satisfies 
a separation in terms of its expected values under the null and the alternative distributions, the test is the thresholded version of $p$ at the midpoint of these expected values.
These tests are usually parameterized by two numbers: the maximum degree $k$ of the polynomial, which quantifies the runtime of evaluating the test, and, the number $n$ of samples used. 
In slightly more formal language,
given a polynomial $p$, 
a null distribution $D_{\emptyset}$ and an alternative distribution family $\cD_{\alt}$ that we aim to distinguish, 
the ``advantage'' $\gamma$ of the polynomial is defined 
as the difference in expected values of $p$ under $D_{\emptyset}$ versus under a random distribution $D_{\alt}$ from $\cD_{\alt}$, relative to the variance of the polynomial. We will use the notation $\bmultix$ as a shorthand notation for $\multix$ throughout the paper.
\begin{restatable}[$\gamma$-advantageous polynomial]{definition}{GOODPOLYNOMIAL}\label{def:good-ld}
Let $\gamma > 0$, $p: \R^{ n \times d } \mapsto \R$ be a degree-$k$, $n$-sample polynomial. 

Let $D_{\emptyset}$ be a distribution in $\R^d$, $\cD_{\alt}$ be a family of distributions in $\R^d$, $\mu$ be a distribution on $\cD_{\alt}$, and $H$ be the hypothesis testing of distinguishing between $D_{\emptyset}$ and $\cD_{\alt}$ with prior $\mu$. 
We say that $p$ is a degree-$k$, $n$-sample, $\gamma$-advantageous polynomial with respect to the testing problem $H$ if:\footnote{Many works focusing on lower bounds,  use only the variance under $D_{\emptyset}$ in the RHS of the condition in \Cref{def:good-ld}. }
\vspace{-0.15cm}
\begin{align*}
\bigg|
\E_{\bmultix \sim D_{\emptyset}}
\lp[  p(\bmultix) \rp]
- \hspace{-8pt}
 \E_{\substack{D_{\alt} \sim \mu \\ \bmultiy \sim D_{\alt}}}
\hspace{-6pt}\lp[  p(\bmultiy) \rp]
\bigg| 
> 
\gamma \; \max \bigg( 
 \Var_{\bmultix \sim D_{\emptyset}}\lp[ 
p(\bmultix)
\rp]  
,
\hspace{-6pt}
\Var_{ \substack{D_{\alt} \sim \mu \\ \bmultiy \sim D_{\alt}}}\lp[ 
p(\bmultiy)
\rp]  
\bigg)^{1/2}\hspace{-12pt}.     
\end{align*}
\vspace{-0.45cm}
\end{restatable}

\noindent    The family of \emph{$n$-sample, $k$-degree polynomial tests} includes all tests $h: \R^{n\times d} \to \{0,1\}$ of the following form. For every polynomial $p: \R^{n\times d} \to \R$ of degree at most $k$, the family contains a test $h(\multix) = \sgn\left( p(\multix) - \kappa \right)$ that thresholds the polynomial at the point $\kappa$ which is defined to be the midpoint of the two expectations $\E_{\multix \sim D_{\emptyset}}
\lp[  p(\vec x^{(1)}, \cdots ,\vec x^{(n)}) \rp]$ and $\E_{D_{\alt} \sim \mu,\multiy \sim D_{\alt}}
\lp[  p(\vec y^{(1)}, \cdots, \vec y^{(n}) \rp]$.

It then follows immediately from Chebyshev's inequality that if there is a $\gamma$-advantageous polynomial $p$, then the test $h$ in the family corresponding to $p$ has bounded error probability:
\begin{align*}
    \Pr_{\multix \sim D_{\emptyset}}\big[ h\left(\multix \right) = 1 \big] + \Pr_{D_{\alt} \sim \mu,\multiy \sim D_{\alt}}\big[ h\left(\multiy \right) = 0 \big] \leq 8/\gamma^2 .
\end{align*}
That is, a polynomial 
with large advantage $\gamma$ directly translates to an effective tester 
with low error probability.
Conversely, an upper bound on $\gamma$ against 
all low-degree polynomials indicates that 
no tests (with high success probability) 
exist under this design framework. 

A convenient fact about \Cref{def:good-ld} is that 
the optimal advantage for a testing problem is relatively 
easy to analyze.\footnote{This holds for the variant 
where the RHS in \Cref{def:good-ld} includes 
only the variance with respect to $D_{\emptyset}$, 
which is usually a much simpler distribution. 
Using this variant suffices for the purpose of proving 
upper bounds on $\gamma$.} 
In particular, if one squares the defining equation 
for $\gamma$-advantage, it becomes a bound on the relative 
size of two explicit quadratic forms 
over the space of all degree at most $k$-polynomials, 
which can be explicitly optimized 
to find the optimal value of $\gamma$.
This yields a convenient framework for analyzing the power of Low-Degree Polynomial tests, in which many quantitatively 
tight lower bounds have been established; 
see, e.g., ~\cite{hopkins2017efficient,bandeira2019computational,kunisky2019notes,BBHLS21,mao2021optimal,schramm2022computational,ding2024subexponential} 
for a variety of statistical problems.\looseness=-1

\vspace{-0.2cm}

\paragraph{(General) Polynomial Threshold Function Tests}
Following the above discussion, one could see that
a key limitation of existing hardness results for the LDP tests family is that 
they only yield lower bounds against the \emph{specific} proof technique based on second-moment Chebyshev's inequality
while leaving the general power of all \emph{polynomial threshold function} tests, namely all tests of the form $\sgn( p( \vec x^{(1)}, \cdots, \vec x^{(n)} ) )$, where $p$ is an arbitrary low-degree polynomial, and $\vec x^{(i)}$ are the samples drawn, poorly understood; see the definition below for a formal definition of a low-degree PTF test.
\begin{definition}[$\beta$-good PTF test]\label{def:goodptf}

Let $D_{\emptyset}$ be a distribution in $\R^d$, $\cD_{\alt}$ be a family of distributions in $\R^d$, $\mu$ be a distribution on $\cD_{\alt}$, and $H$ be the hypothesis testing problem whose null distribution, alternative distributions family, and prior distribution are given by $D_{\emptyset}$, $\cD_{\alt}$, and $\mu$ accordingly. 
Let $h: \R^{n \times d} \mapsto \{0, 1\}$ be a polynomial threshold function of degree-$k$. 
We say $h$ is a $\beta$-good PTF test for $H$ if it satisfies that
\begin{align}
\label{eq:good-ptf-def}
\bigg| 
\E_{ \multix \sim D_{\emptyset} }
\lp[ 
h( \vec x^{(1)}, \cdots, \vec x^{(n)} )
\rp]
- 
\E_{ D_{\alt} \sim \mu, \multiy \sim D_{\alt} }
\lp[ 
h( \vec y^{(1)}, \cdots, \vec y^{(n)} )
\rp]
\bigg|
\geq \beta.    
\end{align}
\end{definition}

There is no clear reason to expect that the failure of this specific technique on a test problem would rule out all PTFs. 
Indeed, as we show in \Cref{sec:comparison}, there exist simple examples of pairs of null and alternative distributions for which no polynomial satisfies the separation condition of \Cref{def:good-ld}, yet PTF tests effectively solve the testing problem with high success probability.

Understanding the power of general PTF tests is considered 
as a prominent research direction within the relevant community~\cite{hopkins2024, wein2024}. Specifically, 
a recent workshop~\cite{hopkins2024lowdegree} 
on information-computation tradeoffs 
highlighted PTF tests as one of the main directions 
in the frontier of this field. In his new survey~\cite{WeinSurvey}, 
Wein writes: \emph{``[...] it is an interesting open problem to rule out other notions of success such as thresholding, but this seems beyond our current capabilities.''} 
In this work, we therefore ask the question:\looseness=-1
\begin{center}
\emph{Can we rigorously prove information-computation gaps within the family of all PTF tests?}
\end{center}
As our main contribution, we answer this question in the affirmative. 
In particular, we give the first  PTF testing 
lower bound for the fundamental task of \emph{Non-Gaussian Component Analysis}. As an immediate corollary, we obtain PTF lower bounds 
for a number of other statistical problems.

\vspace{-0.4cm}

\paragraph{Non-Gaussian Component Analysis}

Historically, NGCA
is a problem that originates from the signal processing literature \cite{blanchard2006search}, and has since
attracted much attention from the algorithmic statistics and theoretical machine learning communities; see \cite{sugiyama2008approximating,sasaki2016non,diederichs2013sparse,goyal2019non,dudeja2024statistical}. 
Informally, the problem corresponds to the task of
searching for a non-Gaussian direction of some high-dimensional distribution. 
The testing version of NGCA aims at 
distinguishing between a high dimensional standard Gaussian $\normal(\vec 0,\vec I)$ and a distribution that is non-Gaussian along an unknown direction, but behaves like standard Gaussian in every other orthogonal direction. 
The alternative distribution, commonly referred to as the hidden-direction distribution is defined below:
\begin{definition}[High-Dimensional Hidden Direction Distribution]\label{def:hidden_distr}
    For a distribution $A$ on $\R$ and a unit vector $v$ of $\R^d$, we denote by $\ngca$ the distribution of the random variable $\vec x + \xi \vec{v}$, where $\vec x \sim \normal(\vec 0, \vec I - \bv \bv^\top)$ and $\xi \sim A$. That is, $\ngca$ is the distribution which coincides with $A$ on the direction of $\bv$ and is standard Gaussian in the orthogonal subspace.
\end{definition}

The standard assumption is that the non-Gaussian component $A$ is also 
similar to $\normal(0,1)$ in the sense that their first $m$ moments 
match. The higher the $m$, the harder it becomes to distinguish 
$\ngca$, for random $\bv$, from the standard multivariate Gaussian. 
On the other hand, for algorithms to exist, one needs to assume that 
the $(m+1)$st moment differs by a non-trivial amount.
\begin{problem}[Non-Gaussian Component Analysis]\label{def:ngca_problem}
    Let $A$ be a distribution in $\R$ that matches the first $m$ moments with $\normal(0, 1)$.
    We are given $n$ samples generated in one of the following two ways:
    \begin{itemize}[leftmargin=1.5em, itemsep=0pt, topsep=3pt, partopsep=0em]
        \item (Null Hypothesis) The samples are drawn \iid from $D_{\emptyset} =\normal(\vec 0, \vec I)$.
        \item (Alternative Hypothesis) First, a unit vector $\vec v \in \R^d$ is drawn uniformly at random from the unit sphere, then the samples are drawn \iid from $D_{\alt} = \ngca$ from \Cref{def:hidden_distr}.
    \end{itemize}
    Given the samples, the goal is to distinguish between the two cases with high constant probability.
\end{problem}

A concrete motivation in studying NGCA is that it exhibits \emph{information-computation gaps} 
when $A$ is carefully constructed to match many moments with the standard Gaussian.
Information theoretically, it is known that the sample complexity of this problem is $O(d)$ under some mild assumptions on the distribution $A$ (see, e.g.,~\cite{vempala2011structure}).
Nonetheless, all known efficient algorithms require significantly larger resources (see e.g., \cite{dudeja2024statistical}).
Prior work has given formal evidence that
solving NGCA requires either access to a large amount of information (specifically, many \iid samples from the test distribution) 
or significant computational resources.
Concretely, this phenomenon has been established for several well-studied families of algorithms, including Statistical Query (SQ) algorithms \cite{DKS17-sq,diakonikolas2023sq}, Sum-of-Squares (SoS) algorithms \cite{diakonikolas2024sum}, and Low-Degree Polynomial tests \cite{mao2021optimal,BBHLS21}.
Since NGCA can be used to embed hard instances of several other statistical tasks~\cite{DKS17-sq}, a lower bound against 
NGCA directly implies lower bounds for several other problems.
In this work, we give the first lower bound of NGCA against the family of PTF tests.\looseness=-1

\vspace{-0.4cm}

\paragraph{Connection to Pseudorandom Generators}
There is a close connection between lower bounds for hypothesis 
testing and the theory of pseudorandom generators (PRGs). Specifically, 
for a test to effectively distinguish between the null and alternative 
hypotheses, the probability of accepting each hypothesis must differ significantly. 
In the language of PRGs, \Cref{def:goodptf} failing implies that 
$D_{\text{alt}}^{\otimes n}$ ``fools'' $h$ with respect to 
$D_\emptyset^{\otimes n}$ with error at most $\beta$. Thus, proving 
lower bounds against degree-$k$ PTF tests with $n$ samples reduces to 
showing that $D_{\text{alt}}^{\otimes n}$ fools low-degree PTFs with 
respect to $D_\emptyset^{\otimes n}$.\looseness=-1

There is an extensive literature on PRGs for PTFs. 
These works generally aim to construct low-entropy distributions that fool PTFs under structured high-entropy distributions such as Gaussians,  
see, e.g.,~\cite{DGJ+09,DKN10,meka2010pseudorandom,Kane11ccc,Kane11focs,Kane12subpoly,o2020fooling,kelley2022random}. Despite this distinction, 
some techniques 
developed in the PRG literature (in particular \cite{kelley2022random}) can be leveraged in our setting.

For our specific NGCA problem, the null distribution is precisely the standard Gaussian---for which some of the strongest PRGs for PTFs are known. Moreover, the alternative distributions are assumed to match many moments with the Gaussian, a condition that is essentially necessary for indistinguishability and one that is also required by SQ and Low-degree lower bounds for establishing information-computation gaps in the literature.\looseness=-1

\subsection{Our Result}
Our main theorem establishes that when the non-Gaussian component $A$ from \Cref{def:ngca_problem} matches the first $m$ moments with $\normal(0,1)$,  
there is no degree-$k$ PTF that satisfies \Cref{def:goodptf} with $\beta = 0.11$ for the NGCA hypothesis testing problem of \Cref{def:ngca_problem}, unless at least one of the following holds: The sample complexity $n$ is at least $d^{\Omega(m)}$ or the degree $k$ of the PTF is at least $d^{\Omega(1)}$ (suggesting a runtime of $(nd)^{ d^{\Omega(1)} }$).
Notably, this matches quantitatively with the computation-statistic tradeoff established for the weaker model of LDP tests.\looseness=-1

\begin{theorem}[Main Result]
\label{thm:main}
There exists a sufficiently large absolute constant $C^*$ such that the following holds.
For any $c^* \in (0, 1/4)$, $d,k,n,m \in \mathbb Z_+$ such that (i) $m$ is even, (ii) $\max(k,m) < d^{c^*/C^*}$, and (iii) $n < d^{ (1/4 - c^*) m}$,
we have that 
if $p: \R^{n \times d} \mapsto \R$ is a degree-$k$ polynomial,  and $A$ is a distribution on $\R$ that matches the first $m$ moments with $\normal(0,1)$, then:
\vspace{-3pt}
\begin{align}\label{eq:thm_conclusion}
\Bigg| \E_{  \substack{\vec v \sim \usphere\\ \multix \sim \ngcav }}
\lp[  \sgn(p(\vec x^{(1)}, \cdots, \vec x^{(n)})) \rp]
-  
\E_{ \multix \sim \normal(\vec 0, \vec I)}    
\lp[  \sgn(p(\vec x^{(1)}, \cdots, \vec x^{(n)})) \rp] \Bigg| 
\leq 0.11.
\end{align}
where $\ngcav$ denotes the hidden direction distribution from \Cref{def:hidden_distr}, and 
$\sgn : \R \to \{0,1\}$ is the sign function with $\sgn(x) = 1$ if and only if $x \geq 0$.
\end{theorem}

One way of interpreting \Cref{thm:main} is as follows. 
Let $c^*$ be an arbitrary constant chosen from $(0, 1/4)$.
If the PTF test uses  $n < d^{(1/4 - c^*) m}$ many samples, then the degree of the tester $k$ must be at least some polynomial in the sample dimension $d$ in order for the tester to be effective.
The degree $k$ can be interpreted as a parameter controlling the runtime of the tester.
For an arbitrary polynomial $p$, 
the runtime for this computation 
is on the order of $\poly((nd)^k)$---as this is the space required 
for storing all relevant coefficients of $p$. 
Consequently, \Cref{thm:main} implies an inherent trade-off between the 
exponential runtime $(nd)^{d^{\Theta(c^*)}}$, and 
the sample complexity $d^{(1/4 - c^*) m}$ for the family of PTF tests.\looseness=-1

One might wonder to what extent is the tradeoff established in \Cref{thm:main} optimal, and how our result compares quantitatively to the hardness results shown in other restricted computation models.
For hardness results against Statistical Query algorithms
and LDP tests, existing lower bounds imply that there are no ``efficient'' testers within these families if the number of samples drawn is less than $d^{ (1/2 - c^*) m }$ (see \Cref{cor:low-deg-hardness-general-problem} and \Cref{thm:SQ}
for the formal statements).
Interestingly, the constant $1/2$ in the exponent of 
these results is better than the constant $1/4$ 
that appears in \Cref{thm:main}.
Surprisingly, this gap---rather than being an artifact of our proof technique---is inherent for the family of PTF tests.
In particular, if one fixes $m$ to be some constant, we show 
in \Cref{thm:optimality} that there exists a degree-$\Theta(\log d)$ PTF that draws only
$\tilde \Theta( d^{m/4} )$ many samples, and effectively solves the NGCA problem for some specific moment-matching distribution $A$.
As a sharp contrast, the SQ, and LDP test lower bounds predict that no such test should exist.
This suggests that the constant $1/4$ that appears in \Cref{thm:main} is indeed worst-case optimal, 
and that PTF tests are (slightly) more powerful 
than SQ and LDP tests for certain NGCA problem instances. 
We leave it as an interesting open question whether the lower bound against PTF tests can be improved by making further structural assumptions on the non-Gaussian component $A$ beyond the moment-matching condition.
\looseness=-1

Before we end this subsection, we briefly comment on the subtle condition that $m < d^{c^*/C^*}$.
This condition is commonly used in the literature for lower bounds  against NGCA (see \cite{DKS17-sq,diakonikolas2023sq}).
Informally, the necessity of such a condition can be seen as follows. 
If $m \approx d$, one cannot hope to prove a computation lower bound 
of $(nd)^{m} \approx 2^{ d \log d }$. 
There is always a simple test that runs in time $2^{ O(d) }$: 
construct an exponential size cover of all possible directions 
$\vec v$, project the distribution along each possible direction 
and reduce to a one-dimensional testing problem.

\vspace{-0.4cm}

\paragraph{Application to Other Statistical Tasks}
For many important statistical problems---such as robust mean estimation, list-decodable mean estimation, and learning Gaussian mixture models--- that are seemingly unrelated to each other, 
one can construct hard instances that can be encoded as 
NGCA instances; see Chapter 8 of \cite{diakonikolas2023algorithmic} 
or Section 1 of \cite{diakonikolas2024sum} for a more thorough treatment. This means that proving an information-computation gap 
for NGCA within a given computational model translates 
to information-computation gap for all these tasks. 
See \Cref{tab:sample_complexity} for a summary of some PTF 
testing lower bounds obtained as corollaries of \Cref{thm:main}.

\begin{table}[h]
\centering
\textbf{\small Information-Computation Gaps for PTF tests}
\scalebox{0.9}{
\begin{tabular}{@{}l|c|c@{}}
\toprule
\makecell{Statistical Estimation Task}                      & Information-Theoretic & 
\makecell{Sample Complexity \\for low-degree PTFs}
\\ 
\midrule
\makecell{Robust Mean Estimation up to $\ell_2$-error 
\\
$O(\tau\sqrt{\log(1/\tau)}/B^2)$   
with Isotropic Gaussians}
& $O_{\tau,B}(d)$             & $\Omega(d^{ 
 B(1-c^*)/4})$ \\ \hline
\makecell{Robust Mean Estimation up to $\ell_2$-error \\
$O(\frac{1}{m} \tau^{1 - 1/m})$  with bounded $m$-th moments}
& $O_{\tau}(d )$     \arraybackslash       & $\Omega(d^{m(1-c^*)/4})$ \\ \hline
\makecell{List-decodable Mean Estimation to error $O((m\tau)^{-1/m})$}  & $O_{\tau}(d)$         & $\Omega(d^{m(1-c^*)/4})$ \\ \hline
\makecell{Learning the mixture of $m$ Gaussians}                 &   $\widetilde{O}(m d)$         & $\Omega(d^{m(1-c^*)/2})$ \\ 
\bottomrule
\end{tabular}
}
\caption{\small Comparison of our PTF lower bounds with the information-theoretic sample complexity for various tasks. The parameter $\tau$ is the rate of contamination, and the parameter $c^*$ can be set to any arbitrarily small constant. See \Cref{sec:concrete_apps}
for the formal statements of the results that appear in the table.}\label{tab:sample_complexity}
\end{table}

\vspace{-0.4cm}

\paragraph{Comparison with Existing Lower Bounds}
In general, lower bounds against PTF tests are incomparable to SQ and SoS lower bounds, as they capture different structural limitations of learning algorithms.
For the specific problem of Non-Gaussian Component Analysis, SQ lower bounds are effectively equivalent to Low-Degree Polynomial (LDP) lower bounds as established in \cite{BBHLS21}. Our result therefore strengthens prior SQ and LDP lower bounds by demonstrating hardness under the more general PTF testing framework. 

\vspace{-0.2cm}

\paragraph{Future Directions}
Our work proves the first lower bounds against general PTF tests, a strenthening of the well-studied low-degree polynomial test family, for multiple statistical problems.
Several open questions remain for future research. 
A key technical question is whether our lower bounds can be quantitatively improved under additional structural assumptions on the non-Gaussian component distribution $A$ that arise naturally in learning-theoretic settings, like bounded chi-square distance between $A$ and $\normal(0,1)$. Identifying such refinements could lead to sharper hardness results and a deeper understanding of the limitations of PTF tests. 
Another important direction is to obtain PTF testing 
lower bounds for other fundamental statistical estimation problems 
not covered in this paper, 
such as planted clique \cite{BHKKMP16} 
and sparse principal component analysis \cite{zou2006sparse}.

\subsection{Technical Overview}\label{sec:technical-overview}
We will use the notation $\bmultix$ to denote the sequence of vectors $\multix$.
We want to show that for any low degree polynomial $p: \R^{n \times d} \mapsto \R$, we have that
\begin{align}
\label{eq:technique-main-objective}
\E[\sgn \lp( p( \vec x^{(1:n)} ) \rp)] 
\approx
\E\lp[\sgn \lp(p( \vec y^{(1:n)} ) \rp) \rp] 
\, , 
\end{align}
where each $\vec x^{(i)} \in \R^d$ follows an independent standard Gaussian, and each $\vec y^{(i)} \in \R^d$ follows our hidden direction distribution $\ngcav$
(see \Cref{def:hidden_distr}) using a fixed vector $\vec{v}$, which is sampled once uniformly at random from the unit sphere at the start.

We begin with a brief recap of the setup in the PRG literature, and present a high level comparison between these approaches and ours. 
These works consider a generator distribution $\bar {\vec y} = n^{-1/2} \sum_{i=1}^n \baryi$ in $\R^d$, where each $\baryi$ is chosen to be some low-entropy distribution whose low-degree moments match with the standard Gaussian.\footnote{In fact, each $Y_i$ is 
independently from a $k$-wise independent family of Gaussian distributions, which effectively satisfies the low-entropy and moment-matching conditions.
} 
Given an arbitrary low-degree polynomial $q: \R^d \mapsto \R$, the goal is then to show that
$
\E[  \sgn(q( \vec x  ) ) ]
\approx
\E\lp[  \sgn\lp(q\lp( n^{-1/2} \sum_{i=1}^n \baryi \rp) \rp) \rp] , 
$
where $\vec x$ is a standard Gaussian vector in $\R^d$.
Note that the standard Gaussian $\vec x$ 
can be alternatively written as $\vec x = n^{-1/2} \sum_{i=1}^n \vec x^{(i)}$,  where the $\vec x^{(i)}$'s are themselves independent standard Gaussian distributions. 
Under this setup, their objective can be alternatively formulated 
\vspace{-2pt}
\begin{align}
\label{eq:prg-main-objective}
\E\lp[  \sgn \lp(\bar p \lp( \vec x^{(1:n)}\rp) \rp) \rp]
\approx
\E\lp[  \sgn\lp(\bar p \lp( \bar {\vec y}^{(1:n)} \rp) \rp) \rp]  \quad \text{where} \quad \bar p( \vec z^{(1:n)}  ) = q\bigg(   \tfrac{1}{\sqrt{n}}\sum_{i=1}^n \vec z^{(i)} \bigg).
\end{align}
Under this new formulation, one can see that our objective \Cref{eq:technique-main-objective} looks particularly similar to theirs except for two subtle differences.  
First, in the PRG setup, the polynomial $\bar p: \R^{n \times d} \mapsto \R$ has the specific form of being a lifted version of a significantly \emph{lower-dimensional} polynomial $q: \R^d \mapsto \R$ (see \Cref{eq:prg-main-objective}). 
In fact, the structure turns out to be quite convenient for controlling the higher order derivatives of $\bar p$ in terms of each individual variable $\vec z^{(i)}$, which are subsequently exploited in the PRG literature to establish the left part of \Cref{eq:prg-main-objective}.
In our setup, the polynomial $p$ has a much more complex structure as it can be an \emph{arbitrary} polynomial with $n \times d$ many inputs. 
This makes the task seemingly intractable at a first glance. 
The twist lies in the second difference between our setups:
the underlying distribution of each $\vec y^{(i)}$ in our objective exhibits significantly more structure 
compared to the pseudorandom distributions of $\baryi$.
Specifically, instead of being an arbitrary distribution whose low-degree moments match with the standard Gaussian, the distribution of $\vec y^{(i)}$ is \emph{identical} to the Gaussian distribution in all but some \emph{randomly} chosen direction $\vec v$.
This turns out to be a valuable property that counteracts the complexity of the polynomial $p$ in our setup.
More concretely, our main insight is that the specific properties of the hidden direction distribution
essentially allow us to restrict our attention to higher order \emph{directional} derivatives of $p$ (projected onto the hidden direction $\vec v$).
By exploiting the fact that 
$\vec v$ is chosen randomly, we manage to show that the higher order directional derivatives of $p$ can be effectively bounded in a way that is qualitatively similar to (but quantitatively different from) the case in the PRG literature, allowing us to escape from the full complexity of the polynomial $p$ in all directions.

In what follows, we give a more detailed sketch of our arguments interleaved with technical comparisons with the prior work in the PRG literature.

\paragraph{Overall Framework: Hybrid Argument and Mollification}
The basic proof strategy is via the so called \emph{hybrid argument} developed in the PRG literature.
At a high level, the framework is motivated by the wishful thinking that the task may be significantly easier if 
one progressively replaces the samples $\vec x^{(1:n)}$ into $\vec y^{(1:n)}$.
In particular, in a single ``replacement step'', 
the goal is to show the following intermediate approximation steps for all $i=1,\ldots,n$:
\begin{align}
\E\lp[\sgn \lp( p \lp (\bx^{(1:i-1)},\bx^{(i)},\by^{(i+1:n)}\rp) \rp)\rp]    
\approx 
\E\lp[\sgn \lp( p \lp (\bx^{(1:i-1)},\by^{(i)},\by^{(i+1:n)}\rp) \rp)\rp] \pm o(1/n).
\label{eq:technique-single-replace-1}
\end{align}
If so, applying the triangle inequality $n-1$ times would complete the proof of \Cref{eq:technique-main-objective}.

In both the PRG and the NGCA settings, 
one would like to exploit the assumption that the alternative distribution matches many moments with the standard Gaussian.
Hence, a natural attempt to show \Cref{eq:technique-single-replace-1} would be to rewrite the expected values in terms of the moments of the corresponding distributions. 
If the PTF function were a smooth function, this can be achieved by considering the Taylor expansion of the function.
However, it is not hard to see that the function is discontinuous when the polynomial $p$ evaluates to $0$, making Taylor's theorem not directly applicable.

To circumvent the issue, 
the PRG literature uses the idea of \emph{mollification}.
Specifically, 
 we would like to construct a sufficiently smooth function $h: \R^{n \times d} \mapsto [0, 1]$ that approximates the PTF $\sgn(p( \cdot ))$ well.
As we have said, the PTF function $\sgn(p(\cdot))$ is discontinuous when $p(\vec x) = 0$.
Hence, we cannot hope to have a smooth and pointwise close approximation near the zeros of the polynomial. 
Instead, we want $h$ to have the same behavior as the PTF when the polynomial is ``large'' (in a technical sense that will be specified later), and then smoothly interpolate in the other case.
On the one hand, this effectively ensures the smoothness property of $h$ globally.
On the other hand, $h$ indeed approximates the PTF well under the Gaussian distribution\footnote{
A technical detail we omit here is that naively it seems like one also needs to show a similar approximation result under the distribution of $\vec y^{(1:n)}$, which is technically challenging.
Fortunately, this can be solved by a standard sandwich trick that reduces the task into constructing two mollified PTFs $h_+$, $h_-$ such that
(i) $h_+\lp( \cdot \rp) \geq \sgn\lp( p\lp(\cdot  \rp) \rp) \geq h_-\lp( \cdot \rp)$, 
and (ii) $h_+$, $h_-$ approximates the PTF function well under the Gaussian distribution. 
This saves us from the trouble of showing how well the mollified PTF approximates the PTF under the less structured distributions of $\vec y^{(1:n)}$.
We refer the reader to \Cref{sec:framework}, and more specifically \Cref{lem:sandwich} for more detail on this.
}
as the only disagreement region is when the polynomial $p$ is small, whose probability mass can be bounded by the Gaussian anti-concentration properties or some variants (that will be discussed later on).
For convenience, we refer to this smoothed approximation $h$ of the PTF as the \emph{mollified} PTF, and our task is now reduced to showing 
\begin{align}
\E\lp[h \lp( p \lp (\bx^{(1:i-1)},\bx^{(i)},\by^{(i+1:n)}\rp) \rp)\rp]    
\approx 
\E\lp[h \lp( p \lp (\bx^{(1:i-1)},\by^{(i)},\by^{(i+1:n)}\rp) \rp)\rp] \pm o(1/n).
\label{eq:technique-single-replace}
\end{align}

It then remains for us to (1) construct a smooth mollified PTF $h$ that closely approximates $\sgn(p(\cdot))$ under Gaussian inputs, and (2) show a single replacement step (\Cref{eq:technique-single-replace}) for this smooth function $h$ (the formal version of this is in \Cref{prop:replacement-step}).

\vspace{-0.4cm}

\paragraph{Mollification by Strong Anti-Concentration  }
The difficulty of (1) is as follows: since the PTF function is discontinuous and $h$ is smooth, we cannot expect $\lp|h\lp( \vec x^{(1:n)} \rp)-\sgn \lp( p\lp( \vec x^{(1:n)} \rp) \rp)\rp|$ to be small for all inputs. 
Alternatively, one must rely on some type of \emph{anti-concentration} result saying that the probability of $ \vec x^{(1:n)} $ lying in the disagreement region between $h$ and $\sgn(p(\cdot))$ is small under the Gaussian distribution.
A natural idea to do so would be to construct 
a smooth approximation of the indicator function 
$ g( \vec x^{(1:n)} ) \approx
\mathbbm 1\{ |p( \vec x^{(1:n)} )| > \eps \}$, and define
the mollified PTF to be the product $g\lp(\vec x^{(1:n)} \rp) \sgn\lp( p\lp(\vec x^{(1:n)}\rp)\rp)$.
On the one hand, 
this ensures that the function will smoothly interpolate between $0$ and $1$ for inputs $\vec x^{(1:n)}$ near the zeros of $p$.
On the other hand, 
the disagreement region will be roughly the same as the 
set $\{ \vec x^{(1:n)} \mid p( \vec x^{(1:n)} ) < \eps \}$, 
which is guaranteed to have small mass by the famous Gaussian anti-concentration theorem from \cite{CW:01}. 
However, for an arbitrary polynomial $p$, the anti-concentration property decays rapidly
as the degree of the polynomial increases, i.e., 
there exists a degree-$k$ polynomial $p$ such that $\Pr[ p(\vec x^{(1:n)}) < \eps ] \approx \eps^{1/k}$. 
In the context of the PRG literature, this leads to an exponential dependency on the seed length\footnote{This was indeed the case for the early work \cite{meka2010pseudorandom} based on such naive mollification procedures.} while in our context the approach might break entirely once the polynomial degree $k$ becomes larger than the logarithm of the sample dimension $d$.
Needless to say, such an assumption would significantly weaken the lower bound on $k$ compared to the target 
of $k = d^{ \Omega(1) }$ in \Cref{thm:main}.

To tackle the issue, we borrow the ideas from \cite{o2020fooling,Kane11focs,kelley2022random} that take advantage of the \emph{strong anti-concentration} properties of polynomials. In particular, 
strong anti-concentration (see \Cref{eq:slow-growth}) is a relative notion of anti-concentration on the sizes of the derivatives of the polynomials:  given a polynomial $p$ of degree $k$, 
for any $0 \leq t \leq k-1$ and $\eps \in (0, 1)$, it holds that $ \|  \nabla^t p\lp( \vec x^{(1:n)} \rp) \|_F > \eps \|  \nabla^{t+1} p\lp( \vec x^{(1:n)} \rp) \|_F$ with probability at least  $1 - O(k^2 \eps)$, where $ \nabla^t p\lp( \vec x^{(1:n)} \rp) $ denotes the tensor containing all $t$-th order partial derivatives of $p$.
Notably, unlike the Carbery-Wright anti-concentration theorem, the failure probability here is only a polynomial in $k$, and it is precisely this polynomial dependency that makes the stronger lower bound on $k$ from \Cref{thm:main} possible.
By setting $\eps = k^{-2.1}$,  applying the union bound, and chaining the inequalities obtained, we obtain that\looseness=-1
\begin{align}
\label{eq:derivative-decay-intro}
\lp | p\lp( \vec x^{(1:n)} \rp) \rp| > k^{-2.1} \lp\| \nabla p \lp( \vec x^{(1:n)} \rp) \rp \|_F 
> k^{-6.2} \left\| \nabla^2\lp( \vec x^{(1:n)} \rp) \right\|_F > \ldots > k^{-2.1k}  \left\| \nabla^k p \lp( \vec x^{(1:n)} \rp) \right\|_F \, ,
\end{align}
with high constant probability.
Note that $\nabla^k p$ is a constant tensor.
Thus, 
as long as \Cref{eq:derivative-decay-intro} holds approximately (within a constant factor),
we can infer that $p\lp( \vec x^{(1:n)} \rp)$ must be non-zero and the PTF $\sgn\lp( p\lp( \vec x^{(1:n)} \rp)  \rp)$ will be smooth.
Therefore, it suffices to modify the function on inputs $\vec x^{(1:n)}$ where \Cref{eq:derivative-decay-intro} is not approximately true.
In particular, we can define a function $\rho: \R \mapsto [0, 1]$ that serves as a smooth approximation of the indicator function $\mathbbm 1\{ |z| \leq 1 \}$, and set the mollified PTF to be\looseness=-1
\begin{align}
\label{eq:first-attempt-mollify}
\tilde h \lp(  \vec x^{(1:n)} \rp) := 
    \prod_{t=1}^k \rho \lp(   
\frac{ k^{-\Theta(1)}  \lp\| \nabla^{t} p( \multix ) \rp \|_F^2 }{ \lp\| \nabla^{t-1} p( \multix ) \rp \|_F^2} \rp)
\; \sgn\lp( p\lp( \multix \rp)  \rp).
\end{align}
When \Cref{eq:derivative-decay-intro} holds, 
$\tilde h$ is simply the same as the smooth part of the original PTF function.
As the value of $| p( \mathbf x^{(1:n)} ) |$ approaches $0$, $\frac{ k^{-\Theta(1)}  \lp\| \nabla^{t} p( \multix ) \rp \|_F^2 }{ \lp\| \nabla^{t-1} p( \multix ) \rp \|_F^2}$ for some $t$ must begin to exceed $1$.
Consequently, the function $\rho(\cdot)$, which we define to be a smooth approximation of the indicator function $\mathbbm 1\{|x| \leq 1\}$ will decay smoothly until it reaches $0$, ensuring the smoothness of the function $\tilde h$ globally.
Moreover, since \Cref{eq:derivative-decay-intro} holds with high constant probability, this ensures $\tilde h$ approximates $\sgn(p(\cdot))$ up to a small constant error over Gaussian inputs.

As promising as it may seem, there are still substantial technical difficulties in showing the replacement step (\Cref{eq:technique-single-replace}) for this particular mollified PTF $\tilde h$.
At a high level, the difficulty stems from a design flaw in which \Cref{eq:first-attempt-mollify} fails to leverage the hidden directional structure of the underlying distribution of $\vec y^{(1:n)}$.
In the rest of the subsection,
we will present a natural attempt to prove the replacement step for this mollified PTF $\tilde h$, 
illustrate the difficulty encountered, and present a simple modification on top of $\tilde h$ to obtain our actual mollified PTF $h$ (see \Cref{eq:actual-mollifier}).\looseness=-1

\vspace{-0.4cm}

\paragraph{Replacement Step by Taylor Expansion}
Consider the following natural attempt in showing the $i$-th replacement step (\Cref{eq:technique-single-replace})
for the mollified PTF $\tilde h$.
Thanks to the smoothness property of $\tilde h$, we can rewrite $\tilde h$ using its Taylor expansion in terms of the variable $\vec z$ to be replaced: 
\begin{align}
\label{eq:intro-full-taylor-expansion}
\tilde h \lp( \vec x^{(1:i-1)}, \vec z,  
\vec y^{(i+1:n)} \rp)=  
\sum_{ t=0 }^{\infty}
 \lp \langle \nabla^{t}_i \tilde h\lp( \vec x^{(1:i-1)}, \vec 0,  
\vec y^{(i+1:n)} \rp)  \, , \, 
\vec z^{\otimes t} \rp \rangle \, ,
\end{align}
where $\nabla^{t}_i \tilde h$ denotes the tensor containing all $t$-th order partial derivatives of $\tilde h$ with respect to its $i$-th argument.
Suppose that the first degree-$m$ moments of $\vec z= \vec x^{(i)} \sim \normal(\vec 0, \vec I)$ and $\vec z= \vec y^{(i)} \sim \ngcav$ match exactly. We then have that the difference 
$ \E_{ \vec z \sim \normal(\vec 0, \vec I)} 
\lp[ 
\tilde h \lp( \vec x^{(1:i-1)}, \vec z,  
\vec y^{(i+1:n)} \rp) \rp]
- \E_{ \vec z \sim \ngcav} 
\lp[ 
\tilde h \lp( \vec x^{(1:i-1)}, \vec z,  
\vec y^{(i+1:n)} \rp) \rp]
$
comes only from the higher order terms (the terms with $t > m$ from \Cref{eq:intro-full-taylor-expansion}). 
As we have said, the function $\tilde h$ is carefully constructed to be as smooth as the polynomial $p$.
Concretely, one can show by some straightforward computation that $\lp\| \nabla^{t}_i \tilde h\lp(\vec x^{(1:n)}\rp) \rp\|_F$ should be roughly the same as  
$\lp\| \nabla^{t}_i p \lp(\vec x^{(1:n)}\rp) \rp\|_F
/ \lp| p\lp(\vec x^{(1:n)}\rp) \rp|$.
In the PRG literature, the particular polynomial $p$ showing up is of the form
$ p\lp(  \vec x^{(1:n)} \rp) = q\lp(  \tfrac{1}{\sqrt{n}} \lp( 
\vec x^{(1)} + \cdots + \vec x^{(n)}\rp) \rp) $.
Due to the specific form of $p$, one can see that $p$ has only a mild dependence on the $i$-th sample. 
Consequently,
it is not hard to show that $\lp\| \nabla^{t}_i p \lp(\vec x^{(1:n)}\rp) \rp\|_F
/ \lp| p\lp(\vec x^{(1:n)}\rp) \rp|$ is at most $ n^{-\Theta(t)}$. 
If we were to have the same bound on higher-order derivatives in our case, we could then consider the Taylor expansion truncated to the first degree-$m$ terms, which gives that
\begin{align*}
\tilde h \lp( \vec x^{(1:i-1)}, \vec z,  
\vec y^{(i+1:n)} \rp)=  
\sum_{ t=0 }^{m-1}
 \lp \langle \nabla^{t}_i \tilde h\lp( \vec x^{(1:i-1)}, \vec 0,  
\vec y^{(i+1:n)} \rp)  \, , \, 
\vec z^{\otimes t} \rp \rangle \, 
+ 
 \lp \langle \nabla^{m}_i \tilde h\lp( \vec x^{(1:i-1)}, \hat {\vec z},  \vec y^{(i+1:n)} \rp)  \, , \, 
\vec z^{\otimes m} \rp \rangle \, ,
\end{align*}
where $\hat {\vec z}$ is some vector that lies in the line between $\vec z$ and $\vec 0$.
After that, we note that the expected values of the first $(m-1)$ terms match exactly under $\normal(\vec 0, \vec I)$ and $\ngcav$ while the last term is on the order of $n^{-\Theta(m)} \ll o(1/n)$.
This would then readily conclude the proof of the replacement step.

However, in our case, since $p$ can be an arbitrary low-degree polynomial, we cannot say that $p$ and subsequently $\tilde h$ have only a weak dependence on the $i$-th sample.
In particular, the best bound on  $\lp\| \nabla^{t}_i p \lp(\vec x^{(1:n)}\rp) \rp\|_F
/ \lp| p\lp(\vec x^{(1:n)}\rp) \rp|$ (and therefore $\| \nabla_i^t \tilde h \|_F$) will be on the order of $k^{\Theta(t)}$ (due to the tightness of \Cref{eq:derivative-decay-intro}), which is an \emph{increasing} function in $t$.
As a result, the mollified PTF $\tilde h$ simply cannot be approximated by its low-degree Taylor expansion, 
and this appears to be a substantial barrier in showing the replacement step (\Cref{eq:technique-single-replace}) for $\tilde h$.

Our key insight is that we can instead leverage the specific structure of the underlying hidden direction distribution.
In particular, we can exploit the fact that the distribution of $\vec y^{(i)}$ only differs from that of $\vec x^{(i)}$ in some randomly chosen direction $\vec v$. 
As a result, instead of taking the Taylor expansion of $\tilde h$ viewed as a multivariate function in terms of the entire $i$-th sample, we can now consider the \emph{directional} Taylor expansion of $\tilde h$ restricted to the $\vec v$ direction. That is, we view $\vec z$ as $\bar {\vec z} + \xi \vec v$ for some $\bar {\vec z}$ orthogonal to $\vec v$, and we use the Taylor expansion with respect to the scalar variable $\xi$:\looseness=-1
\begin{align*}
\tilde h\lp( \vec x^{(1:i-1)} , \bar {\vec z} + \xi \vec v ,\vec y^{(i+1:n)} \rp) = 
\sum_{t=0}^{m-1}
D_{i, \vec v}^t \tilde h\lp( \vec x^{(1:i-1)} , \bar {\vec z},\vec y^{(i+1:n)} \rp) \xi^t
+
D_{i, \vec v}^m \tilde h\lp( \vec x^{(1:i-1)} , \bar {\vec z} + \hat \xi,\vec y^{(i+1:n)} \rp) \xi^m
\, ,
\end{align*}
where $\hat \xi$ is some point in $[0, \xi]$,  and
$D_{i, \vec v}^t \tilde h$ denotes the $t$-th order
directional derivative of $\tilde h$ with respect to the $i$-th sample, i.e., $D_{i, \vec v} \tilde h = \vec v^\top \nabla_i \tilde h$.
Using the directional Taylor expansion, one can then write\looseness=-1
\begin{align*}
&\E_{ \vec z \sim \normal(\vec 0, \vec I)} 
\lp[ 
\tilde h \lp( \vec x^{(1:i-1)}, \vec z,  
\vec y^{(i+1:n)} \rp) \rp]
- \E_{ \vec z \sim \ngcav} 
\lp[ 
\tilde h \lp( \vec x^{(1:i-1)}, \vec z,  
\vec y^{(i+1:n)} \rp) \rp] \\
&= \E_{ \bar {\vec z} \sim \normal( \vec 0, \vec I - \vec v \vec v^\top ) }
\lp[ 
\sum_{t=0}^\infty
D_{i, \vec v}^t \tilde h\lp( \vec x^{(1:i-1)} , \bar {\vec z},\vec y^{(i+1:n)} \rp)
\lp( 
\E_{ \xi \sim A } \lp[ \xi^t  \rp]
- 
\E_{ \xi \sim \normal(0, 1) } \lp[ \xi^t  \rp] \rp)
\rp].
\end{align*}
Since we assume that the first degree $m$ moments of $A$ match exactly with $\normal(0, 1)$, it suffices to argue  
that the size of the $t$-th order directional derivative 
$D_{i,\vec v}^t \tilde h$ is of diminishing size as a function of $t$.\looseness=-1

\vspace{-0.4cm}

\paragraph{Derivative Decay in a Random Direction}
Towards this goal, let $\vec x^{(1:n)}$ 
be input samples following the standard Gaussian distribution,
and we will show that the directional derivative  
$  D_{i, \vec v}^{(t)} \tilde h( \vec x^{(1:n)} )$ will have small size with high constant probability, where the randomness is over $\vec v$ and $\vec x^{(1:n)}$.
With some straightforward computation, one can show that 
the above derivative 
has about the same size as 
the tensor containing all $t$-th order directional derivatives of the polynomial $p$ along the direction of $\vec v$. 
Specifically, we consider the tensor $p^{[t], \vec v}( \vec x^{(1:n)} ) \in \lp( \R^n \rp)^{\otimes t}$ defined as follows:
$$
p^{[t], \vec v}_{i_1, \ldots, i_t}(\multix)
= D_{i_1, \vec v} D_{i_2, \vec v} \cdots D_{i_t, \vec v} p(\multix).
$$
Alternatively, this directional derivative of $p$
can be written as a product between
the gradient tensor $\nabla^{t} p(\vec x^{(1:n)}) $
and the ``random direction tensor'' $\vec v^{\otimes t}$.
By the strong anti-concentration property of Gaussian, $\vec x^{(1:n)}$ satisfies \Cref{eq:derivative-decay-intro}, and $\nabla^{t} p(\vec x^{(1:n)}) $ must have size at most $k^{\Theta(t)} \lp| p(\vec x^{(1:n)}) \rp| $ with high constant probability.
Therefore, the directional derivative can be large only if the full gradient tensor $\nabla^t p( \vec x^{(1:n)} )$
correlates well with $\vec v^{\otimes t}$. 
Since the distribution of $\vec x^{(1:n)}$ is rotationally invariant (as a property of the standard Gaussian) and $\vec v$ is chosen randomly, it can be intuitively seen that the correlation will be small with high constant probability.
In particular, via some technical tensor computation, we show in \Cref{lem:derivative-decay} that
$\lp\| p^{[t], \vec v}_{i_1, \ldots, i_t}\lp( \vec x^{(1:n)} \rp) \rp\|_F   / | p( \vec x^{(1:n)} )| $
will be on the order of $d^{-t/4} k^{ O(t) }$ with high constant probability over the random choice of $\vec v$ and $\vec x^{(i)} \sim \normal(\vec 0, \vec I)$.
Consequently, if we assume that the degree $k$ is a sufficiently small polynomial in $d$, i.e., 
$k \ll d^{ \Theta(1) }$, $d^{-t/4} k^{ O(t) }$ should then be an exponentially decreasing function in $t$, implying that 
$\tilde h$ is indeed well-approximated by its low-degree directional Taylor expansion.

However, one subtle technical issue remains. 
That is, the above argument works only for the first replacement step when the input samples all follow the standard Gaussian distribution.
In later replacement steps, some of the samples have already been replaced with the hidden direction distribution, which can make the gradient correlate strongly with the hidden direction $\vec v$.
As a result, we can no longer effectively control the size of the directional derivative even for $\vec x^{(1:n)}$ satisfying the derivative decay condition in \Cref{eq:derivative-decay-intro}.

\vspace{-0.4cm}

\paragraph{Controlling Derivatives by Fine-Tuning the Mollifier}

To circumvent the issue, we proceed with the following simple modification to the mollifier function to explicitly check for derivative decay along the $\vec v$-direction: instead of 
zeroing out the inputs on which the full derivative tensor $\nabla^{t} p$ violates the strong anti-concentration property stated in \Cref{eq:derivative-decay-intro},
we now define a new mollified PTF defined directly with respect to the directional derivatives of $p$:
\begin{align}
\label{eq:actual-mollifier}
h \lp(  \vec x^{(1:n)} \rp):=
    \prod_{t=1}^k \rho \lp(   
\frac{ d^{\Theta(t)}  \lp\| p^{[t], \vec v}( \multix ) \rp \|_F^2 }{  p( \multix )^2} \rp) \sgn\lp( p( \multix )\rp) \, ,
\end{align}
where $\rho$ is as before a smooth approximation of the indicator function $\mathbbm 1\{ |z| \leq 1 \}$.
On the one hand,
as a direct implication of \Cref{lem:derivative-decay}, 
we have that 
the disagreement region between $h$ and the original PTF will be small for a randomly chosen $\vec v$ under the Gaussian distribution.
On the other hand, this simple modification ensures smoothness along the $\vec v$ direction for an arbitrary input $\vec x^{(1:n)}$ that makes the mollifier PTF $h$ non-zero.\looseness=-1

By some tedious but straightforward computation (\Cref{lem:sum-of-product}), one can show that the size of the $t$-th order directional derivative of $h$ is still roughly comparable to  $\lp \| p^{[t], \vec v}\lp( \vec x^{(1:n)} \rp) \rp \|_F / \lp| p\lp( \vec x^{(1:n)} \rp) \rp| $.
We can then do a simple case analysis.
If the ratio is large, then the term
$\rho \lp(   
\frac{ d^{\Theta(t)}  \lp\| p^{[t], \vec v}( \multix ) \rp \|_F^2 }{  p( \multix )^2} \rp)$ ensures that the entire mollified PTF evaluates to $0$.
Otherwise, we can conclude that the $t$-th order directional derivative $D_{i, \vec v}^t h$ will be roughly on the order of $ d^{ -\Theta(t)}$ (cf. \Cref{lem:mollify-derivative-decay}).
Due to the moment matching condition, 
when one computes the difference $ \E_{ \vec z \sim \normal(\vec 0, \vec I)} 
\lp[ 
h \lp( \vec x^{(1:i-1)}, \vec z,  
\vec y^{(i+1:n)} \rp) \rp]
- \E_{ \vec z \sim \ngcav} 
\lp[ 
h \lp( \vec x^{(1:i-1)}, \vec z,  
\vec y^{(i+1:n)} \rp) \rp]
$ via the directional Taylor expansion, the first degree $m$ terms in $\xi$ cancel exactly, leaving us with the dominating term
$$
D_{i, \vec v}^{m+1} h\lp( \vec x^{(1:i-1)} , \bar {\vec z},\vec y^{(i+1:n)} \rp) 
\lp(
\E_{\xi \sim A} \lp[
\xi^{m+1}
\rp] 
- 
\E_{\xi \sim \normal(0, 1)} \lp[
\xi^{m+1}
\rp]
\rp)
 \, ,
$$
which can be appropriately bounded by
$d^{ -\Theta(m)}$.\footnote{ 
Technically, for the bound to holds, we will need to assume that the support of $A$ is contained within $[d^{-c}, d^c]$ for some sufficiently small constant $c$.
To circumvent the issue, we instead show the replacement step for the truncated distribution $\bar A$, and then relate the expected value of $h$ under $\xi \sim \bar A$ back to the one under $\xi \sim A$ using the fact that $h$ is bounded between $[0, 1]$ (cf. \Cref{lem:truncation}).
} 
As long as we have $d^{ -\Theta(m) } \ll o(1/n)$, 
the replacement step will go through, and this concludes the sketch of our proof of \Cref{thm:main}.

\section{Notation}\label{sec:notation}

\subsection{Basic Notation}
We use $\mathbb{Z}_+$ for the set of positive integers and $[n]$ to denote $\{1,\ldots,n\}$. 
We use bold lowercase letters for vectors and bold uppercase letters for tensors.
We use $\vec 1_d$ for the $d$-dimensional all-one vector, $\vec 0_d$ for the $d$-dimensional all-zero vector,
and $\vec I_d$ for the $d \times d$ identity matrix. When the dimension is clear from the context, we will drop the subscript.
For a set $S$, we use $\mathcal{U}(S)$ to denote the uniform distribution on $S$.
Given a distribution $D$ in $\R^d$, we write $\vec x^{(1)}, \cdots, \vec x^{(n)} \sim D$ to denote $n$ \iid samples from $D$. For the sake of saving space, we often write $\bmultix$ to denote the sequence of vectors $\multix$.  We write $a \gg b$ to denote that $a \geq C \; b$ for some sufficiently large constant $C>0$.

\subsection{Tensor Notation}\label{sec:tensor}
We frequently use tensors in $(\R^n)^{\otimes k}$. For some $\vec A \in (\R^n)^{\otimes k}$, we denote by $\vec A_{i_1,\ldots,i_k}$ the entry in $\vec A$ indexed by $i_1,\ldots,i_k \in [n]$. For two tensors $\vec A, \vec B \in (\R^n)^{\otimes k}$, we define the inner product (or dot product) between them as$\langle\vec A, \vec B \rangle = \sum_{i_1,\ldots,i_k \in [n]} \vec A_{i_1,\ldots,i_k} \vec A_{i_1,\ldots,i_k}$. We use $\vec A^\flat$ to denote the flattened version of $\vec A$, i.e.,  the vector in $\R^{n^k}$ obtained by stacking all entries of $\vec A$ into a single vector in lexicographic order. We define the Frobenius norm of tensor $\vec A$ to be $\| \vec A\|_F := \sqrt{\langle\vec A, \vec A \rangle}$.
For two tensors $\vec A \in (\R^n)^{\otimes{k}}$ and $\vec B \in (\R^n)^{\otimes{\ell}}$ with $k > \ell$ we write $\vec A \vec B$ to denote the tensor in $(\R^n)^{\otimes(k-\ell)}$ defined as follows: 
$$(\vec A \vec B)_{i_{\ell+1},\ldots,i_k } := \sum_{i_1,\ldots,i_\ell \in [n]} \vec A_{i_1,\ldots,i_\ell} \vec B_{i_1,\ldots,i_\ell}.$$
Note that for $\vec A, \vec B \in (\R^n)^{\otimes{k}}$ it holds that $\vec A \vec B$ defined as above is the same as the inner product $\langle\vec A, \vec B \rangle$.
Moreover, using the tensor product notation, if $\vec e(1),\ldots, \vec e(n) \in \R^n$ denote the standard basis vectors, then $\vec T \, \vec e(i)$ is simply the tensor $\vec T$ restricted on the first index being equal to $i$. Similarly $\vec T \, (\vec e(i_1) \otimes \cdots \otimes \vec e(i_k))$ selects the sub-tensor with the first $k$ indices being $i_1,\ldots,i_k$.

\subsection{Derivative Notation}\label{sec:derivative}
We write $p: \R^{n \times d} \mapsto \R$ to denote a multi-variable function with $n$ variables, where each variable is a $d$-dimensional vector.
We write $\nabla_i p$ to denote the vector of partial derivatives with respect to the $i$-th argument, i.e., $\nabla_i p: \R^{n \times d} \mapsto \R^d$ and 
$\lp( \nabla_i p(\multix) \rp)_j = \frac{\partial }{ \partial \vec x^{(i)}_j }
p(\multix)
$.
Without the subscript, $\nabla p(\multix)$ denotes the usual gradient vector that takes derivatives with respect to all arguments at the same time, i.e., $\nabla p: \R^{n \times d} \mapsto \R^{n \times d}$ and 
$\lp( \nabla p(\multix) \rp)_{i,j} = \frac{\partial }{ \partial \vec x^{(i)}_j }
p(\multix)$. We also define the $t$-th order derivative tensor as follows:

\begin{align}\label{eq:t-th-derivative}
\lp( \nabla^t  p(\multix) \rp)_{i_1,j_1,\ldots,i_t,j_t} = \frac{\partial}{\partial \vec x^{(i_t)}_{j_t}} \cdots \frac{\partial}{ 
\partial \vec x^{(i_1)}_{j_1} }
p(\multix)    
\end{align}

We will also make use of directional partial derivatives
$
D_{i, \vec v} p: \R^{n \times d} \mapsto \R
$ defined as:
\begin{align*}
  D_{i, \vec v} p( \vec x^{(1)}, \cdots, \vec x^{(n)} )
= \vec v^\top \nabla_i p( \vec x^{(1)}, \cdots, \vec x^{(n)} )
\end{align*}
More generally, 
we denote by $p^{[k], \vec v}: \R^{n \times d} \mapsto (\R^n)^{\otimes k}$ the following directional derivative tensors:
$$
p^{[k], \vec v}_{i_1, \ldots, i_k}(\multix)
= D_{i_1, \vec v} D_{i_2, \vec v} \cdots D_{i_k, \vec v} p(\multix).
$$
When the unit vector $\vec v$ is clear from the context (typically inside proofs), we will just write $p^{[k]}$.

Recall the convention of multiplying two tensors of different dimensions in \Cref{sec:tensor}. 
Following that convention, 
we also note the following equivalence, which we will use sparingly in the paper:
$$
D_{i, \vec v} p^{[k], \vec v}(\multix) =  p^{[k+1], \vec v}(\multix) \vec e(i) ,
$$
where $\vec e(i)$ is the $i$-th standard basis vector.

\section{Proof of \Cref{thm:main}}

\subsection{Decay of Derivatives Restricted to a Random Direction}
\label{sec:derivative-decay}
In this subsection, we show that the 
$t$-th order directional partial derivatives of a low-degree polynomial $p: \R^{n \times d} \mapsto \R$ along a random direction $\vec v$ can be bounded from above by an exponentially decreasing function in $t$.  
The formal statement is given below.
\begin{lemma}[Derivative Decay]
\label{lem:derivative-decay}
For any $k \in \Z_+$ and $\eps \in (0,1)$, 
if $p:\R^{n \times d} \to \R$ is a degree-$k$ polynomial and $\vec v \sim \mathcal U( \mathcal S^{d-1} )$,
then the following holds with probability $0.99$ over the randomness of~$\vec v$:
\begin{align}
\label{eq:derivative-decay}
\Pr_{ \multix \sim \normal(\vec 0,\vec I)  }
\lp[ \forall t \in [k]: \;
\| p^{[t], \vec v}( \multix ) \|_F  \leq  (k/\eps)^{4t}  d^{-t/4}  |p(\multix)|
\rp]   \geq 1 - \eps. 
\end{align}
\end{lemma}
Our starting point is the so called 
strong anti-concentration properties of Gaussian distributions used frequently in the PRG literature.
Specifically, it states that the output of a polynomial is usually not too small compared to its derivative.
As a simple corollary of this property, we obtain another interesting property of polynomials that come in handy in showing \Cref{lem:derivative-decay}.
That is, the sizes of the higher order derivative tensor of a polynomial grow rather slowly (notably, the growth rate is independent of the input dimension).
\begin{fact}[Slow Growth of Derivatives; Lemma 1.6 in \cite{kelley2022random}]\label{fact:strong-anti-conc}%
For any $k \in \Z_+$ and $\eps \in (0,1)$, if
     $p:\R^{n \times d} \to \R$ is a degree-$k$ polynomial, it holds that
    \begin{align}
\label{eq:slow-growth}
\Pr_{ \vec x^{(1)}, \cdots , \vec x^{(n)} \sim \normal  }
\lp[ \exists t \in [k]: \quad
\|  \nabla^t p(\vec x^{(1)}, \cdots, \vec x^{(n)}) 
\|_F  >
\lp( k^{3t} / \eps^t \rp)
  |p(\multix)|
\rp]   \leq O(\eps).
\end{align}
\end{fact}
As argued in \Cref{sec:technical-overview}, a direct application of the strong anti-concentration property is not enough for our purposes. 
Instead, we need to leverage the specific structure of the distribution in the NGCA problem: the distribution $\ngca$ is non-Gaussian only along a single randomly chosen direction $\vec v$, and is the same as the standard Gaussian in every orthogonal direction.
Our main result in this subsection is that if we fix a low-degree polynomial $p$ and take $\vec v$ to be some random unit vector, the directional derivative tensor $p^{[k], \vec v}$ will instead be shrinking with high probability. Intuitively, this is because $ p^{[t], \vec v}( \multix )$ is just a 
\emph{random} sub-part of the entire tensor $\nabla^t p(\vec x_1, \cdots,  \vec x_n)$ that appeared in \Cref{fact:strong-anti-conc}.

\begin{proof}[Proof of \Cref{lem:derivative-decay}]
Fix some arbitrary $\multix \in \R^d$.
For convenience, define the tensor $\vec M :=\nabla^t p(\multix)$, and $\vec N:=p^{[t], \vec v}(\multix)$.
Recall that the full degree-$t$ gradient tensor
$\vec M $ lives in the space $( \R^{n \times d} )^{\otimes t}$. 
We can therefore label each entry of 
$\vec M$ by some indices
$i_1, j_1, \cdots, i_t, j_t$, where $i_1, \cdots i_t \in [n]$ and $j_1, \cdots j_t \in [d]$.
On the other hand, 
the projected derivative tensor $\vec N $ lives in the space $( \R^{n} )^{\otimes t}$. 
We can therefore label each entry of $\vec N$ by some indices $i_1, \ldots i_t$, where $i_1, \ldots i_t \in [n]$.
One can then verify using the definitions from \Cref{sec:derivative} that the entries of $\vec N$ and $\vec M$ satisfy the following relationship:
\begin{align*}
\vec N_{i_1, \cdots, i_t}
&= 
\sum_{j_1 \in [d]} \vec v_{j_1}
\frac{ \partial }{
\partial \vec x^{(i_1)}_{j_1} }
\sum_{j_2 \in [d]} \vec v_{j_2}
\frac{ \partial }{
\partial \vec x^{(i_2)}_{j_2} }
\cdots
\sum_{j_t \in [d]} \vec v_{j_t}
\frac{ \partial }{
\partial \vec x^{(i_t)}_{j_t} }
p (\multix) \\
&=
\sum_{j_1, \cdots ,j_t \in [d]} \vec v_{j_1}
\cdots \vec v_{j_t}
\; 
\frac{ \partial }{
\partial \vec x^{(i_1)}_{j_1}}
\cdots
\frac{ \partial }{
\partial \vec x^{(i_t)}_{j_t}}
p (\multix) \\
&=
\sum_{j_1, \ldots, j_t \in [d]} \vec v_{j_1}
\cdots \vec v_{j_t}
\; 
\vec M_{i_1, j_1, \ldots i_t, j_t}.
\end{align*}

\noindent Let $\vec v^{\otimes t}$ be the $t$-fold tensor product of $\vec v$ with itself, i.e., the tensor with entries:
$$
\left( \vec v^{\otimes t} \right)_{ j_1, \cdots, j_t  }
:=   \vec v_{j_1} \cdots \vec v_{j_t} \, ,
$$
and $\vec M^{(i_1, \cdots, i_t)} \in \lp(  \R^d \rp)^{\otimes t}$
be the sub-tensor of $\vec M$ of the form:
\begin{align}\label{eq:M-subtensor}
    \vec M^{(i_1, \cdots, i_t)}_{j_1, \cdots, j_t}
:= \vec M_{  i_1, j_1, \cdots, i_t, j_t }.
\end{align}

We can then write
\begin{align}
 \E_{ \vec v \sim \usphere } \lp[ \|  \vec N  \|_F^2 \rp]
&= 
\sum_{i_1, \cdots, i_t \in [n]}
\E_{ \vec v \sim \usphere }
\lp[ 
\lp( \sum_{j_1, \cdots j_t} \vec v_{j_1}
\cdots \vec v_{j_t} 
\; 
\vec M_{i_1, j_1, \cdots i_t, j_t} \rp)^2 \rp]
\nonumber \\
&= 
\sum_{i_1, \cdots, i_t \in [n]}
\E_{ \vec v \sim \usphere }
\lp[ 
  \left\langle  \vec v^{\otimes t}, \vec M^{(i_1, \cdots, i_t)} \right\rangle^2 \rp] \nonumber  \\
&= 
\sum_{i_1, \cdots, i_t \in [n]}
\lp( \vec M^{(i_1, \cdots, i_t)} \rp)^{\flat, \top}
\E_{ \vec v \sim \usphere } \lp[ 
 \left( \vec v^{\otimes t} \right)^{\flat}  \left( \vec v^{\otimes t}  \right)^{\flat, \top} \rp]
\lp( \vec M^{(i_1, \cdots, i_t)} \rp)^{\flat}.
\label{eq:vector-matrix-flat}
\end{align}

\noindent We claim the following bound on the matrix $  
 \E_{ \vec v \sim \usphere } \lp[ \left( \vec v^{\otimes t} \right)^{\flat}  \left( \vec v^{\otimes t}  \right)^{\flat, \top}   \rp]\in \R^{  d^t  \times d^t}$ that appeared earlier:
\begin{claim}\label{eq:W(v)-bound}
    Define $\vec W(\vec v):= 
 \E_{ \vec v \sim \usphere } \lp[ \left( \vec v^{\otimes t} \right)^{\flat}  \left( \vec v^{\otimes t}  \right)^{\flat, \top}   \rp]\in \R^{  d^t  \times d^t}$. 
 Then we have
 \begin{align*}
\| \vec W(\vec v) \|_F \leq    \lp(\frac{2t}{e \sqrt{d}} \rp)^{t}   \;.
\end{align*}
\end{claim}
We defer the proof of the above claim to the end of this subsection. Before that, we show how to conclude the proof of \Cref{lem:derivative-decay} using the claim.
Combining \Cref{eq:W(v)-bound,eq:vector-matrix-flat} gives the following (recall that $\vec M^{(i_1, \cdots, i_t)}$ denotes the tensor from  \Cref{eq:M-subtensor}):
\begin{align*}
\E_{ \vec v \sim \usphere } \lp[ \|  \vec N  \|_F^2 \rp]
&\leq 
\sum_{i_1, \cdots, i_t \in [n]}
\| \vec M^{(i_1, \cdots, i_t)}  \|_F^2
\|\vec W(\vec v) \|_F 
\\
&\leq  (2t/e)^{t} d^{-t/2} \sum_{i_1, \cdots, i_t \in [n]} \| \vec M^{(i_1, \cdots, i_t)}  \|_F^2 
\tag{\Cref{eq:W(v)-bound}}
\\
&=  (2t/e)^{t} d^{-t/2} \| \vec M \|_F^2.
\end{align*}
By Jensen's inequality, we have that
\begin{align}
\label{eq:N-M-relation}
\E_{\vec v \sim \usphere} \lp[ \|  \vec N  \|_F  \rp]
\leq \lp( 2t/e \rp)^{t/2} d^{-t/4} \| \vec M \|_F.
\end{align}
Recall that at the beginning of the proof we fixed some arbitrary $\multix$ and defined the quantities $\vec N: = p^{[t], \vec v}\lp( \multix \rp), \vec M:= \nabla^t p\lp( \multix \rp)$ accordingly.
Therefore, \Cref{eq:N-M-relation} implies that
\begin{align}\label{eq:intermediate}
\E_{\vec v \sim \usphere} \lp[ \| p^{[t], \vec v}\lp( \multix \rp)  \|_F  \rp]
\leq \lp( 2t/e \rp)^{t/2} d^{-t/4} \| \nabla^t p\lp( \multix \rp) \|_F.
\end{align}
Define the event $\mathcal{E}_{\multix} = \left\{ \|p^{[t], \vec v}\lp( \multix \rp)  \|_F \neq 0, \nabla^t p\lp( \multix \rp) \neq 0 \right\}$ with respect to a sequence of points $\multix$. 
Then we have that
\begin{align*}
&\E_{ \vec v \sim \usphere  }
\lp[ 
\E_{ \multix \sim \normal(\vec 0, \vec I) } \lp[ 
\frac{
\| p^{[t], \vec v}\lp( \multix \rp)  \|_F}{
\| \nabla^t p\lp(  \multix \rp) \|_F  }  \Biggm\vert \mathcal{E}_{\multix}
\rp]
\rp] 
\\
&= \E_{ \multix \sim \normal(\vec 0, \vec I) }
\lp[ 
\E_{\vec v \sim \usphere} \lp[ 
\frac{
\| p^{[t], \vec v}\lp( \multix \rp)  \|_F}{
\| \nabla^t p\lp(  \multix \rp)  \|_F} \Biggm\vert \mathcal{E}_{\multix}
\rp] \rp] \\
&\leq \lp( 2t/e \rp)^{t/2} d^{-t/4}. \tag{using \Cref{eq:intermediate} and the definition of $\mathcal{E}_{\multix}$}
\end{align*}
Applying Markov's inequality on $\vec v$ then gives that
\begin{align}
\label{eq:x-expectation-bound}
\E_{ \multix \sim \normal(\vec 0, \vec I) } \lp[ 
\frac{
\| p^{[t], \vec v}\lp( \multix \rp)  \|_F}{
\| \nabla^t p\lp(  \multix \rp) \|_F  } \Biggm\vert \mathcal{E}_{\multix}
\rp]
\leq O(1) \; \lp( 2t/e \rp)^{t/2} d^{-t/4}
\end{align}
with high constant probability.
In the remaining analysis, we condition on some $\vec v$ such that \Cref{eq:x-expectation-bound} holds. 
Note that 
in the case that the event $\mathcal{E}_{\multix}$ does not hold
then we still have
$\lp \| p^{[t], \vec v}\lp( \multix \rp)  \rp \|_F 
\leq O \lp( \lp( 2t/e \rp)^{t/2} d^{-t/4} \rp) \; \lp \| \nabla^t p\lp(  \multix \rp) \rp \|_F$ as both sides of the inequality are zero.
We thus have that
\begin{align*}
&\Pr_{ \multix \sim \normal(\vec 0, \vec I) }
\lp[ 
\lp \| p^{[t], \vec v}\lp( \multix \rp)  \rp \|_F 
\leq O \lp( \eps^{-1} \lp( 2t/e \rp)^{t/2} d^{-t/4} \rp) \; \lp \| \nabla^t p\lp(  \multix \rp) \rp \|_F
\rp]\\
&\geq 
\Pr_{ \multix \sim \normal(\vec 0, \vec I) }
\lp[ 
\lp\| p^{[t], \vec v}\lp( \multix \rp)  \rp\|_F 
\leq O \lp( \eps^{-1} \lp( 2t/e \rp)^{t/2} d^{-t/4} \rp) \; \lp \| \nabla^t p\lp(  \multix \rp) \rp \|_F \Biggm\vert \mathcal{E}_{\multix}
\rp]\\
&\geq 1 - \eps / 2,
\end{align*}
where we used that the event $\mathcal{E}_{\multix}$ happens with probability 1 (this is because the complement of the event amounts to $\multix$ being the exactly equal to the roots of a polynomial), and Markov's inequality.

By a union bound over $t \in [k]$, we further get that
\begin{align*}
    &\Pr_{ \multix \sim \normal(\vec 0, \vec I) }
\lp[\forall  t \in [k]:
\lp \| p^{[t], \vec v}\lp( \multix \rp)  \rp \|_F 
\leq O \lp( \eps^{-1} k \lp( 2t/e \rp)^{t/2} d^{-t/4} \rp) \; \lp \| \nabla^t p\lp(  \multix \rp) \rp \|_F
\rp]\\
&\geq 1 - \eps/2 .
\end{align*}

Combining this with \Cref{eq:slow-growth} and the union bound
then gives that
\begin{align*}
&\Pr_{ \multix \sim \normal(\vec 0, \vec I) }
\lp[ \forall  t \in [k]:
\lp\| p^{[t], \vec v}\lp( \multix \rp)  \rp\|_F 
\leq O \lp( \eps^{-1} k \lp( 2t/e \rp)^{t/2} d^{-t/4} \rp) \;
k^{3t} (2/\eps)^t
\lp|  p\lp(  \multix \rp) \rp|
\rp]\\
&\geq 1 - \eps.    
\end{align*}
We can finally simplify the expression $O \lp( \eps^{-1} k \lp( 2t/e \rp)^{t/2} d^{-t/4} \rp)$ that appears above as follows:
\begin{align*}
    O \lp( \eps^{-1} k \lp( 2t/e \rp)^{t/2} d^{-t/4} \rp) \;
k^{3t} (2/\eps)^t \leq (k/\eps)^{4t}  d^{-t/4}
\end{align*}
which concludes
the proof of \Cref{lem:derivative-decay}.
\end{proof}

We conclude this section by showing \Cref{eq:W(v)-bound}:

\begin{proof}[Proof of \Cref{eq:W(v)-bound}]
To show that,  we will first relate the expected value of $\vec W(\vec v)$ under $\vec v \sim \usphere$ to that under $\vec v \sim \normal(\vec 0, \vec I)$.
In particular, denote by $g: \R \mapsto \R$ the probability density function of the random variable $\| \mathbf v  \|_2$, where $ \vec v \sim \normal(\vec 0, \vec I)$.
Then we have that
\begin{align*}
\lp \| \E_{ \vec v \sim \normal(\vec 0, \vec I) }
\lp[ \vec V(\vec v)^\flat \vec V(\vec v)^{\flat, \top}  
\rp] \rp \|_F
&=
\lp\|
\int\limits_{0}^{+\infty}
\E_{ \vec v \sim \normal(\vec 0, \vec I) }
\lp[ \vec V(\vec v)^\flat \vec V(\vec v)^{\flat, \top}  \bigm\vert \|  \vec v \|_2 = b
\rp] g(\vec v) \, \d b  
\rp \|_F
\\
&=
\lp \|
\E_{ \vec v \sim \usphere }
\lp[ \vec V(\vec v)^\flat \vec V(\vec v)^{\flat, \top} \rp] 
\int\limits_{0}^{+\infty}
b^{2t} g(\vec v) \, \d b \rp \|_F\\
&\geq
\lp \|
\E_{ \vec v \sim \usphere }
\lp[ \vec V(\vec v)^\flat \vec V(\vec v)^{\flat, \top} \rp] 
\rp \|_F
d^{t} 
=  \| \vec W(\vec v) \|_F \, d^t \,
\end{align*}
where the last line used that $\int_0^{+\infty} b^{2t} g(b) \, \d b = \E_{\vec v \sim \normal(0,\vec I)}[\|\vec v\|^{2t}] \geq \E_{\vec v \sim \normal(0,\vec I)}[\|\vec v\|^{2}]^t = d^t$.
Rearranging this gives the following upper bound on the Frobenius norm of $\vec W(\vec v)$:

\begin{align}
\| \vec W(\vec v) \|_F= 
\lp \| \E_{\vec v \sim \usphere }  \lp[
 \vec V(\vec v)^{\flat} \vec V(\vec v)^{\flat, \top}
\rp] \rp \|_F
&\leq
d^{-t}
\lp \| 
\E_{\vec v \sim \normal(0, \vec I) }  \lp[
 \vec V(\vec v)^{\flat} \vec V(\vec v)^{\flat, \top}
\rp] \rp \|_F \nonumber \\
&= d^{-t}
\sqrt{
\sum_{  j_1, \cdots, j_{2t} \in [d]}
 \left( \E_{\vec v \sim \normal(0,\vec I)}\lp[  \vec v_{j_1} \cdots \vec v_{j_{2t} }    \rp] \right)^2 . \label{eq:beforeIserrlis}
} 
\end{align}
We can further bound the quantity $\E^2_{\vec v \sim \normal(\vec 0, \vec I)}\lp[  \vec v_{j_1} \cdots \vec v_{j_{2t} }    \rp]$ that appears in the right hand side above using  Iserrlis' theorem (\Cref{fact:iserrlis}) as follows ($P_k^2$ below denotes the set of all matchings among $\{1,\ldots,k\}$):
\begin{align}
    \sum_{j_1, \cdots, j_{2t} \in [d]} \left( \E\lp[  \vec v_{j_1}  \cdots \vec v_{j_{2t} }    \rp] \right)^2
    &=  \sum_{j_1, \cdots, j_{2t} \in [d]}\lp( \sum_{p \in P_{2t}^2} \prod_{ \{ k, \ell \} \in p } \E[ \vec v_{j_k} \vec v_{j_\ell}] \rp)^2 \nonumber \\
    &\leq | P_{2t}^2 |  \sum_{j_1, \cdots, j_{2t} \in [d]} \; \sum_{p \in P_{2t}^2}\; \prod_{ \{ k, \ell \} \in p} \lp( \E[ \vec v_{j_k} \vec v_{j_\ell}]\rp)^2 \nonumber \\
    &=  | P_{2t}^2 |  \sum_{p \in P_{2t}^2} \;\prod_{ \{ k, \ell \} \in p} \; \sum_{j_k, j_\ell \in [d]} \lp( \E[ \vec v_{j_k} \vec v_{j_\ell}]\rp)^2 \nonumber \\
    &= | P_{2t}^2 |  \sum_{p \in P_{2t}^2} d^t  \nonumber \\
    &=  \, \lp( (2t-1)!! \rp)^2 \, d^t \leq (2t/e)^{2t} d^t 
    \label{eq:afterIserrlis}\, , 
\end{align}
where the first line is an application of Iserrlis' theorem (\Cref{fact:iserrlis}), the second line  
uses the inequality $2 ab \leq a^2 + b^2$, 
the fourth line uses that $\E[ \vec v_{j_k} \vec v_{j_\ell}] = \1(j_k = j_\ell)$ and $|p|$ (the number of pairs within the matching) is $t$, and the last line uses that the number of all possible matchings over $[2t]$ is $(2t-1)!! < (2t)!! = 2^t t! < 2^t (t/e)^t = (2t/e)^t$. 
\noindent Combining \Cref{eq:beforeIserrlis} and \Cref{eq:afterIserrlis} then gives that $\| \vec W(\vec v) \|_F  \leq  d^{-t/2}  (2t/e)^{t}$, concluding the proof of \Cref{eq:W(v)-bound}.
\end{proof}

\subsection{Framework: Mollification, Sandwiching, and Hybrid Argument}
\label{sec:framework}
In this subsection, we lay out the high level proof strategy for our main theorem based on the ideas of mollification and the hybrid argument.

To begin with, we need a smooth function $\rho: \R \mapsto [0, 1]$ satisfying the following conditions:
\begin{align}
\label{eq:rho-property}
\rho(x) = 1 \text{ if } |x| < 1 \, ,
\rho(x) = 0 \text{ if } |x| \geq 3 \, ,
\| \rho^{(t)}(x) \|_{\infty} \leq O\lp( t^{t} \rp).
\end{align}
There are standard ways to construct such a function, deferred to \Cref{lem:smooth-function-construct} in \Cref{ap:sign}.
We then use it to define the following mollifier function $g$:
\begin{align}
\label{eq:mollifier-def}
g(\multix) = \prod_{t=1}^k \rho \lp(   
\frac{d^{ c_g t} \lp\| p^{[t]}( \multix ) \rp \|_F^2 }{p^2(\multix  )} \rp) \, ,
\end{align}
where $c_g \in (0, 1/2)$ is some constant that we will specify later. 
Intuitively, $g$ is constructed such that
if some points $\vec x^{(1:n)}$ satisfy the derivative decay condition  from \Cref{lem:derivative-decay},  then $g(\multix)$ should evaluate to $1$.
Conversely, if $g(\multix)$ evaluates to $1$,
we can infer from its definition that the weaker derivative decay condition 
$ 
\lp\| p^{[t]}( \multix ) \rp \|_F^2
\leq 3  d^{ -c_g t}  p^2(\multix  )
$ must hold for the input points.

Finally, the mollified version of the PTF is the following function:
\begin{align}
\label{eq:mollified-ptf}
h(\multix) := \sgn( p(\multix) ) g(\multix).    
\end{align}
Thanks to \Cref{lem:derivative-decay} and our construction of $g$,
we can show that $h(\multix) $ is a good approximation of $\sgn( p(\multix) )$ under the Gaussian distribution.
Moreover, since $g$ is at most $1$,
we note that $h(\cdot)$ is bounded from above by $\sgn(p(\cdot))$ pointwise. 
Combining these two observations with a sandwiching argument allows us to show the following: if $\ngca$ fools the mollified PTF $ h $ with respect to the Gaussian distribution, then $\ngca$ also fools the original PTF. 
\begin{lemma}[Sandwiching]
\label{lem:sandwich} 
Let $p: \R^{n \times d} \mapsto \R$ be a degree-$k$ polynomial, $\vec v \in \R^d$ be some vector satisfying \Cref{eq:derivative-decay} with $\eps = 0.05$, 
and $h:\R^{n \times d} \mapsto \R$ 
be the mollified PTF of $p$ defined as in \Cref{eq:mollified-ptf}.
Assume that $(k/0.05) < d^{ (1/4 - \cg/2)/4}$, where $\cg$ is the parameter used in \Cref{eq:mollifier-def}.
The following statement holds:
If 
\begin{align}
\label{eq:fool-mollify}
\lp| 
\E_{ \multiy \sim \ngcav } 
\lp[    h(\bmultiy)  \rp] 
- \E_{ \multix \sim \normal(\vec 0, \vec I) }
\lp[  h(\bmultix)   \rp]\rp| \leq \delta \, ,     
\end{align}
then it holds that
\begin{align}
\label{eq:fool-mollify2}
\lp|  \E_{ \multiy \sim \ngcav } \lp[ \sgn\lp( p ( \bmultiy ) \rp)\rp] 
- \E_{ \multix \sim \normal(\vec 0, \vec I)  }
\lp[  \sgn\lp( p ( \bmultix ) \rp)   \rp]\rp|
\leq \delta + 0.05.
\end{align}
\end{lemma}
\begin{proof}
For $\multix \in \R^d$ define the event $$\mathcal{E}_{\vec v} (\bmultix) = \lp\{\forall t \in [k] \;\;
\| p^{[t], \vec v}( \multix ) \|_F  \leq  d^{-\cg t/2}  |p(\multix)| \rp\}.$$

Fix $\eps = 0.05$.
Recall that we assume $\vec v$ is chosen such that the derivative decay condition in \Cref{eq:derivative-decay} holds.
For convenience, we restate the condition below.
\begin{align}\label{eq:probbound}
    \Pr_{ \multix \sim \normal(\vec 0, \vec I)   }
\lp[ \forall t \in [k] \;\;
\| p^{[t], \vec v}( \multix ) \|_F  \leq (k/\eps)^{4t}  d^{-t/4} |p(\multix)|
\rp]   \geq 1 - \eps.
\end{align}
Using our assumption $k/\eps = (k/0.05) < d^{(1/4-\cg/2)/4}$, we have that $(k/\eps)^{4t}d^{-t/4} < d^{-\cg t/2}$. Thus a simplified version of \Cref{eq:probbound} holds, and in particular implies that:
\begin{align}\label{eq:derivatives-event}
  \Pr_{ \multix \sim \normal(\vec 0, \vec I)   }
\lp[ \mathcal{E}_{\vec v}(\bmultix)
\rp]   \geq 1 - \eps.  
\end{align}

We will first show that $|\E_{\multix \sim \normal(\vec 0, \vec I) }[h(\bmultix)] - \E_{\multix \sim \normal(\vec 0, \vec I) }[\sign (p(\bmultix))]| \leq \eps$ for any degree-$k$ polynomial $p: \R^{n \times d} \mapsto \R$. 
Then, we will use that to show \Cref{eq:fool-mollify2}. 
To see the first claim, using the definition of $g(\bmultix)$ (\Cref{eq:mollifier-def}) and \Cref{eq:derivatives-event}
we have that
\begin{align}
    &\lp| \E_{\multix \sim \normal(\vec 0, \vec I) }[h(\bmultix)] - \hspace{-13pt}\E_{\multix \sim \normal(\vec 0, \vec I) }[\sign (p(\bmultix))] \rp| \\
    &= \lp|  \E_{\multix \sim \normal(\vec 0, \vec I) }\lp[(g(\bmultix) - 1)\sign (p(\bmultix))\rp] \rp| \notag \\
    &= \lp|  \E_{\multix \sim \normal(\vec 0, \vec I) }\lp[(g(\bmultix) - 1)\sign (p(\bmultix)) \cdot \1\lp(\mathcal E_{\vec v}^c(\bmultix)\rp) \rp] \rp| \label{eq:indicator}\\
    &\leq \Pr_{\multix \sim \normal(\vec 0, \vec I) }\lp[\mathcal E_{\vec v}^c(\bmultix)\rp] \leq \eps, \label{eq:firstclaim}
\end{align}
where \Cref{eq:indicator} used that $g(\bmultix)$ can be different than $1$ only under the complement of the event $\mathcal E_{\vec v}(\bmultix)$.

We can now show \Cref{eq:fool-mollify2}, which states that the two expectations (under $\normal(\vec 0, \vec I) $ and $\ngca$) of $\sign(p(\multix))$ are close to each other in absolute value. We will show both a lower bound and an upper bound on their difference, which together yield the bound in absolute value. We start with the lower bound. Using that $g(\bmultix) \leq 1$ and $\sign ( p ( \bmultix ) ) \in \{0,1\}$, we have the pointwise relationship $\sign \lp( p \lp( \bmultix \rp) \rp) \geq \sign \lp( p \lp( \bmultix \rp) \rp)g(\bmultix)$. 
In particular, this implies that:
\begin{align*}
    \E_{\multiy \sim \ngca}&\lp [\sign \lp( p \lp( \bmultiy \rp) \rp) \rp] \\
    &\geq \E_{\multiy \sim \ngca}\lp [  \lp( \sign \lp( p \lp( \bmultiy \rp) \rp)  \rp) g(\bmultix)  \rp]  \\
    &\geq \E_{\multix \sim \normal(\vec 0, \vec I) }\lp [  \sign \lp( p \lp( \bmultix \rp) \rp)  g(\bmultix)  \rp] - \delta \tag{by \Cref{eq:fool-mollify}}\\
    &\geq \E_{\multix \sim \normal(\vec 0, \vec I) }\lp [  \sign \lp( p \lp( \bmultix \rp) \rp)    \rp]  - \delta - \eps  \tag{by \Cref{eq:firstclaim}}.
\end{align*}

We can prove the other direction $\E_{}[\sign (p(\bmultiy))] \leq \E_{}[\sign (p(\bmultix))] +  \delta + \eps$ by repeating the same argument with $-p$ in place of $p$.
This concludes the proof of \Cref{lem:sandwich}.
\end{proof}
Given the above sandwiching lemma, it then suffices for us to bound from above the difference
$|\E_{\bmultiy) \sim \ngca}[h(\bmultiy)] - \E_{\bmultix \sim \normal(\vec 0, \vec I)}[h(\bmultix)]|$.
We will show this via the hybrid argument.
In particular, let $\bmultix$ be \iid samples from the standard Gaussian $\normal(\vec 0, \vec I)$ and $\bmultiy$ be \iid samples from the hidden direction distribution $\ngca$.
In the $i$-th replacement step of the hybrid argument, we compare the expected values of 
$h(\vec x^{(1)}, \cdots, \vec x^{(i)}, \vec y^{(i+1)}, \cdots, \vec y^{(n)})$ and
$h(\vec x^{(1)}, \cdots, \vec y^{(i)}, \vec y^{(i+1)}, \cdots, \vec y^{(n)})$.
We show that the difference between the expected values is on the order of $d^{ - \cg m / 2  }$.

\begin{restatable}[Replacement Step]{proposition}{MAINPROPOSITION}
\label{prop:replacement-step}
For any $c \in (0, 1/4)$, $d,m,k \in \Z_+$ such that 
 $m$ is even, and $d > \max( m^{C/c}, k^{C /c} )$ 
 for some sufficiently large constant $C$, 

if $p: \R^{n \times d} \mapsto \R$ is a degree-$k$ polynomial,  $\vec v \in \R^d$ is a unit vector satisfying \Cref{eq:derivative-decay} with $\eps = 0.05$, $A$ is a one-dimensional distribution that matches the first $m$ moments with $\normal(0, 1)$, and
$h:\R^{n \times d} \mapsto \R$ is the mollified PTF from \Cref{eq:mollified-ptf},
then the following holds: For every $i \in [n]$
\begin{align}
\label{eq:mollified-replacement-step}
\lp| \E  
\lp[  h( \vec x^{(1)},
\cdots , \vec x^{(i)}, \vec y^{(i+1)}, \cdots , \vec y^{(n)} \rp]
-  
\E  
\lp[  h( \vec x^{(1)},
\cdots , \vec y^{(i)}, \vec y^{(i+1)}, \cdots , \vec y^{(n)} \rp] \rp|
\leq   d^{ - \cg m /2 +  c m   } \, ,
\end{align}
where $\cg$ is the parameter used in \Cref{eq:mollified-ptf},
$\vec x^{(1)}, \cdots, \vec x^{(n)} \sim \ngca$, and $\vec y^{(1)}, \cdots, \vec y^{(n)} \sim \normal(\vec 0,\vec I)$.
\end{restatable}
Our main result follows immediately from 
\Cref{lem:sandwich} and \Cref{prop:replacement-step}.
\begin{proof}[Proof of \Cref{thm:main}]
Let $c^*$ be the constant parameter in the statement of \Cref{thm:main}.
Fix $\eps = 0.05$, $c = c^*/2 $, and $\cg =2 \; \lp( 1/4 - c^*/ 2\rp) $.
By \Cref{lem:derivative-decay}, we have that a randomly chosen $\vec v$ satisfies \Cref{eq:derivative-decay} with high probability $0.99$. We will condition on such a $\vec v$ in the rest of the proof.
First, combining \Cref{prop:replacement-step} (the proposition is applicable because of our assumptions $
d >  \max\lp(k^{C^* / c^*}, m^{C^*/c^*}\rp)
$ in the statement of \Cref{thm:main} and $c = c^* / 2$) for each position $i \in [n]$ with the triangle inequality yields that
\begin{align*}
\lp| 
\E_{ \multiy \sim \ngcav } 
\lp[    h(\bmultiy)  \rp] 
- \E_{ \multix \sim \normal(\vec 0, \vec I) }
\lp[  h(\bmultix)   \rp]\rp| \leq n d^{ - \cg m /2 + c m   }.
\end{align*}
Since we have $\cg = 2(1/4 - c^*/2)$ and $c = c^*/2$, 
the right hand side can be further bounded from above by $n d^{ -(1/4 - c^*) m}$, which is at most $0.05$ by our assumption that $n \ll d^{ (1/4 - c^*) m}$.
We can then apply the sandwiching lemma (\Cref{lem:sandwich}) with $\delta = 0.05$.\footnote{
The lemma is applicable since the assumption $k < d^{ c^*/C^* }$ from \Cref{thm:main} implies that $ (k/0.05) < d^{  (1/4 - \cg/2) / 4 } $ as long as 
$\cg = 2(1/4 - c^*/2)$ and $C^*$ is sufficiently large.
}
This yields that
\begin{align}\label{eq:almost}
 \lp|  \E_{ \multiy \sim \ngcav } \lp[ \sgn\lp( p ( \bmultiy ) \rp)\rp] 
-\hspace{-10pt} \E_{ \multix \sim \normal }
\lp[  \sgn\lp( p ( \bmultix ) \rp)   \rp]\rp|
\leq \delta + \eps =  0.1 .   
\end{align}

So far we have shown that \Cref{eq:almost} holds with probability $0.99$. From this, we can complete the proof as follows. First, denote by $\cE$ the event in \Cref{eq:almost}. 
We have the following by Jensen's inequality:
\begin{align}
    &\lp| \E_{\substack{\vec v \sim  \usphere\\ \multix \sim \ngcav}}
\lp[  \sgn(p(\vec x^{(1)}, \cdots, \vec x^{(n)})) \rp]
-  
\E_{ \multix \sim \normal(\vec 0, \vec I)}    
\lp[  \sgn(p(\vec x^{(1)}, \cdots, \vec x^{(n)})) \rp] \rp| \\
&\leq \E_{\vec v \sim  \usphere} \left[ \lp|  \E_{ \multiy \sim \ngcav } \lp[ \sgn\lp( p ( \bmultiy ) \rp)\rp] 
-\hspace{-10pt} \E_{ \multix \sim \normal(\vec 0, \vec I)  }
\lp[  \sgn\lp( p ( \bmultix ) \rp)   \rp]\rp| \right].\label{eq:almost2}
\end{align}
Now let the random variable $$Z = \lp|  \E_{ \multiy \sim \ngcav } \lp[ \sgn\lp( p ( \bmultiy ) \rp)\rp] 
-\hspace{-10pt} \E_{ \multix \sim \normal(\vec 0, \vec I)  }
\lp[  \sgn\lp( p ( \bmultix ) \rp)   \rp]\rp|$$ for brevity. We can further bound the RHS in  \Cref{eq:almost2} as follows:
\begin{align*}
    \E_{\vec v \sim  \usphere}[Z]
    &=\E_{\vec v \sim  \usphere}[Z \1(\cE)] + \E_{\vec v \sim  \usphere}[Z \1(\cE^c)] \\
    &\leq 0.1 + \E_{\vec v \sim  \usphere}[ \1(\cE^c)]\\
    &\leq 0.1 + 0.01 = 1.01 \;,
\end{align*}
where we used that under the event $\cE$ we have $Z \leq 0.1$, we also used that $Z$ is always at most $1$ since it is a difference of PTFs that take values in $\{0,1\}$ and we used that $\cE$ happens with probability at least $0.99$ over the choice of $\vec v$.

This concludes the proof of \Cref{thm:main}.
\end{proof}

\subsection{Proof of \Cref{prop:replacement-step}}

In the rest of this section, we focus on establishing a single replacement step by proving \Cref{prop:replacement-step}.
This subsection is organized as follows. In \Cref{sec:condition-and-truncation}, we introduce a truncation procedure that 
reduces the problem to the case where the non-Gaussian component $A$ has bounded support, which will be helpful later on, and argue that the truncation  will not significantly change the expectation of the mollified PTF. In \Cref{sec:control-derivatives}, we show that if the derivative of the polynomial is small at any point in the truncated interval, it is small throughout the entire interval. In \Cref{sec:taylor}, we perform a Taylor expansion of the mollified PTF, derive expressions for the derivatives in the higher-order terms, and upper bound them using the results from \Cref{sec:control-derivatives}. Finally, \Cref{sec:wrapup} concludes the proof of \Cref{prop:replacement-step} by carefully combining the results from the previous sections.

\subsubsection{Domain Truncation}
\label{sec:condition-and-truncation}
Let $A$ be a distribution on $\R$ which matches the first $m$ moments with $\normal(0, 1)$.
Note that the hidden direction distribution $\ngca$ from \Cref{def:hidden_distr} can be viewed of as the sum $\bar {\vec y} + \xi \vec v$, where
$\bar {\vec y} \sim \lnormal$ and $\xi \sim A$.
Hence, \Cref{eq:mollified-replacement-step} can be alternatively written as
\begin{align}
\label{eq:orthogonal-deocmpose}
\lp| \E  
\lp[  h( \vec x^{(1)},
\cdots , \vec x^{(i)}, \bar {\vec y} + \xi \vec v, \cdots  , \vec y^{(n)}) \rp]
-  
\E  
\lp[  h( \vec x^{(1)},
\cdots , \vec x^{(i)}, \bar {\vec y} + z \vec v, \cdots , \vec y^{(n)}) \rp] \rp|
\leq d^{ - \cg m /2 + cm } \, ,
\end{align}
where $\xi \sim A,z \sim \normal(0, 1),\bar {\vec y} \sim \lnormal,\vec x^{(1)}, \ldots, \vec x^{(i)} {\sim} \ngca$, and $\vec y^{(i+1)}, \cdots, \vec y^{(n)} {\sim} \normal(\vec 0, \vec I)$.

To show \Cref{eq:orthogonal-deocmpose}, we will condition on fixed values for all  random variables except from $\xi$ and $z$
and show that
\begin{align}
\label{eq:Y-fool-h}
\lp| 
\E_{\xi \sim A}  
\lp[  h( \vec x^{(1)},
\cdots , \vec x^{(i)}, \bar {\vec y} + \xi \vec v, \cdots , \vec y^{(n)}) \rp]
-  
\E_{z \sim \normal(0, 1)}  
\lp[  h( \vec x^{(1)},
\cdots , \vec x^{(i)}, \bar {\vec y} + z \vec v, \cdots , \vec y^{(n)}) \rp] \rp|
\leq d^{ - \cg m /2 + cm }.
\end{align}
\Cref{eq:orthogonal-deocmpose} will then follow from \Cref{eq:Y-fool-h} via an averaging argument.
It turns out that the arguments of the rest of the subsection (in particular the part that analyzes the Taylor expansion in \Cref{sec:taylor}) will be easier if $\xi$ and $z$ are bounded random variables. 
Fortunately, the two are both concentrated around $0$, and we can therefore truncate them to the interval $[- d^{- \ct }, d^{\ct}  ]$ for some constant $\ct$.
In particular, we argue that after the truncation we still have an approximate moment matching condition, and the mass of the non-Gaussian component $A$ outside the truncated interval cannot be too large.

\begin{lemma}[Domain Truncation]\label{lem:truncation}
For any $d,m \in \mathbbm Z_+$, where $m$ is even and $d \gg m$, and every $\ct > 0$, 
if $A$ is a distribution on $\R$ which matches the first $m$ moments with $\normal(0, 1)$, 
then for all $t \in [ m-1 ]$ it holds that
\begin{align}\label{eq:truncation1}
    \lp| \E_{x \sim  A}[ x^t \mid x \in [-d^{\ct}, d^{\ct}]]
    - 
    \E_{y \sim \normal(0, 1)}[ y^t 
    \mid y \in [-d^{\ct}, d^{\ct}]
    ] 
    \rp| \leq  m^{O(t)} d^{- \ct (m-t) }.
\end{align}
Moreover, it holds
\begin{align}
    \Pr_{x \sim A}[x \not\in[-d^{\ct}, d^{\ct}] ] \leq 
    m^m 
    d^{- \ct m}.\label{eq:truncation2}
\end{align}
\end{lemma}
\begin{proof}[Proof of \Cref{lem:truncation}] 
We start with \Cref{eq:truncation1}.
Let us first denote $I = [-d^{\ct}, d^{\ct}]$ and $p := \Pr_{x \sim A}[x \not \in I]$. 
We can write $$\E_{x \sim A}[x^t | x \in I] = \frac{\E_{x \sim A}[x^t \1(x \in I)]}{1-p} =  \frac{\E_{x \sim A}[x^t]}{1-p} -  \frac{\E_{x \sim A}[x^t \1(x \not \in I)]}{1-p}.$$ Rearranging, we have that
\begin{align} \label{eq:moment_error}
\lp| \E_{x \sim A}[x^t | x \in I]  - \E_{x \sim A}[x^t]   \rp| \leq p \lp| \E_{x \sim A}[x^t | x \in I ]  \rp| + \lp| \E_{x \sim A}[x^t  \1(x \not \in I)]  \rp|.
\end{align}
We can upper bound $p$ using the higher-order Chebyshev's inequality and the moment matching property of $A$:
\begin{align}
    p := \Pr_{x \sim A}[ |x| > d^{\ct}] \leq \frac{\E_{x \sim A}[|x|^m]}{d^{\ct m}} 
    = \frac{\E_{x \sim \normal(0,1)}[|x|^m]}{d^{\ct m}}
    \leq \frac{m^{m/2}}{d^{\ct m}} \label{eq:temp43}
\end{align}
 Using this bound on $p$, the first term in the RHS of  \Cref{eq:moment_error} is 
\begin{align}
p \lp| \E_{x \sim A}[x^t | x \in I ]  \rp|
    &\leq p \cdot  \E_{x \sim A}[|x|^t | x \in I ] \leq \frac{p}{1-p}\E_{x \sim A}[|x|^t] \nonumber \\
    &= \frac{p}{1-p}\E_{x \sim \normal(0,1)}[|x|^t]
    \leq m^{O(m)} d^{- \ct m},
    \label{eq:rhs-term1}
\end{align}
where we used the definition of conditional expectation, the moment matching property of $A$ and that $d \gg m$.
The second term in the RHS of \Cref{eq:moment_error} can be upper bounded using Holder's inequality as follows:
\begin{align}
    \lp| \E_{x \sim A}[x^t  \1(x \not \in I)]  \rp| 
    &\leq 
    \lp( \E_{ x \sim A } \lp[  x^m \rp] \rp)^{t/m}
    p^{1 - t/m} \nonumber \\
    &=  \lp( \E_{ x \sim \normal(0, 1) } \lp[  x^m \rp] \rp)^{t/m} p^{1 - t/m} \nonumber \\
    &\leq  m^{O(t)} d^{- \ct \; (m-t) }.
    \label{eq:rhs-term2}
\end{align}
Combining \Cref{eq:rhs-term1,eq:rhs-term2} then gives that $\lp| \E_{x \sim A}[x^t | x \in I]  - \E_{x \sim A}[x^t]   \rp| \leq 
m^{O(t)} d^{- \ct (m-t)}
$. 
Repeating the same steps (using $\normal(0,1)$ in place of $A$), it can be shown that 
$$
\lp| \E_{y \sim \normal(0,1)}[y^t | y \in I]  - \E_{y \sim \normal(0,1)}[y^t]   \rp| \leq m^{O(t)} d^{- \ct (m-t)}
$$ as well. 
The rest of the proof of \Cref{eq:truncation1} then follows from 
the triangle inequality and the fact that $\E_{y \sim \normal(0,1)}[y^t] = \E_{x \sim A}[x^t]$.
Finally, we note that \Cref{eq:truncation2} has already been shown in \Cref{eq:temp43}.
This completes the proof of \Cref{lem:truncation}.
\end{proof}

In the rest of the subsection, we focus in bounding the difference in expected values under the truncated distributions. 
In particular, our goal is to show that
\begin{align}
\label{eq:truncate-Y-fool-h}
\lp| 
\E_{\xi \sim \bar A}  
\lp[  h( \vec x^{(1)},
\cdots, \vec x^{(i)}, \bar {\vec y} + \xi \vec v, \cdots, \vec y^{(n)}) 
\rp]
-  
\E_{z \sim \bar \normal(0, 1)}  
\lp[  h( \vec x^{(1)},
\cdots, \vec x^{(i)}, \bar {\vec y} + z \vec v, \cdots, \vec y^{(n)}) 
\rp] \rp|
\leq d^{ - \cg m + c/2} \, ,
\end{align}
where $\bar A$ and $\bar \normal(0, 1)$ correspond to the distributions $A$ and $\normal(0, 1)$ conditioned on the domain $[ -d^{\ct}, d^{\ct} ]$ respectively.

\subsubsection{Controlling Derivatives of Nearby Points}\label{sec:control-derivatives}
Recall that the mollifier $g$ (\Cref{eq:mollifier-def}) gives zero value to points where the derivatives of the polynomial $p$ fail to decay at the desired rate.
Fix some points $\vec x^{(1)}, \cdots, \vec x^{(i-1)}, \bar {\vec y}, \vec y^{(i+1)}, \cdots,  \vec y^{(n)} \in \R^d$.
If it happens to be the case that for all $\xi \in 
[ -d^{ \ct }, d^{\ct} ]$ there exists some $t \in [k]$ such that
$$
\lp \| p^{[t]}\lp( \interpolate{\xi} \rp) \rp\|_F^2
\geq 
3 d^{- \cg t } p^2 \lp( \interpolate{\xi} \rp), 
$$
it follows by the definition of the mollifier that the mollified PTF will be zero 
over the entire truncated domain, i.e., for all $\xi \in [-d^{\ct}, d^{\ct}]$ we will have
$$
\E_{\xi \sim \bar A}  
\lp[  h( \vec x^{(1)},
\cdots, \vec x^{(i)}, \bar {\vec y} + \xi \vec v, \cdots \vec y^{(n)}) 
\rp]
=   
\E_{z \sim \bar \normal(0, 1)}  
\lp[  h( \vec x^{(1)},
\cdots, \vec x^{(i)}, \bar {\vec y} + z \vec v, \cdots \vec y^{(n)}) 
\rp]= 0.
$$
Consequently,  \Cref{eq:truncate-Y-fool-h} will hold trivially.

Hence, it suffices to consider the complementary case: there exists some $\xi^* \in 
[ -d^{ \ct }, d^{\ct} ]$ such that the desired derivative decay holds for all $t \in [k]$.
We formalize this complementary condition in the definition of a \emph{well-behaved} point set below.
\begin{definition}[Well-Behaved Point Set]
\label{def:well-behave}
Let $p:\R^{d\times n}\to \R$ be a polynomial on $n$ points of $\R^d$, $I \subseteq \R$ be an interval, and $\cg > 0$ be a parameter. Let $\vec x^{(1)}, \cdots, \vec x^{(i-1)}, \bar {\vec y}, \vec y^{(i+1)}, \cdots,  \vec y^{(n)}  \in \R^d$.
We say that these points form a well-behaved point set at position $i$  (with respect to $p,I$ and $\cg$) if there exists some $\xi^* \in 
I$ such that 
\begin{align*}
\forall t \in [k] \;\; 
\lp \| p^{[t]}\lp( \interpolate{\xi^*} \rp) \rp\|_F^2
\leq 
3 d^{- \cg t } \lp| p \lp( \interpolate{\xi^*} \rp) \rp|^2.
\end{align*}
\end{definition}
Throughout this section, we will always use $I = [-d^{\ct}, d^{\ct}]$ for the interval, and $\cg$ will be the same parameter as in the definition of the mollified PTF (\Cref{eq:mollifier-def}).
A key technical lemma we will prove is that if we condition on a well-behaved point set, the derivatives must also be ``approximately'' well-behaved for all $\xi$ in the truncated domain.
\begin{lemma}[Derivative Decay of Nearby Points]
\label{lem:nearby-derivative-decay}
Let $p:\R^{d\times n}\to \R$ be a polynomial on $n$ points of $\R^d$, and $\cg,\ct >0$ be parameters satisfying $\cg/2 - \ct \geq \Omega(1)$.
Let $\vec x^{(1)}, \cdots, \vec x^{(i-1)}, \bar {\vec y}, \vec y^{(i+1)}, \cdots  \vec y^{(n)}  \in \R^d$ be a well-behaved point set at position $i$ with respect to $p$,  interval $I = [-d^{ \ct }, d^{ \ct }]$, and $c_g$.
Then it holds that
\begin{align}
\label{eq:less-well-behave}
\forall t \in [k] \;\; \lp \| p^{[t]} \lp( \interpolate{\xi} \rp) \rp \|_F^2
\leq 
O(1) d^{ - \cg t }  p^2 \lp(  \interpolate{\xi} \rp)
\end{align}
for all $\xi \in I$.
\end{lemma}
\begin{proof}
By \Cref{def:well-behave}, there exists some point $\xi^*$ such that
\begin{align}
\label{eq:well-behave-restate}
\forall k' \in [k]: 
\lp \| p^{[k']}\lp( \interpolate{\xi^*} \rp) \rp \|_F^2
\leq 
3 d^{- \cg t }  p^2\lp( \interpolate{\xi^*} \rp).
\end{align}
Fix an arbitrary $\xi \in [-d^{\ct}, d^{\ct}]$.
We claim that the following holds.

\begin{claim}\label{cl:helper}
    Consider the setting and notation of \Cref{lem:nearby-derivative-decay}. For every $k' \in \{0, 1,\ldots ,k \}$, the following holds:
    \begin{align}
&\lp |
\lp\| p^{[k']} \lp( \interpolate{\xi} \rp)    \rp\|_F^2 
- \lp \| p^{[k']}\lp(  \interpolate{\xi^*} \rp) \rp\|_F^2
\rp| \nonumber \\
&\leq 
O \lp( d^{ - \cg k' + (\cg/2 - \ct)   } \rp)
\;   p^2\lp( \interpolate{\xi^*} \rp) .
\label{eq:derivative-rel-induction}
\end{align}
\end{claim}

We first show how to obtain \Cref{eq:less-well-behave} using \Cref{cl:helper}.
Applying \Cref{eq:derivative-rel-induction} with $k' = 0$ yields
\begin{align*}
&\lp |
 p^2\lp( \interpolate{\xi} \rp) 
-   p^2\lp(  \interpolate{\xi^*} \rp) \rp| \nonumber \\
&\leq 
O \lp( d^{ - (\cg/2 - \ct) } \rp)
\;  p^2\lp( \interpolate{\xi^*} \rp) .    
\end{align*}
Since we assume that $(\cg/2 - \ct) \geq \Omega(1)$,
it then follows that
\begin{align}
\label{eq:same-function-value}
    p^2\lp( \interpolate{\xi^*} \rp) 
    \leq
    (1 + o(1)) \; p^2\lp( \interpolate{\xi} \rp).
\end{align}
Applying \Cref{eq:derivative-rel-induction} for  $k' \in [k]$ gives
\begin{align*}
&\lp\| p^{[k']} \lp( \interpolate{\xi} \rp)    \rp\|_F^2 \\
&\leq 
\lp\| p^{[k']} \lp( \interpolate{\xi^*} \rp)    \rp\|_F^2 +
O \lp(  d^{-\cg (k'+0.01)} \rp)
p^2 \lp(  \interpolate{\xi^*} \rp) \\
&\leq \lp( d^{-\cg k'}
 + O \lp(  d^{-\cg k' + (\cg/2 - \ct)} \rp) \rp)
p^2 \lp(  \interpolate{\xi^*} \rp)  \tag{using \Cref{eq:well-behave-restate}}  \\
&\leq (1 + o(1)) d^{-\cg k'} \; p^2 \lp(  \interpolate{\xi} \rp)  \tag{using \Cref{eq:same-function-value} }.
\end{align*}
This proves \Cref{eq:less-well-behave} and completes the proof of \Cref{eq:less-well-behave}.
\end{proof}
It remains to show \Cref{cl:helper}. 
\begin{proof}[Proof of \Cref{cl:helper}]
We will use the Taylor expansion of $\|  p^{[k']}\lp(  \interpolate{\xi} \rp)  \|_F^2$ (which is a degree at most $2(k-k')$ polynomial) in the variable $\xi$ centered at $\xi^*$. 
In particular, Taylor's theorem gives that
\begin{align}
&\lp \|  p^{[k']}\lp(  \interpolate{\xi} \rp)  \rp \|_F^2 \nonumber \\
&= 
\lp \|  p^{[k']}\lp(  \interpolate{\xi^*} \rp)  \rp \|_F^2 \nonumber \\
&+ 
\sum_{t=1}^{2(k - k')} 
\lp(  \left. \frac{\partial^t}{\partial \xi^t} 
\lp \|  p^{[k']} \lp(  \interpolate{\xi} \rp)    \rp \|_F^2 \right|_{\xi = \xi^*} \rp)
\frac{ \lp(  \xi - \xi^* \rp)^{t } }{ t ! }.
\label{eq:p-full-taylor-expand}
\end{align}
We claim that the derivative 
$ \left. \frac{\partial^t}{\partial \xi^t} 
\lp \|  p^{[k']} \lp(  \interpolate{\xi} \rp)    \rp \|_F^2 \right|_{\xi = \xi^*}$ satisfies the following bound:
\begin{align}
\label{eq:directional-derivative-bound}
&\left. \frac{\partial^t}{\partial \xi^t} 
\lp \|  p^{[k']} \lp(  \interpolate{\xi} \rp)    \rp \|_F^2 \right|_{\xi = \xi^*} \nonumber \\
&\leq O(1) \; 2^{t} \;   d^{ - \cg (t/2 + k') } \; p^2\lp(  \interpolate{\xi^*} \rp).
\end{align}

Combining \Cref{eq:directional-derivative-bound} and  $|\xi-\xi^*|\leq d^{  \ct}$ then gives that
\begin{align*}
\lp| \left. \frac{\partial^t}{\partial \xi^t} 
\lp \|  p^{[k']} \lp(  \interpolate{\xi} \rp)    \rp \|_F^2 \right|_{\xi = \xi^*}    
(\xi - \xi^*)^t
\rp|
&\leq O(1) \; 2^t \; d^{  -  c_g (t/2 +k') + t \; \ct  } \\
&= O(1) \; 2^t \; d^{ -c_g k' - ( \cg/2 - \ct ) t }.
\end{align*}
Combining the above with \Cref{eq:p-full-taylor-expand}, 
and the fact that $\sum_t 2^t / t! = O(1)$ then gives that
\begin{align*}
&
\lp| \lp \|  p^{[k']}\lp(  \interpolate{\xi} \rp)  \rp \|_F^2
-
\lp \|  p^{[k']}\lp(  \interpolate{\xi^*} \rp)  \rp \|_F^2 \rp| \\
&\leq 
 O \lp(  d^{ -c_g k' - ( \cg/2 - \ct ) t }  
\rp) \; p^2\lp(  \interpolate{\xi^*} \rp).
\end{align*}
This concludes the proof of \Cref{eq:derivative-rel-induction}.
It remains to show \Cref{eq:directional-derivative-bound}. Using the product rule, the derivative $\left. \frac{\partial^t}{\partial \xi^t} 
\lp \|  p^{[k']} \lp(  \interpolate{\xi} \rp)    \rp \|_F^2 \right|_{\xi = \xi^*}$ is a sum of at most $2^{t}$ terms of the following form:
\begin{align*}
\lp \langle 
p^{[\beta]}\lp(  \interpolate{\xi^*} \rp)    \vec e(i)^{ \otimes \beta' }, p^{[\gamma]}\lp(  \interpolate{\xi^*} \rp) \vec e(i)^{ \otimes \gamma' } 
\rp \rangle
\, ,    
\end{align*}
where $\beta, \beta', \gamma, \gamma'$ are natural numbers such that
$ \beta -\beta' = \gamma - \gamma'$, and $\beta + \gamma = 2k' + t$.
Combining the above observation with the Cauchy–Schwarz inequality then gives the bound
\begin{align*}
&\lp| \left. \frac{\partial^t}{\partial \xi^t} 
\lp \|  p^{[k']} \lp(  \interpolate{\xi} \rp)    \rp \|_F^2 \right|_{\xi = \xi^*} \rp| \\
&\leq 
2^{t}
\max_{ \beta,\gamma: \beta + \gamma = 2k' + t }
\lp \|  p^{[\beta]}\lp(  \interpolate{\xi^*} \rp) \rp \|_F
\lp \| p^{[\gamma]}\lp(  \interpolate{\xi^*} \rp) \rp \|_F \\
&\leq
O(1)\; 2^t d^{ -\cg \lp(  k' + t/2 \rp) }  p^2( \interpolate{\xi^*} )
\tag{\Cref{eq:well-behave-restate}},
\end{align*}
where the last line used that the set of points $\vec x^{(1)}, \cdots, \vec x^{(i-1)}, \bar {\vec y}, \vec y^{(i+1)}, \cdots  \vec y^{(n)}$ is well-behaved (\Cref{def:well-behave}).
This concludes the proof of \Cref{cl:helper}.
\end{proof}

\subsubsection{Taylor Expansion of the Mollified PTF and Bounds for the Higher-Order Terms }\label{sec:taylor}

As explained in the technical overview of \Cref{sec:technical-overview}, the goal is to prove \Cref{eq:Y-fool-h} by performing a Taylor expansion of $h$. We then use the moment-matching property of the hidden direction distribution to bound the contribution of the low-order terms to the difference of expectations on the LHS of \Cref{eq:Y-fool-h} and leverage the derivative decay property of the mollifier to show that the contribution from the high-order error term is also small. The main result of this subsection formalizes the second argument that bounds the derivative appearing in the Taylor error term in \Cref{lem:mollify-derivative-decay}.
In particular, we 
consider the degree-$m$
expansion of $h$ along the direction of $\vec v$ at its $i$-th coordinate around some point $\xi^*$ that will be specified later: 
\begin{align}
h( \interpolate{\xi} ) &= \sum_{t=0}^{m-1}  \lp( 
\lp(\xi - \xi^* \rp)^t / t!\rp) D_{i,\vec v}^t h\lp(\origin\rp) \nonumber \\
&+ \lp( \lp(\xi - \xi^* \rp)^m / m!\rp) D_{i,\vec v}^m h\lp( \interpolate{\hat \xi}  \rp),
\label{eq:Taylor-Expand}
\end{align}
where $\hat \xi$ is some point between $\xi$ and $\xi^*$ which also depends on $ \vec x^{(1)}, \cdots  \bar {\vec  y}, \cdots, \vec y^{(n)}$.
Recall that $A$ has its first $m$-moments matched with $\normal(0, 1)$.
So the expected values of the first $m$ terms in \Cref{eq:Taylor-Expand} are identical under $\xi \sim A$
and $\xi \sim \normal(0, 1)$.
The rest of the section will focus on how we control the magnitude of the last term $(\lp(\xi - \xi^* \rp)^m / m!) D_{i,\vec v}^m h( \interpolate{\hat \xi}  )$.

We now proceed to control the magnitudes of the derivatives $D_{i,\vec v}^m h\lp( \interpolate{\hat \xi}  \rp)$.
\begin{lemma}[Mollified PTF Derivative Decay]
\label{lem:mollify-derivative-decay}
Let $p:\R^{n \times d} \to \R$ be a degree $k$ polynomial, $\vec v$ be a unit vector satisfying \Cref{eq:derivative-decay} with $\eps = 0.05$, and $h$ be the mollified PTF defined in \Cref{eq:mollified-ptf}. 
Let $\cg$ be the constant that appears in the definition of $h$, and $\ct>0$ be a parameter satisfying $\cg/2 - \ct \geq \Omega(1)$.
Let $\vec x^{(1)}, \cdots, \bar {\vec y}, 
\cdots, \vec y^{(n)}$ be a well-behaved point set (cf. \Cref{def:well-behave}) at position $i$, with respect to $p,[- d^{ \ct }, d^{\ct}  ]$, and $\cg$.
Then for all $t \in \mathbbm Z_+$ and  all $ \xi \in [- d^{ \ct }, d^{\ct}  ]$ it holds that 
$$
D_{i,\vec v}^t h\lp( \interpolate{\xi}  \rp) \leq 
(k+t)^{O(t)}
d^{ -  \cg t/2   }.
$$
\end{lemma}
The first observation is that 
computing $D_{i,\vec v}^t h\lp( \cdot  \rp)$ essentially boils down to computing the derivatives of the mollifier $D_{i,\vec v}^{t} g(\cdot)$.
\begin{claim}
\label{lem:trivial-derivative}
For all $\multix \in \R^d$, we have that
$$
D_{i,\vec v}^t h(\multix) = \sgn( p(\multix) ) 
D_{i,\vec v}^{t} g(\multix).
$$
\end{claim}
\begin{proof}
The claim follows from the product rule and the fact that the derivatives of $\sgn( p(\multix) ) $  are $0$ almost everywhere.
\end{proof}
Writing down the exact expression of 
the derivatives of the mollifier i.e., $D_{i,\vec v}^{t} g(\multix)$, is quite tedious. 
However, we show that the derivative is the sum of at most $(2k)^t$ many terms of a specific functional form.
For presentation purposes, similarly to earlier sections, we abbreviate $\multix$ by $\vec x^{(1:n)}$.
\begin{lemma}[Unfolded Derivatives of the Mollifier]
\label{lem:sum-of-product}
Let $i \in [n], t \in [k]$, and $\vec v$ be some unit vector, and $p,\rho,g$ defined as in \Cref{eq:rho-property}.
Then $D_{i,\vec v}^{t} g(\bmultix)$ is a sum of at most $T:= k^{O(t)}$ terms where the $j$-th term is of the form:
\begin{align}
\label{eq:single-term-expansion}
\Lambda_{j, t}:=
    &\pm d^{ (c_g/2) \kappa_j^{(t)}} \lp( \prod_{ (\alpha, \alpha') \in \mathcal A_j^{(t)}  }
    \rho^{(\alpha)} \lp( 
    \frac{     d^{ c_g  \alpha' }
    \| p^{[\alpha']}(\bmultix) \|_F^2 }{
    p^2 \lp( \bmultix \rp)
    }
 \rp) \rp) \nonumber \\ 
    &\lp( \prod_{ ( \beta, \beta', \gamma, \gamma' ) \in \mathcal B_j^{(t)}  }
    \frac{    \lp \langle
    p^{[\beta]}\lp(\bmultix\rp) \; \lp( \vec e^{(i)} \rp)^{\otimes \beta'},
    p^{[\gamma]}\lp(\bmultix\rp)
    \; \lp( \vec e^{(i)} \rp)^{\otimes \gamma'}
    \rp \rangle}{  p^2(\bmultix) } \rp) \, ,
\end{align}
where $\kappa_j^{(t)} \in \mathbb Z_+$, 
$\mathcal A_j^{(t)}$ is a multiset made up of elements from $\lp( \{0\} \cup [k] \rp)^{2} $,
$\mathcal B_j^{(t)}$ is a multiset made up of elements from $\lp( \{0\} \cup [k] \rp)^{4}$ \footnote{Of course, for the expression to be well defined, we will need $\beta - \beta' = \gamma - \gamma'$ as the tensor dimensions will not match up otherwise.}, and $\vec e(i)$ is the $i$-th standard basis vector.
Moreover, $\mathcal A_j^{(t)}, \mathcal B_j^{(t)}$ satisfy
\begin{itemize}
    \item \textbf{Maximum $\rho$ derivative degree:} $\sum_{ \alpha \in \mathcal A_j^{(t)} } \alpha \leq t$ for all $\mathcal A_j^{(t)}$.
    \item \textbf{Cardinality bound: }
    $| \mathcal A_{j, t} | + | \mathcal B_{j, t} | \leq k + t$.
    \item \textbf{Degree growth:}
    $-\kappa_j^{(t)} + \sum_{ (\beta, \beta', \gamma, \gamma')
    \in \mathcal B_j^{(t)}} (\beta + \gamma) \geq t$ for every $j \in [T]$.
\end{itemize}

\end{lemma}
\begin{proof}
We will show this by induction on $t$.
The base case is when $t = 0$. 
This corresponds to the case where we just have one term, where $\kappa_{1}^{(0)} = 0$,  $\mathcal A_1^{(0)} = \{ (i,i) \mid i \in [k] \}$, and $\mathcal B_1^{(0)} = \emptyset$.
The properties are then immediate.

We proceed to show the inductive step.
For convenience, define
\begin{align*}
&\Lambda_{j,t}^{ -(\alpha, \alpha') }
:= \Lambda_{j, t}
 \lp( \rho^{(\alpha)} \lp( 
 \frac{  d^{c_g} \alpha' \| p^{[\alpha']} \lp( \bmultix \rp) \|_F^2 }{ p^2( \bmultix ) }
 \rp) \rp)^{-1} \,  ,     \\
 &\Lambda_{j,t}^{ -(\beta, \beta', \gamma, \gamma') }
 := 
 \Lambda_{j, t} 
 \;    
 \lp( \frac{    \lp \langle
    p^{[\beta]}\lp(\bmultix\rp) \; \lp( \vec e^{(i)} \rp)^{\otimes \beta'},
    p^{[\gamma]}\lp(\bmultix\rp)
    \; \lp( \vec e^{(i)} \rp)^{\otimes \gamma'}
    \rp \rangle}{  p^2(\bmultix) } \rp)^{-1}.
\end{align*}

By the inductive hypothesis, we have that 
$D_{i,\vec v}^{t+1} g( \multix )
= \sum_{ j } D_{i,\vec v} \Lambda_{j,t}.
$
So it remains to compute $D_{i,\vec v} \Lambda_{j,t}$.
In particular, by the product rule, we have that
\begin{align*}
D_{i,\vec v} \Lambda_{j,t}
:= 
&\sum_{  ( \alpha, \alpha' ) \in \mathcal A_{j,t} }
\Lambda_{j, t}^{ -( \alpha, \alpha' ) } 
\; D_{i,\vec v} \rho^{(\alpha)} \lp( 
 \frac{  d^{c_g} \alpha' \| p^{[\alpha']} \lp( \bmultix \rp) \|_F^2 }{ p^2( \bmultix ) }
 \rp) \\
 &+ 
 \sum_{  ( \beta, \beta', \gamma, \gamma') \in \mathcal B_{j,t} }
\Lambda_{j, t}^{ -( \beta, \beta', \gamma, \gamma' ) } 
\; D_{i,\vec v}
\frac{    \lp \langle
    p^{[\beta]}\lp(\bmultix\rp) \; \lp( \vec e^{(i)} \rp)^{\otimes \beta'},
    p^{[\gamma]}\lp(\bmultix\rp)
    \; \lp( \vec e^{(i)} \rp)^{\otimes \gamma'}
    \rp \rangle}{  p^2(\bmultix) }.
\end{align*}
We first analyze the term
$D_{i,\vec v}
\left( \lp \langle
    p^{[\beta]}\lp(\bmultix\rp) \; \lp( \vec e^{(i)} \rp)^{\otimes \beta'},
    p^{[\gamma]}\lp(\bmultix\rp)
    \; \lp( \vec e^{(i)} \rp)^{\otimes \gamma'}
    \rp \rangle \right) p^2(\bmultix) $.
\begin{align*}
&D_{i,\vec v}
\frac{    \lp \langle
    p^{[\beta]}\lp(\bmultix\rp) \; \lp( \vec e^{(i)} \rp)^{\otimes \beta'},
    p^{[\gamma]}\lp(\bmultix\rp)
    \; \lp( \vec e^{(i)} \rp)^{\otimes \gamma'}
    \rp \rangle}{  p^2(\bmultix) } \\
&= 
\frac{    D_{i,\vec v}\lp \langle
    p^{[\beta]}\lp(\bmultix\rp) \; \lp( \vec e^{(i)} \rp)^{\otimes \beta'},
    p^{[\gamma]}\lp(\bmultix\rp)
    \; \lp( \vec e^{(i)} \rp)^{\otimes \gamma'} 
    \rp \rangle}{  p^2(\bmultix) } \\
&- 
\frac{ \lp \langle
    p^{[\beta]}\lp(\bmultix\rp) \; \lp( \vec e^{(i)} \rp)^{\otimes \beta'},
    p^{[\gamma]}\lp(\bmultix\rp)
    \; \lp( \vec e^{(i)} \rp)^{\otimes \gamma'} \rp \rangle
    D_{i,\vec v} p^2(\bmultix)
    }{  p^4(\bmultix) }  \tag{Quotient rule} \\
&= 
\frac{    \lp \langle
    p^{[\beta+1]}\lp(\bmultix\rp) \; \lp( \vec e^{(i)} \rp)^{\otimes \beta'+1},
    p^{[\gamma]}\lp(\bmultix\rp)
    \; \lp( \vec e^{(i)} \rp)^{\otimes \gamma'} 
    \rp \rangle}{  p^2(\bmultix) }
+ 
\frac{    \lp \langle
    p^{[\beta]}\lp(\bmultix\rp) \; \lp( \vec e^{(i)} \rp)^{\otimes \beta'},
    p^{[\gamma+1]}\lp(\bmultix\rp)
    \; \lp( \vec e^{(i)} \rp)^{\otimes \gamma'+1} 
    \rp \rangle}{  p^2(\bmultix) } \\
&- 
\frac{ \lp \langle
    p^{[\beta]}\lp(\bmultix\rp) \; \lp( \vec e^{(i)} \rp)^{\otimes \beta'},
    p^{[\gamma]}\lp(\bmultix\rp)
    \; \lp( \vec e^{(i)} \rp)^{\otimes \gamma'} 
    \rp \rangle}{  p^2(\bmultix) } \;
\frac{ 2\lp(  p^{[1]}(\bmultix) \vec e^{(i)} \rp) p(\bmultix)  }{ p^2(\bmultix) }. \tag{Product rule}
\end{align*}

One can check that each term in the summation above is still of the desired form. 
Moreover, we have that:
\begin{itemize}
\item The total derivative degree on $\beta$ is $\beta + \gamma + 1$ while the power $\kappa_j$ in the leading constant stays unchanged, ensuring the desired potential growth.
\item Since the last line contains $3$ terms, it follows that the number of terms in $$ \sum_{  ( \beta, \beta', \gamma, \gamma') \in \mathcal B_{j, t}}
\Lambda_{j, t}^{ -( \beta, \beta', \gamma, \gamma' ) } 
\; D_{i,\vec v}
\frac{    \lp \langle
    p^{[\beta]}\lp(\bmultix\rp) \; \lp( \vec e^{(i)} \rp)^{\otimes \beta'},
    p^{[\gamma]}\lp(\bmultix\rp)
    \; \lp( \vec e^{(i)} \rp)^{\otimes \gamma'}
    \rp \rangle}{  p^2(\bmultix) }
    $$ is $3 | \mathcal B_{j,t}|$.
\item The set cardinality bound increases by $1$ due to the third term.
\end{itemize}

Next, we analyze the term $
D_{i,\vec v} \rho^{(\alpha)} \lp( 
   d^{c_g} \alpha' \| p^{[\alpha']} \lp( \bmultix \rp) \|_F^2/ p^2( \bmultix )  
\rp)
$.
\begin{align}
&D_{i,\vec v} \rho^{(\alpha)} \lp( 
 \frac{  d^{c_g \alpha'}  \| p^{[\alpha']} \lp( \bmultix \rp) \|_F^2 }{ p^2( \bmultix ) }
\rp) \\
&=
\rho^{(\alpha+1)}
\lp( 
 \frac{  d^{c_g \alpha'}  \| p^{[\alpha']} \lp( \bmultix \rp) \|_F^2 }{ p^2( \bmultix ) }
\rp)
\; D_{i,\vec v}
 \frac{  d^{c_g \alpha'}  \| p^{[\alpha']} \lp( \bmultix \rp) \|_F^2 }{ p^2( \bmultix ) }  \tag{Chain rule} \\
 &=
 d^{c_g \alpha'}
\rho^{(\alpha+1)}
\lp( 
 \frac{  d^{c_g \alpha'}\| p^{[\alpha']} \lp( \bmultix \rp) \|_F^2 }{ p^2( \bmultix ) }
\rp)
\; 
\lp( 
 \frac{    
 D_{i,\vec v} \| p^{[\alpha']} \lp( \bmultix \rp) \|_F^2 }{ p^2( \bmultix ) }  
 - 
 \frac{  \| p^{[\alpha']} \lp( \bmultix \rp) \|_F^2
 D_{i,\vec v} p^2(\bmultix)
 }{ p^4( \bmultix ) } \rp)
 \tag{Quotient rule} \\
 &= 
 d^{c_g \alpha'}
\rho^{(\alpha+1)}
\lp( 
 \frac{  d^{( c_g / 2)  2\alpha'}  \| p^{[\alpha']} \lp( \bmultix \rp) \|_F^2 }{ p^2( \bmultix ) }
\rp)
\; 
\bigg( 
 \frac{   
 2 \lp \langle p^{[\alpha'+1]} \lp( \bmultix \rp) \vec e^{(i)}, 
 p^{[\alpha']} \lp( \bmultix \rp)
 \rp \rangle }{ p^2( \bmultix ) }    \\
 &- 
 \frac{   
\lp \langle p^{[\alpha']} \lp( \bmultix \rp) \vec e^{(i)}, 
 p^{[\alpha']} \lp( \bmultix \rp)
 \rp \rangle
 }{ p^2( \bmultix ) } 
 \frac{
\lp \langle p^{[1]} \lp( \bmultix \rp) \vec e^{(i)}, 
 p \lp( \bmultix \rp) \rp \rangle 
 }{ p^2( \bmultix ) } 
 \bigg).
 \tag{Product Rule}
\end{align}
One can check that each term in the summation above is still of the desired form. 
Moreover, we have that
\begin{itemize}
\item The total derivative degree is $2 \alpha' + 1$ while the power $\kappa_j$ in the leading constant increases by $2 \alpha'$, ensuring the desired potential growth.
\item Since the above equation has $2$ additive terms, it follows that
$$
\sum_{  ( \alpha, \alpha' ) \in \mathcal A_{j,t} }
\Lambda_{j, t}^{ -( \alpha, \alpha' ) } 
\; D_{i,\vec v} \rho^{(\alpha)} \lp( 
 \frac{  d^{c_g} \alpha' \| p^{[\alpha']} \lp( \bmultix \rp) \|_F^2 }{ p^2( \bmultix ) }
 \rp)
$$
contains at most $2 | \mathcal A_{j, t} |$ terms.

\item The set cardinality bound increases by $1$ due to the last term.
\item The maximum derivative degree of $\rho$ increases by $1$.
\end{itemize}
The total number of terms in $D_{i,\vec v}^{t+1} g( \multix )$ is at most
$ \left( O(k) \right)^t  \;  \lp( 
3 | \mathcal B_{j,t}| + 2 | \mathcal A_{j, t} |
\rp)  \leq  \lp( O(k) \rp)^{t+1}$.
This concludes the inductive step as well as the proof of \Cref{lem:sum-of-product}.
\end{proof}

We are now ready to conclude the proof of \Cref{lem:mollify-derivative-decay}.
\begin{proof}[Proof of \Cref{lem:mollify-derivative-decay}]
For notational convenience, we define $\vec x^{(i)} = \bar {\vec y} + \xi \vec v$.
Recall that $g$ and $h$ are the mollifier and the mollified PTF respectively, defined in \Cref{eq:mollifier-def,eq:mollified-ptf}.
By \Cref{lem:trivial-derivative}, it holds
$$
\lp| D_{i,\vec v}^t h(\multix) \rp| \leq 
\lp| D_{i,\vec v}^{t} g(\multix) \rp|.
$$
By \Cref{lem:sum-of-product}, $D_{i,\vec v}^{t} g(\multix)$ is the sum of at most $ k^{O(t)}$ terms $\Lambda_{j, t}$ of the form given in \Cref{eq:single-term-expansion}.
We now claim that each term in the form of \Cref{eq:single-term-expansion} is at most 
$ t^{O(t)} d^{ - c_g t/2  }$.
It follows immediately that
\begin{align*}
    \lp| D_{i,\vec v}^{t} g(\multix) \rp| 
    &\leq 
t^{O(t)} k^{O(t)} d^{ - c_g t/2  } \\
&\leq (t+k)^{O(t)} d^{ - c_g t/2  }
\end{align*}
for all $t \in \mathbbm Z_+$.

It remains to show that the quantity $\Lambda_{j, t}$ in \Cref{eq:single-term-expansion} is at most $ t^{O(t)} d^{ - \cg t/2  }$.
Recall that $ \| \rho^{(\alpha)} \|_{\infty}$ is at most $ \alpha^{ O(\alpha) }$ by \Cref{eq:rho-property}.
Combining this with the maximum degree property of \Cref{lem:sum-of-product} then gives that 
\begin{align}
\label{eq:alpha-bound}
    \prod_{ (\alpha, \alpha') \in \mathcal A_j^{(t)}  }
    \rho^{(\alpha)} \lp( 
    \frac{     d^{ c_g  \alpha' }
    \| p^{[\alpha']}(\bmultix) \|_F^2 }{
    p^2 \lp( \bmultix \rp)
    }
 \rp)
\leq t^{O(t)}.
\end{align}

Next, applying the Cauchy–Schwarz  inequality gives that
\begin{align}
\label{eq:cz-decouple}
    \frac{    \lp \langle
    p^{[\beta]}\lp(\bmultix\rp) \; \lp( e^{(i)} \rp)^{\otimes \beta'},
    p^{[\gamma]}\lp(\bmultix\rp)
    \; \lp( e^{(i)} \rp)^{\otimes \gamma'}
    \rp \rangle}{  p^2(\bmultix) } 
    &\leq 
    \frac{\| p^{[\beta]}\lp(\bmultix\rp)  \|_F }{
    \lp| p(\bmultix) \rp|}
    \frac{\| p^{[\gamma]}\lp(\bmultix\rp) \|_F}{\lp| p(\bmultix) \rp|}.
\end{align}
Recall that we define $\vec x^{(i)} := \bar {\vec y} + \xi \vec v$. By the assumption of the lemma, 
$\vec x^{(1)}, \cdots, \bar {\vec y}, \cdots, \vec x^{(n)}$ form a well-behaved point set with respect to $p,[- d^{ \ct }, d^{\ct}  ]$, and $\cg$ (cf. \Cref{def:well-behave}), and $\xi \in [-d^{\ct}, d^{\ct}]$.
Therefore, \Cref{lem:nearby-derivative-decay} is applicable. This gives that
\begin{align}\label{eq:tmp5345}
    \frac{\| p^{[\beta]}\lp(\bmultix\rp)  \|_F }{
    \lp| p(\bmultix) \rp|}
    \frac{\| p^{[\gamma]}\lp(\bmultix\rp) \|_F}{\lp| p(\bmultix) \rp|}
    \leq O\lp( d^{ -  c_g \lp( \beta + \gamma\rp)/2 }\rp).
\end{align}
We will use \Cref{eq:tmp5345} to bound the remaining part of $\Lambda_{j,t}$ from \Cref{eq:single-term-expansion}; all the factors in the RHS of \Cref{eq:single-term-expansion} excluding the $\prod_{ (\alpha, \alpha') \in \mathcal A_j^{(t)}  }
    \rho^{(\alpha)} \lp( 
        d^{ c_g  \alpha' }
    \| p^{[\alpha']}(\bmultix) \|_F^2 /
    p^2 \lp( \bmultix \rp)
 \rp)$ that have already been bounded earlier. Combining \Cref{eq:tmp5345} with \Cref{eq:cz-decouple} gives the following bound for the term
 \begin{align*}
   &d^{ (c_g/2) \kappa_j^{(t)}} \prod_{ (\beta, \beta', \gamma, \gamma') \in \mathcal B_{j, t}}
    \frac{    \lp \langle
    p^{[\beta]}\lp(\bmultix\rp) \; \lp( e^{(i)} \rp)^{\otimes \beta'},
    p^{[\gamma]}\lp(\bmultix\rp)
    \; \lp( e^{(i)} \rp)^{\otimes \gamma'}
    \rp \rangle}{  p^2(\bmultix) }  \\
    &\leq O\lp( d^{ -  (c_g/2)
    \lp( \lp( \sum_{ (\beta, \beta', \gamma, \gamma') \in \mathcal B_{j, t}}
     \beta + \gamma \rp) - \kappa_j^{(t)}\rp) } \rp).     
 \end{align*}
By the degree growth property of \Cref{lem:sum-of-product}, it follows that
\begin{align}
\label{eq:beta-gamma-bound}
   d^{(c_g/2) \kappa_j^{(t)}} \prod_{ (\beta, \beta', \gamma, \gamma') \in \mathcal B_{j, t}}
    \frac{    \lp \langle
    p^{[\beta]}\lp(\bmultix\rp) \; \lp( e^{(i)} \rp)^{\otimes \beta'},
    p^{[\gamma]}\lp(\bmultix\rp)
    \; \lp( e^{(i)} \rp)^{\otimes \gamma'}
    \rp \rangle}{  p^2(\bmultix) } 
    \leq d^{ -  \cg t/2}.
\end{align}
Combining \Cref{eq:beta-gamma-bound} and \Cref{eq:alpha-bound} then gives that the expression in \Cref{eq:single-term-expansion} is at most
$ t^{ O(t) } d^{-  \cg t /2  } $.
This concludes the proof of \Cref{lem:mollify-derivative-decay}.
\end{proof}
\subsubsection{Putting Everything Together: Proof of \Cref{prop:replacement-step}}\label{sec:wrapup}
We are now ready to conclude the proof of \Cref{prop:replacement-step}, restated below for convenience:
\MAINPROPOSITION*

\begin{proof}
First,  using an averaging argument, it suffices to show  the following for  an arbitrary set of points $S_{i} := \{ \vec x^{(1)}, \cdots, \vec x^{(i-1)},\bar {\vec y}, \vec y^{(i+1)}, \cdots \vec y^{(n)} \}$:
\begin{align}
&\lp| 
\E_{\xi \sim A}  
\lp[  h( \vec x^{(1)},
\cdots , \vec x^{(i)}, \bar {\vec y} + \xi \vec v, \cdots , \vec y^{(n)}) \rp]
-  
\E_{z \sim \normal(0, 1)}  
\lp[  h( \vec x^{(1)},
\cdots , \vec x^{(i)}, \bar {\vec y} + z \vec v, \cdots , \vec y^{(n)}) \rp] \rp|
\nonumber
\\
&
\leq d^{ - \cg m / 2 + cm }
.
\label{eq:goal1}
\end{align}

Throughout the proof, we will fix $\ct = \cg/2 - c/2$, 
and set the truncated domain to be $I = [ -d^{\ct}, d^{\ct} ]$.

It then follows from \Cref{lem:truncation} that 
$$
\max \lp( \E_{ \xi \sim A } \lp[ \xi \in I  \rp]  \, , \, 
\E_{ \xi \sim \normal(0, 1) } \lp[ \xi \in I  \rp] \rp)
\leq m^m d^{ -\ct m  } \leq d^{- \cg m/2 + c m } \, ,
$$
where the last inequality follows from the assumption that 
$m < d^{c / C}$

Since the mollified PTF is constructed to be bounded from above by $1$, 
showing \Cref{eq:goal1} can be reduced to showing
\begin{align}
&\lp| 
\E_{\xi \sim \bar A}  
\lp[  h( \vec x^{(1)},\cdots, \vec x^{(i)}, \bar {\vec y} + \xi \vec v, \cdots, \vec y^{(n)}) 
\rp]
-  
\E_{\xi \sim \bar \normal(0, 1)}  
\lp[  h( \vec x^{(1)},\cdots, \vec x^{(i)}, \bar {\vec y} + z \vec v, \cdots, \vec y^{(n)}) 
\rp] \rp| \nonumber \\
&
\leq d^{ - \ct m + c m / 2  }
\, ,
\label{eq:truncate-Y-fool-h-restate}
\end{align}
where $\bar A$ and $\bar \normal(0, 1)$ are the truncated  versions of the distributions that condition on $\xi$ and $z$ being inside the interval
$[ -d^{\ct}, d^{\ct} ]$.

We will show \Cref{eq:truncate-Y-fool-h-restate} by considering two cases.
The first is when $S_{i}$ is not a well-behaved point set (\Cref{def:well-behave}) at position $i$. That is, we have that for every $\xi \in [-d^{\ct}, d^{\ct}]$, there exists some $t$ such that
$$ \lp \| p^{[t]} \lp( \vec x^{(1)}, \cdots, \bar {\vec y} + \xi \vec v, \cdots, \vec y^{(n)} \rp) \rp \|_F^2
> 3 d^{- c_g t}
 \lp( p( \vec x^{(1)}, \cdots, \bar {\vec y} + \xi \vec y, \cdots, \vec y^{(m)} ) \rp)^2.
$$
By the definition of the mollified PTF (cf. \Cref{eq:mollifier-def}), we immediately have that in this case 
$h( \vec x^{(1)},\cdots, \vec x^{(i)}, \bar {\vec y} + \xi \vec v, \cdots, \vec y^{(n)}) = 0$ for all $\xi \in [ -d^{\ct}, d^{ \ct} ]$.
Thus, \Cref{eq:truncate-Y-fool-h-restate} follows trivially in this case as the left hand side is zero.

Now consider the complementary case where $S_{i}$ is a well-behaved point set at position $i$.
By \Cref{lem:mollify-derivative-decay}, we have that 
\begin{align}
\label{eq:less-well-behave-2}
D_{i,\vec v}^t h\lp( \interpolate{\xi}  \rp) \leq 
\lp( k+t\rp)^{O(t)} d^{ -  \cg t/2  }
\end{align}
for all $t \in [m]$ and $\xi \in [-d^{\ct}, d^{\ct}]$.
In this case, we will show \Cref{eq:truncate-Y-fool-h-restate} by rewriting $h$ in terms of its Taylor expansion, and then bounding the differences between the Taylor terms.
Consider the degree-$m$ Taylor expansion of $h$ around $0$:
\begin{align}
h( \interpolate{\xi} ) &= \sum_{t=0}^{m-1}  \lp( \xi^t / t!\rp) D_{i,\vec v}^t h\lp(\origin\rp) \nonumber \\
&+ \lp( \xi^m / m!\rp) D_{i,\vec v}^m h\lp( \interpolate{\hat \xi}  \rp),     
\label{eq:Taylor-Expand-2}
\end{align}
where $\hat \xi$ is some point between $\xi$ and $\xi^*$ which also depends on $ \vec x^{(1)}, \cdots  \bar {\vec y}, \cdots, \vec y^{(n)}$ and $\xi^*$ lies in $[- d^{\ct}, d^{\ct}]$.
For convenience, we write
$ \Delta(a) := \E_{ \xi \sim \bar A }[ \xi^a ] - \E_{ \xi \sim \bar \normal(0, 1) }[ \xi^a ].$
For $t \in [m-1]$, we have that
\begin{align}
&\lp| 
D_{i,\vec v}^t h\lp(\origin\rp)
\lp( 
\E_{ \xi \sim \bar A }
\lp[  \xi^t \rp]
-
\E_{ \xi \sim \bar \normal(0, 1) }
\lp[  \xi^t \rp]
\rp)
\rp| \\
&= 
\lp| 
D_{i,\vec v}^t h\lp(\origin\rp)
\Delta(a)\rp| \\
&\leq
(m+t)^{O(t)}
d^{- c_g t/2 } \;
d^{  -\ct (m-t) }
\tag{by \Cref{eq:less-well-behave-2}, and \Cref{lem:truncation}} \\
&\leq 
(m+t)^{O(t)}
d^{ - \ct m}.  \tag{using  $\ct = \cg/2 - 0.01 $}
\end{align}
In particular, this implies that
\begin{align}
\label{eq:low-order-moment-diff-bound}
\bigg |
&\E_{ \xi \sim \bar A }
\lp[ \sum_{t=1}^{m-1}  \lp( \xi^t / t!\rp) D_{i,\vec v}^t h\lp(\origin\rp) \rp]
-
\\
&\E_{ \xi \sim \bar \normal(0, 1) }
\lp[ 
\sum_{t=1}^{m-1}  \lp( \xi^t / t!\rp) D_{i,\vec v}^t h\lp(\origin\rp)
\rp] \bigg | \\
&\leq m \; m^{O(m)}
d^{ - \ct m}
=  m^{O(m)} \; d^{ -  \ct m } \, ,
\end{align}
where we used the triangle inequality.
For the last Taylor remainder term, 
applying \Cref{eq:less-well-behave-2} again gives that
\begin{align*}
\lp( \xi^m / m!\rp)  D_{i,\vec v}^m h\lp( \interpolate{\hat \xi}  \rp) 
\leq 
\lp( \xi^m / m!\rp)  (k+m)^{O(m)} d^{ -\cg m/2 }  \, ,
\end{align*}
for all $\xi$.
In particular, this implies that the expected value of the Taylor remainder term under the distribution of $\xi \sim \bar A$
is at most
\begin{align*}
&(k+m)^{O(m)} d^{ -\cg m/2 }
\E_{ \xi \sim \bar A }    
\lp[ 
\xi^m / m! 
\rp] \\
&\leq
(k+m)^{O(m)} d^{ -\cg m/2 }
\E_{ \xi \sim A }    
\lp[ 
 \xi^m / m! 
\rp] \tag{since the mass of $A$ within the truncated interval is $1 - o(1)$ } \\
&= 
(k+m)^{O(m)} d^{ -\cg m/2 }
\E_{ \xi \sim \normal(0, 1) }    
\lp[ 
\lp( \xi^m / m!\rp)  
\rp] \tag{since we assume the degree-$m$ moment of $A$ match with $\normal(0, 1)$ } \\
&\leq 
(k+m)^{O(m)} d^{ -\cg m/2 } \tag{by the Gaussian moment bound }
\end{align*}
The same bound can be established for the Taylor remainder term under the distribution of $\xi \sim \bar \normal(0, 1)$ as well.
It then follows from the trinagle inequality that the difference between the expected value of the Taylor remainder term under $\bar \normal(0, 1)$ and $\bar A$ is at most $(k+m)^{O(m)} d^{ -\cg m/2 }$.
Combining this with \Cref{eq:low-order-moment-diff-bound,eq:Taylor-Expand-2} then shows that
\begin{align*}
&\lp| 
\E_{\xi \sim \bar A}  
\lp[  h( \vec x^{(1)},\cdots, \vec x^{(i)}, \bar {\vec y} + \xi \vec v, \cdots, \vec y^{(n)}) 
\rp]
-  
\E_{\xi \sim \bar \normal(0, 1)}  
\lp[  h( \vec x^{(1)},\cdots, \vec x^{(i)}, \bar {\vec y} + z \vec v, \cdots, \vec y^{(n)}) 
\rp] \rp| \nonumber \\
&
\leq 
(k+m)^{O(m)}
d^{ - \ct m   }
\leq d^{ - \ct m + c m /2 }
\, ,    
\end{align*}
where the last inequality follows from the assumption that
$d > \max\lp(  k^{C/c}, m^{C/c} \rp)$.
This concludes the proof of \Cref{eq:truncate-Y-fool-h-restate}, as well as \Cref{prop:replacement-step}.
\end{proof}

\newpage 
\bibliographystyle{alpha}
\bibliography{mydb.bib}

\newpage 

\appendix

\section*{Appendix}

\section{Additional Preliminaries}

\subsection{Basics of Hermite Polynomials} \label{app:Hermite_polynomials}

Hermite polynomials form a complete orthogonal basis of the vector space $L^2(\R,\normal(0,1))$ of all functions $f:\R \to \R$ such that $\E_{x\sim \normal(0,1)}[f^2(x)]< \infty$. 
	We will use the \emph{normalized probabilist's} Hermite polynomials, which have unit norm and are pairwise orthogonal with respect to the Gaussian measure, i.e., $\int_\R h_k(x) h_{m}(x) e^{-x^2/2} \d x = \sqrt{2\pi} \mathbf{1}(k=m)$ 
    These polynomials are the ones obtained by Gram-Schmidt orthonormalization of the basis $\{1,x,x^2,\ldots\}$ with respect to the inner product $\langle f,g \rangle_{\normal(0,1)}:=\E_{x \sim \normal(0,1)}[f(x)g(x)]$. Every function  $f \in L^2(\R,\normal(0,1))$ can be uniquely written as $f(x) = \sum_{i =0}^\infty a_i h_i(x)$ and we have $\lim_{n \rightarrow \infty}\E_{x \sim \normal(0,1)}[(f(x)- \sum_{i =0}^n a_i h_i(x))^2] = 0$. 
    We have the following closed form formula (see, e.g., \cite{Sze67}):
\begin{align}\label{eq:Hermite_formula}
        h_n(x) = 
            \sqrt{n!}\sum_{j=0}^{\lfloor n/2 \rfloor} \frac{(-1)^j}{j! (n-2j)! 2^j}x^{n-2j} \;.
    \end{align}
To extend the basis to $d$-dimensions, we use a multi-indices
$\vec J\in
    \N^d$ to define the $d$-variate normalized Hermite polynomial. For $\vec J = (v_1,\ldots,v_d)$ we define $h_{\vec J}(\vec x) =
    \prod_{i=1}^d h_{v_i}(\x_i)$.  The total degree of $h_{\vec J}$ is $|{\vec J}| = \sum_{v_i \in
        {\vec J}} v_i$.  Given a function $f \in L^2(\R^d,\normal(\vec 0,\vec I))$ we compute its Hermite coefficients as
\(
\hat{f}({\vec J}) = \E_{\vec x\sim \normald{n}} [f(\vec x) h_{\vec J}(\vec x)]
\)
and express it uniquely as
\(
\sum_{{\vec J} \in \N^n} \hat{f}({\vec J}) h_{\vec J}(\vec x).
\)
For more details on the Gaussian space and Hermite Analysis (especially from
the theoretical computer science perspective), we refer the reader to
\cite{Don14}.  Most of the facts about Hermite polynomials that we use in this
work are well known properties and can be found, for example, in \cite{Sze67}.

We denote by $f^{[k]}(x)$ the degree $k$ part of the Hermite expansion of $f$,
$f^{[k]} (\vec x) = \sum_{|{\vec J}| = k} \hat{f}({\vec J})\cdot h_{\vec J}(\vec x)$.
We say that a polynomial $q$ is harmonic of degree $k$ if it is
a linear combination of degree $k$ Hermite polynomials, that is $q$ can be
written as
$$ q(\vec x) = q^{[k]}(\vec x) = \sum_{{\vec J}: |{\vec J}| = k} c_{\vec J} h_{\vec J}(\vec x) \;.
$$

We will use the following fact, stating that odd degree Hermite polynomials are small around the origin.
\begin{claim}[Upper Bound on Hermite Polynomial around the Origin]\label{fact:max_coeff}
    Let $\delta \in (0, 1/2)$ be 
    such that $\delta < k^{-C}$ for some sufficiently large constant $C$, and $h_{\vec a}$ be a degree-$k$ multivariate Hermite polynomial.
    We then have that
    $ h_{\vec a}( \delta \; \vec 1 ) < 1$.
\end{claim}
\begin{proof}
Consider a univariate Hermite polynomial $h_k: \R \mapsto \R$.
By the explicit formula of $h_k(\delta)$ in \Cref{eq:Hermite_formula}, it follows that the polynomial is dominated by its constant term (when $k$ is even) or its linear term (when $k$ is odd)
when $\delta$ is a sufficiently small polynomial in its degree $k$. 
It is not hard to verify that the coefficient of the constant term or linear term is smaller than $1$.
It then follows that $h_k(\delta) < 1$ when $\delta$ is a sufficiently small polynomial in $k$.
Since the multivariate Hermite polynomials are just products of many univariate Hermite polynomials, it follows that
$h_{\vec a}(  \delta \; \vec 1 ) < 1$.
\end{proof}

Another useful property of Hermite polynomials is that there exists a nice recurrence relationship between the polynomial itself and its derivative.
\begin{fact}
\label{fact:recurrence-derivative}
We have that $ \frac{d}{dx} h_k(x) =   k h_{k-1}(x) $.
\end{fact}

The fact implies the following bound on the Lipchitzness of the multivariate Hermite polynomials are around the origin.
\begin{claim}[Lipchitz continuity of Hermite Polynomials around the Origin]
\label{clm:Hermite-lip}
Let $k \in \mathbbm Z_+$, and $\delta \in (0, 1/2)$ be at most a sufficiently small polynomial in $k$, and $h_{\vec a}: \R^n \mapsto \R$ be a multivariate degree-$k$ Hermite polynomial.
Then it holds that
$$
\lp| h_{\vec a} \lp(  \delta \vec 1_n \rp) -  h_{\vec a} \lp(  \vec 0_n \rp) \rp| \leq \sqrt{ k^3 n } \delta.
$$
\end{claim}
\begin{proof}
For notational convenience, 
we define $h_{ \vec b }(\vec x)$ to be the constant $0$ function when $\vec b$ contains any negative entries.

Under this notation, 
\Cref{fact:recurrence-derivative} then implies that the gradient of $h_{\vec a}$ is simply
\begin{align*}
\nabla h_{\vec a} (\vec x)
= \lp(  \vec a_1 h_{\vec a - \vec e(1)}(\vec x), \cdots,  \vec a_n h_{\vec a - \vec e(n)}(\vec x)
\rp)^\top,
\end{align*}
where $e(i)$ is the multi-index having zeroes everywhere except from the $i$-th position, where it has $1$.
Combining the above with \Cref{fact:max_coeff} then gives that
\begin{align*}
\lp \| \nabla h_{\vec a} (\delta  \vec 1) \rp\|_2^2
&\leq
\sum_{i=1}^n
k^2 h_{ \vec a - \vec e(i) }^2( \delta \vec 1 ) \\
&= 
\sum_{i: \vec a_i > 0}
k^2 h_{ \vec a - \vec e(i) }( \delta \vec 1 )^2 \\
&\leq  k^3
\, ,
\end{align*}
where the last inequality follows from \Cref{fact:max_coeff} and the fact that there can be at most $k$ positive entries in $\vec a$.
We therefore have that
\begin{align*}
\lp| h_{\vec a} \lp(  \delta \vec 1_n \rp) -  h_{\vec a} \lp(  \vec 0_n \rp) \rp|
\leq 
\| \nabla h_{\vec a}(\delta \vec 1)  \|_2 
\| \delta \vec 1 \|_2
\leq \sqrt{ k^3 n } \delta.
\end{align*}
This concludes the proof of \Cref{clm:Hermite-lip}.
\end{proof}

\subsection{Other Facts}\label{ap:sign}

\begin{fact}[Isserlis's Theorem]\label{fact:iserrlis}
	Let $(x_1,\ldots, x_k) \sim \normal(0,\vec \Sigma)$. Then, 
	\begin{align*}
		\E[x_1 \cdots x_k] = \sum_{p \in P_k^2} \prod_{\{i,j\} \in p} \E[x_i x_j] \;,
	\end{align*}
	where $P_k^2$ is the set of all matchings of $\{1,\ldots,k\}$.
\end{fact}

\begin{lemma}
\label{lem:smooth-function-construct}
There exists a smooth function $\rho: \R \mapsto [0, 1]$ satisfying that (1) $\rho(x) = 1$ if  $|x| < 1 $ \, ,
(2) $\rho(x) = 0$ if $|x| \geq 3$, and $\| \rho^{(t)}(x) \|_{\infty} \leq t^{O(t)}$.
\end{lemma}
\begin{proof}
Within the context of this proof, we call a function smooth if its $t$-th order derivative is bounded from above by $t^{O(t)}$.
First, consider the function $f: \R \mapsto \R$ defined as
\begin{align*}
\rho_0(x) := \begin{cases}
    &0 \text{ if } x \geq 0  \\
    &\exp(-1/x^2) \text{ otherwise}.
\end{cases}    
\end{align*}
Then we have $\rho_0$ is a smooth function, and $\rho_0(x) = 0$ for all $x \geq 0$.
Next, define $\rho_1(x) := \rho_0(-1+x)\rho_0(-1-x)$.
$\rho_1$ is still a smooth function, and we have $\rho_1(x) = 0$ if $|x| \geq 1$, and $\rho_1(x) > 0$ if $|x| < 1$.
We can define a probability distribution supported on $[-1, 1]$ whose probability density function is exactly proportional to $\rho_1$.
Denote by $\rho_2$ the cumulative density function of this probability distribution.
$\rho_2$ remains a smooth function.
Furthermore, we have that $\rho_2(x) = 0$ if $x < -1$, and
$\rho_2(x) = 1$ if $x > 1$.
Finally, define $\rho(x) = \rho_2(2+x) \; \rho_2(2-x)$.
$\rho$ is still a smooth function, and it satisfies that 
$\rho(x) = 1$ if $|x| < 1$, and $\rho(x) = 0$ if $|x| > 3$.
This concludes the proof of \Cref{lem:smooth-function-construct}.
\end{proof}

\section{Separation between PTF Tests and LDP Tests}\label{sec:comparison}
In this section, we show that having no $\gamma$-advantageous polynomials in the sense of \Cref{def:good-ld} for a testing problem  does not necessarily rule out the existence of a good PTF test in the sense of \Cref{def:goodptf}.
In this section, we will work with  hypothesis testing where the family of distributions for the alternative hypothesis consists of only one distribution. We restate the simplified version of $\gamma$-advantageous for simple hypothesis testing below:

\begin{restatable}[$\gamma$-advantageous polynomial]{definition}{GOODPOLYNOMIALSIMPLE}\label{def:good-ld-simple}
Let $\gamma > 0$, $p: \R^{ n \times d } \mapsto \R$ be a degree-$k$, $n$-sample polynomial, and $D_{\emptyset}$ be a distribution in $\R^d$, $\cD_{\alt}$ be a distribution family in $\R^d$ and $\mu$ be the uniform distribution over $\cD_{\alt}$. We say that $p$ is a degree-$k$, $n$-sample, $\gamma$-advantageous polynomial with respect to $D_{\emptyset},\cD_{\alt}$ if:
\begin{align}
&\lp|
\E_{\multix \sim D_{\emptyset}}
\lp[  p(\vec x^{(1)}, \cdots ,\vec x^{(n)}) \rp]
- 
 \E_{ \multiy \sim D_{\alt}}
\lp[  p(\vec y^{(1)}, \cdots, \vec y^{(n}) \rp]
\rp| \notag\\
&> 
\gamma \; \max \lp( 
\sqrt{ \Var_{\multix \sim D_{\emptyset}}\lp[ 
p(\vec x^{(1)}, \cdots, \vec x^{(n)})
\rp] } 
\, , \, 
\sqrt{ \Var_{ \multiy \sim D_{\alt}}\lp[ 
p(\vec y^{(1)}, \cdots \vec y^{(n)})
\rp] } 
\rp).    \label{eq:rhs}
\end{align}
\end{restatable}

We show that even in one-dimension there exist distributions $D_{\emptyset},D_{\alt}$ for which there is no low-degree $\gamma$-advantageous polynomial, but the hypothesis problem can be easily solved with a polynomial threshold test of degree $k=1$ using $n=1$ sample.
The definition of the testing problem, and the statement of the claim is given below.
\begin{definition}[$\delta$-Gap Threshold Test under $\eps$-Gaussian Noise]\label{def:example}
Let $\delta,\eps \in (0, 1)$.
We consider the hypothesis testing problem of distinguishing between the two distributions $D_{\emptyset}$ and $D_{\alt}$ in $\R$ defined as follows:
(1) $D_{\emptyset} := (1 - \eps) p_{0} + \eps \normal(0, 1)$, where $p_{0}$ is a point mass on $0$, and
(2) $D_{\alt} := (1 - \eps) p_{\delta} + \eps \normal(0, 1)$, where $p_{\delta}$ is a point mass on $\delta$.
\end{definition}

\begin{theorem}
\label{thm:separation}
Let $\eps, \gamma, \delta \in (0, 1/2)$, and $n, k \in \mathbbm Z_+$.
Assume that 
$\sqrt{k^n \; \eps^{-n} \; k^3 n } \; \delta 
    \leq \gamma$.
Then there is no degree-$k$, $n$-sample, $\gamma$-advantageous polynomial with respect to $D_{\emptyset},D_{\alt}$ from \Cref{def:example}.
However, the $1$-sample linear threshold function $\mathbbm 1\{x>\delta/2\}$ distinguishes between $D_{\emptyset},D_{\alt}$ with probability at least $1 - \eps$.
\end{theorem}
\begin{proof}
The fact that the linear threshold function test $\mathbbm 1\{x > \delta / 2\}$ distinguishes between $D_{\emptyset},D_{\alt}$ with probability $1 - \eps$ is immediate by the definition of the problem.

We will use the notation $D_{\emptyset}^{\otimes n}$ and $D_{\alt}^{\otimes n}$ to denote the product distribution of $n$ samples under the two hypotheses.
Now, consider a degree-$k$ polynomial $p:\R^{n} \mapsto \R$.
Since \Cref{eq:rhs} is invariant to shifting and scaling of the polynomial, we can assume without loss of generality that 
$  \Var_{ \vec x \sim \lp( D_{\emptyset} \rp)^{\otimes n} }
\lp [ p(\vec x) \rp] = 1 $ and $\E_{\vec x \sim  \lp( D_{\emptyset}\rp)^{\otimes n}}[p(\vec x)] = 0$.

Under this assumption, note that we can bound from above the 
$L_2$-norm of $p$ under the standard Gaussian distribution as
\begin{align*}
\E_{ \vec x \sim \normal(\vec 0, \vec I) }
\lp[  p^2(\vec x) \rp]    
&\leq
\eps^{-n}
\E_{ \vec x \sim \lp( D_{\emptyset} \rp)^{\otimes n} }
\lp[  p^2(\vec x) \rp]  \\
&= 
\eps^{-n}
\Var_{ \vec x \sim \lp( D_{\emptyset} \rp)^{\otimes n} }
\lp[  p(\vec x) \rp] = \eps^{-n} \, ,
\end{align*}
where the first inequality follows from the fact that 
$\vec x$ will be sampled from the $n$-dimensional standard Gaussian distribution with probability $\eps^{n}$ by the definition of $D_{\emptyset}$, and the last two equalities follow from our assumptions on $\Var_{ \vec x \sim \lp( D_{\emptyset} \rp)^{\otimes n} }
\lp [ p(\vec x) \rp]$ and $\E_{\vec x \sim  \lp( D_{\emptyset}\rp)^{\otimes n}}[p(\vec x)]$.

It then suffices for us to show that
\begin{align}
\label{eq:mean-bound}
\lp|
\E_{ \vec x  \sim \lp( D_{\emptyset} \rp)^{\otimes n}  }
\lp[  p(\vec x) \rp]
- 
\E_{\vec y \sim \lp( D_{\alt}\rp)^{\otimes n} }
\lp[  p( \vec y ) \rp]
\rp| \leq \gamma \, ,
\end{align}
for any polynomial $p$ satisfying $\E_{ \vec x \sim \normal(\vec 0, \vec I) } \lp[ p^2(\vec x) \rp] \leq \eps^{-n}$.
We first prove a structural claim.
\begin{claim}
\label{clm:poly-lip-bound}
Let $m \leq n$, and $q: \R^m \mapsto \R$ be an arbitrary degree-$k$
polynomial such that
$ \E_{ \vec x \sim \normal(\vec 0, \vec I) }
\lp[ q^2( \vec x )\rp] \leq \eps^{-n}.
$
Then it holds that
\begin{align*}
\lp|  q( \delta \vec 1_m ) - q(\vec 0_m) \rp| \leq \gamma.
\end{align*}
\end{claim}
\begin{proof}
Assume that $q(\vec x)$ admits the Hermite decomposition
$ q(\vec x) = \sum_{ \vec a \in \mathbbm N^n: |\vec a| \leq k } c_{\vec a} h_{\vec a}(\vec x)$ (see \Cref{app:Hermite_polynomials} for definitions and notation regarding Hermite polynomials).
By our assumption that $\E_{ \vec x \sim \normal(\vec 0, \vec I) }
\lp[ q^2( \vec x )\rp] \leq \eps^{-n}$, the coefficients in the Hermite decomposition should satisfy that $ \sum_{ \vec a \in \mathbbm N^n: |\vec a| \leq k } c_{\vec a}^2 \leq \eps^{-n}$.
Using \Cref{clm:Hermite-lip} and the assumption that $\delta$ is a sufficiently small polynomial in $k$, we have that 
$h_{\vec a}( \delta \vec 1_m )- h_{\vec a}(\vec 0_m)\leq \sqrt{ k^3 m }  \delta$ for all $\vec a \in \mathbbm N^n: |\vec a| \leq k$.
It follows that
\begin{align*}
    |q( \delta \vec 1_m ) - p( \delta \vec 0_m )|
    &\leq \sum_{\vec a \in \mathbbm N^n: |\vec a| \leq k } |c_{\vec a}| \left| h_{\vec a}( \delta \vec 1_m )- h_{\vec a}(\vec 0_m) \right| \\
    &\leq 
    \sum_{\vec a \in \mathbbm N^n: |\vec a| \leq k } |c_{\vec a}| \sqrt{k^3 m } \; \delta  \\
    &\leq \sqrt{k^n \; \eps^{-n} \; k^3 m } \; \delta 
    \leq \gamma
\end{align*}
where the first inequality is by the triangle inequality,
the second inequality is by the bound on $h_{\vec a}(\delta \vec 1_m) - h_{\vec a}(\vec 0_m)$, the third inequality is by Cauchy's inequality and the bound 
$\sum_{ \vec a \in \mathbbm N^n: |\vec a| \leq k } c_{\vec a}^2 \leq \eps^{-n}$, and the last inequality is by our assumption on $k,n,\delta,\gamma,\eps$.
This concludes the proof of \Cref{clm:poly-lip-bound}.
\end{proof}

Given two vectors $\vec y \in \R^m$, $\vec z \in \R^{n - m}$, and a subset of indices $S \subseteq [n]$ with cardinality $|S| = m$, 
 we define $ \vec x:= \vec y \cup_S \vec z $ as the vector that has $\vec x_i = \vec y_i$ if $i \in S$ and $\vec x_i = \vec z_i$ otherwise.
Given an arbitrary subset of indices $S \subseteq [n]$, we can then define the function
$q_S: \R^{n - |S|} \mapsto \R$ as
$$
q_S( \vec y ) :=
\E_{ \vec z \sim \normal(\vec 0, \vec I_{|S|}) }
\lp[ p \lp(  \vec z \cup_S \vec y  \rp) \rp].
$$
It is not hard to see that $q_S$ is a degree at most $k$  polynomial in $\vec y$ 
satisfying 
$\E_{ \vec y \sim \normal(\vec 0,  \vec I_{ n-|S| }) }
\lp[ 
q^2_S(\vec y)
\rp] = 
\E_{ \vec x \sim \normal(\vec 0, \vec I_{ n }) }
\lp[ 
p^2(\vec x)
\rp] $, which is at most $\eps^{-n}$.

Note that we can decompose the difference in \Cref{eq:mean-bound} by conditioning on different subsets of samples that are sampled from the Gaussian distribution. 
In particular, we have that
\begin{align*}
\lp|
\E_{ \vec x  \sim \lp( D_{\emptyset} \rp)^{\otimes n}  }
\lp[  p(\vec x) \rp]
- 
\E_{\vec y \sim \lp( D_{\alt}\rp)^{\otimes n} }
\lp[  p( \vec y ) \rp]
\rp| 
&\leq
\max_{ S \subset [n] }
\lp| 
\E_{ \vec z \sim \normal(\vec 0, \vec I_{|S|}) }
\lp[
p\lp( \vec z \cup_S \lp( \delta \vec 1_{n - |S|} \rp)   \rp)
- 
p\lp( \vec z \cup_S \vec 0_{n - |S|}  \rp)
\rp]
\rp|\\
&=
\max_{ S \subset [n] }
\lp| 
q_S( \delta \vec 1_{n - |S|} ) - q_S( \vec 0_{n - |S|} )
\rp| \leq \gamma \, ,
\end{align*}
where the equality is by the definition of $q_S$, and the last inequality follows from \Cref{clm:poly-lip-bound}.
This shows \Cref{eq:mean-bound}, and concludes the proof of 
\Cref{thm:separation}.
\end{proof}

\section{(Near-)Optimality of the Sample Lower Bound in \Cref{thm:main}}

\begin{theorem}
\label{thm:optimality}
For any $d,m \in \mathbbm Z_+$, 
there exists a univariate distribution $A$ on $\R$ that matches $m$ moments with $\normal(0, 1)$ so that for 
$n \gg  \lp( C \; d \; m \; \log d\rp)^{m/4}$, 
where $C$ is a sufficiently large constant,
and $k > 4  \log n $, 
there exists a degree-$k$ polynomial $p: \R^{n \times d} \mapsto \{0, 1\}$  that successfully distinguishes $\normal(\vec 0, \vec I)$ and $\ngcav$ 
for any unit vector $\vec v \in \R^d$
$\sgn(p(\cdot))$ with constant probability:
$
\E_{ \vec x^{(1)}, \cdots, \vec x^{(n)} \sim \normal(0, 1) }\lp[ \sgn \lp( p \lp(  \vec x^{1:n} \rp) \rp)  \rp] < 1/10$ but $\E_{\vec y^{(1)}, \cdots, \vec y^{(n)} \sim \ngcav }\lp[ 
\sgn\lp( p \lp(  \vec y^{1:n} \rp) \rp)   \rp] > 9/10$ for all unit vector $\vec v \in \R^d$.
\end{theorem}

The basic idea is to construct a distribution 
$A$ that takes values slightly larger than $d^{1/4}$ with probability almost $d^{-m/4}$ but has its first $m$ moments matched with $\normal(0, 1)$. 
If so, any sample $\vec y$ from $\ngcav$ that witnesses the extreme values of $A$ will have an unusually large norm compared to the Gaussian case.
However, since the distribution of $ \| \vec x \|^2_2$ over the standard Gaussian has its variance being approximately $d^{1/2}$, this leads to a detectable discrepancy. 
This forms the basis of our algorithm for distinguishing $\normal(\vec 0, \vec I)$ and $\ngcav$: the algorithm 
draws approximately $d^{m/4}$ many samples and simply looks for a sample with an abnormally large $\ell_2$ norm. 
After that, we show that it is not difficult to turn this algorithm into an actual PTF test with a small loss in efficiency.

To begin with, the following proposition constructs a moment-matching distribution $A$ that has a non-trivial amount of mass on some extreme value $R$.
\begin{proposition}
\label{prop:A-construction}
For any positive integer $m$, and real number $R > 0$,
there exists a distribution $A$ on $\R$ that matches $m$ moments with the standard Gaussian and satisfies that
$\Pr[A = R] \geq R^{-m} / \poly(m).$ 
\end{proposition}
\begin{proof}
Let $\eps$ be $R^{-m}$ divided by a sufficiently large polynomial in $m$. 
We define $A$ to be a probability distribution with the following probability density function:
\begin{align*}
A(x)dx = G(x) dx + \lp( p(x) \mathbbm 1\{|x| < 1 \} dx \rp) + \eps \; \delta_{x=R}   
\end{align*}
where $G: \R \mapsto \R_+$ is the probability density function of the standard Gaussian, $\delta_{x=R}$ is a point mass at $x=R$, and $p(x)$ is some degree-$m$ polynomial we will specify later. 
It is clear that $\Pr_{x \sim A}\lp[ x=R \rp] = \eps$. 
We just need to show that there is a polynomial $p$ 
which ensures that (1) $A(x)$ is non-negative, (2) $\int_{\R} A(x) dx = 1$, and (3) $A$ match enough moments with the standard Gaussian $\normal(0, 1)$:
$$
\int_{\R} G(x) x^t dt = \int_\R A(x) x^t dt
$$
for all integers $0 \leq t \leq m$. 
In particular, we need some polynomial $p$ such that
\begin{align*}
 \int_{-1}^1 p(x) x^t dt = \eps R^t
\end{align*}
for all such $0 \leq t \leq m$,  and $p(x)+G(x) \geq 0$ for all $|x| \leq 1$.  
Such a polynomial can be constructed with standard techniques based on linear programming (see e.g. exercise 8.3 of \cite{diakonikolas2023algorithmic}).
This concludes the proof of \Cref{prop:A-construction}.
\end{proof}

We are now ready to prove \Cref{thm:optimality}.
\begin{proof}[Proof of \Cref{thm:optimality}]
Let $A$ be the distribution given by \Cref{prop:A-construction} with $R = C d^{1/4}$, where $C$ is some sufficiently large constant multiple of $\log^{1/4}(n)$.  
Let $t = \floor{ \log n }$. 
Define the polynomial $p: \R^{n \times d} \mapsto \R$ as
\begin{align}
p( \vec x^{(1)}, \cdots, \vec x^{(n)} ) = 
\sum_{i=1}^n \left[ \lp \| \vec x^{(i)} \rp \|_2^2 - d \right]^{2t}.    
\end{align}
If the $\vec x^{(i)}$'s are independent Gaussians, we have that $\lp \| \vec x^{(i)} \rp \|_2^2 - 1$ is a degree-$2$ polynomial with $L^2$ norm $O(\sqrt{d})$. 
Therefore, by Gaussian hypercontractivity, we have that $\E\lp[   \lp( \lp \| \vec x^{(i)} \rp \|_2^2 -1 \rp)^{2t}\rp] \leq O(t)^t d^t$. 
In particular, this implies that
$$
\E[  p\lp(\vec x^{(1:n)} \rp)] \leq O(t)^t d^t n.
$$
Note that $p$ is a non-negative polynomial.
Applying Markov's inequality therefore gives that
\begin{align*}
\Pr \lp[ 
p\lp(\vec x^{(1:n)} \rp) <  (C't)^t d^t
\rp]  \geq 9/10 \, ,
\end{align*}
where $C'$ is some sufficiently large constant.

On the other hand, suppose that $\vec x$ is drawn from $\ngcav$ conditioned on $\vec v^\top \vec x = R$ (which happens with probability at least $R^{-m}/\poly(m))$. 
Then we immediately have that 
$\| \vec x \|_2^2 = R^2 + \lp|\vec x^\perp \rp|_2^2$, where $\vec x^\perp$ is the part of $\vec x$ that is orthogonal to $\vec v$, which is distributed like $\normal(\vec 0, \vec I - \vec v \vec v^\top)$.
Therefore,
$\E\lp[  \lp \| \vec x \rp \|_2^2 - d \bigg |
\vec v^\top \vec x = R  \rp] = R^2-1$
and
$$ \Var\lp[  \lp \| \vec x \rp \|_2^2 - d \bigg | 
\vec v^\top \vec x = R \rp] 
= \Var\lp[ \lp \|\vec x^\perp \rp \|_2^2 \rp] 
= O(d).
$$
Therefore, by Chebyshev's inequality,
we have
\begin{align*}
\Pr
\lp[ \lp \| \vec x \rp \|_2^2 - d > R^2/2  \bigg | \vec v^\top \vec x = R
\rp] > 1/2 \, ,
\end{align*}
which further implies that
$$
\Pr \lp[ \lp \| \vec x \rp \|_2^2 - d > R^2/2 \rp] > (1/2)R^{-m}/\poly(m) > 10/n.
$$
Thus, if $\vec x^{(1)}, \cdots, \vec x^{(n)}$ are drawn independently from $\ngcav$, the probability that there exists some $i \in [n]$ such that 
$\lp \| \vec x^{(i)} \rp \|_2^2 -d > R^2/2$ is at least $9/10$. However, if this happens,
we immediately have that
$$
p\lp( \vec x^{(1)},\cdots, \vec x^{(n)} \rp) \geq 
\lp(  \lp \| \vec x^{(i)} \rp \|_2^2 -d \rp)^{2t} 
\geq R^{4t}/2^{2t} = C^{4t} d^t / 2^{2t} \geq T.
$$
This shows a separation between the two cases, 
and completes the proof of \Cref{thm:optimality}.    
\end{proof}

\section{Comparison with Information-Computation Gaps for NGCA from Prior Work}

\subsection{Bound on Low-Degree Likelihood Ratio}\label{sec:prior-work-ldlr}

The following result follows by combining \cite{BBHLS21}, (which shows that a lower bound on the statistical query dimension of a hypothesis testing problem implies an upper bound on the norm of the low-degree likelihood ratio) and the statistical query dimension bound from \cite{DKS17-sq}. The details of this combination can be found in \cite{diakonikolas2021statistical} (Corollary 6.4, treating $y$ as a fixed value).

\paragraph{Additional notation} For a distribution $D$ over $\mathcal{X}$, we use $D^{\otimes n}$ to denote the joint distribution of $n$ i.i.d.\ samples from $D$.
For two functions $f:\mathcal{X} \to \R$, $g: \mathcal{X} \to \R$ and a distribution $D$, we use $\langle f, g\rangle_{D}$ to denote the inner product $\E_{X \sim D}[f(X)g(X)]$.
We use $ \|f\|_{D}$ to denote $\sqrt{\langle f, f \rangle_{D} }$.
We say that a polynomial $f:\R^{n \times d} \to \R$ has sample-wise degree $(r,\ell )$ if each monomial uses at most $\ell$ different samples from $\multix$ and uses degree at most $r$ for each of them.
Let $\mathcal{C}_{r,\ell}$ be the linear space of all polynomials of sample-wise degree $(r,\ell)$ with respect to the inner product defined above.
For a function $f:\R^{n \times d} \to \R$, we use $f^{\leq r, \ell}$ to be the orthogonal projection onto $\mathcal{C}_{r,\ell}$   with respect to the inner product $\langle \cdot , \cdot \rangle_{\normal(\vec 0,\vec I)^{\otimes n}}$.  Finally, for the null distribution $\normal(\vec 0,\vec I)$ and a distribution $\ngca$ from \Cref{def:ngca_problem}, define the likelihood ratio $\overline{\mathcal{M}}_{A,\vec v}^{\otimes n}(\multix)$ to be the ratio of the pdf of $\overline{\mathcal{M}}_{A,\vec v}^{\otimes n}$ on the point $(\multix)$ divided by the pdf of $\normal(\vec 0,\vec I)$ evaluated on the point $(\multix)$. The $\chi^2$-distance between two distributions $D$ and $R$ on $\mathcal{X}$ is defined as  $\chi^2(D,R) \eqdef \int_{x \in \mathcal{X}}D^2(x)/R(x) dx - 1$.

        \begin{theorem}\label{cor:low-deg-hardness-general-problem} 
        For any $c \in (0,1/2)$ the following holds.
        There exists a subset $S$ of the $d$-dimensional unit sphere for which the following hold.
        Let a sufficiently small positive constant $c$. Let $\ngca$ denote the distribution from \Cref{def:ngca_problem} and assume that 
        $A$ matches the first $m$  moments with $\normal(\vec 0, \vec I)$ and the vector $\vec v$ is drawn from the uniform distribution over $S$. For any $d \in \Z_+$ with $d = m^{\Omega(1/c)}$, any $n \leq \Omega(d)^{(m+1)(1/2-c)}/\chi^2(A,\normal(\vec 0,\vec I))$ and any even integer $\ell < d^{c/4}$, we have that
        \begin{align}\label{eq:conclusion_ld}
            \left\| \E_{v \sim \mathcal{U}(S)}\left[ \left( \overline{\mathcal{M}}_{A,\vec v}^{\otimes n}  \right)^{\leq \infty, \ell } \right]  - 1  \right\|_{\normal(\vec 0,\vec I)^{\otimes n}} \leq 1\;.
        \end{align}
    \end{theorem}

The quantity in the right hand side of \Cref{eq:conclusion_ld} is the norm of the low-degree likelihood ratio. This is an equivalent rewriting of the best possible advantage $\beta$ in \Cref{def:good-ld} with $D_{\emptyset} = \normal(\vec 0,\vec I)$ and $D_{\alt} = \tfrac{1}{|S|}\sum_{\vec v \in S} \ngca$. In particular the variant of that definition where the right hand side in \Cref{eq:rhs} only scales with the standard deviation under the null distribution $\normal(\vec 0,\vec I)$ instead of the maximum.
Comparing with \Cref{thm:main}, the three conditions of \Cref{thm:main} appear in some form in \Cref{cor:low-deg-hardness-general-problem}: The sample complexity condition now is $n \leq \Omega(d)^{(m+1)/(1/2-c)}/\chi^2(A,\normal(\vec 0,\vec I))$, which has a better constant $1/2$ in the exponent. 
Surprisingly, we show in \Cref{thm:optimality} that the exponent from \Cref{thm:main} is 
essentially best possible, implying that PTF tests are inherently slightly more powerful than LDPs even for the NGCA problem.
 The condition $k < d^{\Omega(1)}$ of \Cref{thm:main} corresponds to the part of the statement in \Cref{cor:low-deg-hardness-general-problem} restring $\ell < d^{c/4}$. Since the low-degree likelihood ratio uses sample-wise degree $(\infty,\ell)$ this means that the total degree of the resulting polynomial is restricted to be  at most $\ell < d^{c/4}$. Finally, the condition $d > m^{\Omega(1)}$ appears in both \Cref{thm:main} and \Cref{cor:low-deg-hardness-general-problem} with different constants.

\subsection{Hardness in the Statistical Query Model}
\label{sec:sq-prior}
In this section we restate the result from \cite{DKS17-sq} regarding the information-computation gap for NGCA within the Statistical Query model.
Before we restate the theorem, we recall the basics of the SQ model~\cite{Kearns:98, FGR+13}. 
Instead of drawing samples from the input distribution, 
SQ algorithms  are only permitted query access to the distribution via the following oracle:
\begin{definition}[STAT Oracle] \label{def:stat}
Let $D$ be a distribution in $\R^d$. A statistical query is a bounded function $f: \R^d \to [-1,1]$. 
For $\tau>0$, the $\mathrm{STAT}(\tau)$ oracle responds to the query $f$ 
with a value $v$ such that $|v - \E_{\bx \sim D}[f(\bx)] \leq \tau$. 
We call $\tau$ the tolerance of the statistical query.
\end{definition}

An information-computation gap in this model for a learning problem $\Pi$ is typically of the following form: 
any SQ algorithm for $\Pi$ must either make a large number of queries $q$ 
or at least one query with small tolerance $\tau$. 
When simulating a statistical query in the standard PAC model 
(by averaging i.i.d.\ samples to approximate expectations), 
the number of samples needed for a $\tau$-accurate query 
can be as high as $\Omega(1/\tau)$. Thus, we can intuitively interpret 
an SQ lower bound as a tradeoff between runtime of $\Omega(q)$ 
or a sample complexity of $\Omega(1/\tau^{2})$.

The statement for NGCA is the following.

\begin{theorem}[See, e.g., Proposition 8.14 in \cite{diakonikolas2023algorithmic}]\label{thm:SQ}
    For any constant $c \in (0,1/2)$ and any $m,d \in Z_+$ with $d \geq ((m+1)\log d)^{2/c}$ the following hold. If $A$ is a distribution on $\R$ that matches the first $m$ moments with $\normal(0,1)$ and $\ngca$ denotes the distribution from \Cref{def:hidden_distr}, then any SQ algorithm for distinguishing between $\ngca$ and $\normal(\vec 0,\vec I)$ (when $\vec v$ is unknown to the algorithm) requires either $2^{\Omega(d^c)}$ many SQ queries or at least one query to $\mathrm{STAT}$ with accuracy $\tau \leq 2 d^{-(m+1)(1/4-c/2)}\sqrt{\chi^2(A,\normal(0,1)}$.
\end{theorem}

The interpretation of the above is a trade-off between exponential runtime and sample complexity at least $d^{(m+1)(1/2-c)}/\chi^2(A,\normal(0,1)$. As shown in \cite{diakonikolas2023sq}, the dependence on $\chi^2(A, \mathcal{N}(0,1))$ was merely an artifact of the original analysis and can be removed, at the cost of a larger constant in the exponent in the sample complexity. As one can see \Cref{thm:SQ} uses the same three assumptions as \Cref{thm:main} up to differences in the constant and polylog factors. In the particular, the assumption $d \geq ((m+1)\log d)^{2/c}$ in \Cref{thm:SQ} ensures that $2^{d^{c/2}} \geq d^{(m+1)(c-1/2)}$, i.e., both the runtime and the sample complexity are at least $d^{(m+1)(c-1/2)}$.

\section{Applications to Learning Mixture Models and Robust Statistics}
\label{sec:concrete_apps}

In this section, we show that \Cref{thm:main} implies strong lower bounds against PTF tests for a range of problems in machine learning theory and robust statistics.

Note that since PTF produces binary output, all the lower bounds will be for the testing version of the corresponding statistical estimation problems.

Similar to the approach taken in \cite{diakonikolas2024sum}, we first define two meta testing problems that can be instantiated to model the testing version of various statistical estimation problems under the total-variation corruption model and the Hubert Contamination model.

\begin{definition}[TV-Corruption Model]\label{def:tv_corruption}
Let $\mathcal{D}$ be a set of
distributions. We define $\cB_{TV}(\tau, \cN(\vec 0, \Id_d), \cD)$ to be the following hypothesis testing problem: Given $n$ i.i.d.\ samples $ \{ \vec x^{(1)}, \cdots, \vec x^{(n)}  \} \subseteq \R^d$ drawn from one of the following two distributions, the goal is to determine which one generated the samples: (a) $\cN(\vec 0, \Id_d)$; and (b) $D'$ such that $d_{TV}(D', D) \leq \tau$ for $D$ drawn uniformly at random from $\cD$.
\end{definition}

\begin{definition}[Huber Contamination Model]\label{def:huber_contamination}
Let $\mathcal{D}$ be a family of distributions. We define $\cB_{\huber}(\tau, \cN(\vec 0, \Id_d), \cD)$ to be the following hypothesis testing problem: Given $n$ i.i.d.\ samples $ \{ \vec x^{(1)}, \cdots, \vec x^{(n)}  \} \subseteq \R^d$ drawn from one of the following two distributions, determine which one generated the samples: (a) $\cN(\vec 0, \Id_d)$; (b) $D'$, which is $(1-\tau) D + \tau B$, where $D$ is drawn uniformly at random from $\cD$ and $B$ is an arbitrary distribution possibly dependent on $D$. 
\end{definition}

We start with the testing version of robust mean estimation of isotropic Gaussian distribution.
\begin{problem}[Hypothesis-Testing-Robust-Mean-Estimation with Identity Covariance]
\label{prob:RME_id}
Let $\tau > 0$ and $B = O(\log^{1/2}(1/\tau))$ be a parameter. The problem is $\cB_{TV}(\tau, \cN(0, \Id_d), \cD)$ (\Cref{def:tv_corruption}), where every $D \in \cD$ is of the form $\cN(\bmu_D, \Id_d)$ and $\norm{\bmu_D}\geq \Omega(\tau \log(1/\tau)^{1/2}) / B^2)$. 
\end{problem}
It is shown in \cite{DKS17-sq} that the above testing problem can be reduced to NGCA whose  non-Gaussian component $A$ matches $m = B$ many moments. 
The lower bound against PTF tests then follows.
\begin{corollary}
Let $h: \R^{n \times d} \mapsto \R$ be a degree-$k$ PTF. If $h$ solves the hypothesis testing problem in \Cref{prob:RME_id}, then we must either have
$n \geq d^{ B(1-c^*)/4}$ or $k \geq  d^{ \Omega(1) }$, where $c^*$ is some small constant.
\end{corollary}
The second application is on robust mean estimation of distributions with bounded $m$-moments.
The testing version of the problem is as follows (see Section 6 from \cite{diakonikolas2022robust} for the justification).
\begin{problem}[Hypothesis-Testing-Robust-Mean-Estimation with Bounded $m$-th Moments]
\label{prob:mean-bdd-mom}
Let $m$ be a positive integer and $\tau\in(0, 1)$. Hypothesis-Testing-RME-Bounded-$m$-Moments is the problem $\cB_{\huber}(\tau, \cN(\vec 0, \Id_d), \cD)$ (\Cref{def:huber_contamination}) where each $D \in \cD$ satisfies the following: (i) the mean vector $\bmu$ satisfies $\| \bmu \| \geq \Omega(\frac{1}{m} \tau^{1-1/m})$; (ii) $D$ has subgaussian tails of order $m$, i.e., for all $\vec v \in \R^d$ and $1 \leq i \leq m$, $\E_{\vec x \sim D}[\abs{\vec v^\top (\vec x - \bmu)}^i]^{1/i} \leq O(\sqrt{i})$. 
\end{problem}
It is shown in \cite{diakonikolas2022robust} that the problem can be reduced to NGCA whose non-Gaussian component $A$ matches $m$ many moments with the standard Gaussian.
We therefore obtain the following lower bound against PTF tests.
\begin{corollary}
Let $h: \R^{n \times d} \mapsto \R$ be a degree-$k$ PTF. If $h$ solves the hypothesis testing problem in \Cref{prob:mean-bdd-mom}, then we must either have
$n \geq d^{ m(1 - c^*) / 4}$ or $k \geq  d^{ \Omega(1) }$, where $c^*$ is some small constant.
\end{corollary}

The third application is on Gaussian list-decodable mean estimation. The testing version is given below.
See \cite{diakonikolas2018list} for a thorough walkthrough of this problem, and the reduction between its learning version and the testing version.
\begin{problem}[Hypothesis-Testing-List-Decodable-Mean-Estimation]
\label{prob:LDME}
Given $\tau\in(0,\frac{1}{2})$ and positive integer $m\geq 2$, the hypothesis-testing-LDME is the problem $\cB_\huber(1-\tau, \cN(\vec 0, \Id_d), \cD)$ (\Cref{def:huber_contamination}), where every $D \in \cD$ has the form $\cN(\bmu_D, \Id_d)$ for some $\bmu_D\in\R^d$ whose $\ell_2$-norm is at least $\Omega((m\tau)^{-1/m})$. 
\end{problem}
It is shown in \cite{diakonikolas2018list} that this problem can be reduced to NGCA whose non-Gaussian component matches $m$ moments with the standard Gaussian. We hence obtain the following lower bound against PTF tests.
\begin{corollary}
Let $h: \R^{n \times d} \mapsto \R$ be a degree-$k$ PTF. If $h$ solves the hypothesis testing problem in \Cref{prob:LDME}, then we must either have
$n \geq d^{ m(1 - c^*) / 4}$ or $k \geq  d^{ \Omega(1) }$, where $c^*$ is some small constant.   
\end{corollary}
The last application is on learning mixtures of $k$ Gaussians.
It is shown in \cite{DKS17-sq} that this learning problem can be reduced from the following testing problem.
\begin{problem}[Hypothesis-Testing-$m$-GMM]\label{prob:gmm_learning}
Let $0<\gamma <1$. Hypothesis-Testing-$m$-GMM is the problem $\cB_{\huber}(0, \cN(0, \Id_d), \cD)$ (\Cref{def:huber_contamination}), where every $D \in \cD$ is a mixture of $m$ Gaussians such that each pair of the Gaussians are $1-\gamma$ apart in total variation and $d_{TV}(D, \cN(\vec 0, \Id_d)) \geq \frac{1}{2}$. 
\end{problem}
It has been shown in the same work that the testing problem can be further reduced from NGCA whose non-Gaussian component matches $2m-1$ moments with $\normal(0, 1)$.
We therefore obtain the following lower bound.
\begin{corollary}
Let $h: \R^{n \times d} \mapsto \R$ be a degree-$k$ PTF. If $h$ solves the hypothesis testing problem in \Cref{prob:gmm_learning}, then we must either have
$n \geq d^{ m(1 - c^*) / 2}$ or $k \geq  d^{ \Omega(1) }$, where $c^*$ is some small constant.   
\end{corollary}

\end{document}